\documentclass[12pt]{article}

\usepackage{setspace}

\usepackage{natbib}
\usepackage[nottoc,notlof,notlot]{tocbibind}

\bibpunct{(}{)}{;}{a}{}{;}
\usepackage{breakcites}
\usepackage{graphicx,xcolor}
\graphicspath{ {images/} }
\usepackage[algoruled]{algorithm2e}
\usepackage[pagebackref=false,breaklinks=true]{hyperref} %
\hypersetup{
    colorlinks=true,
    citecolor=blue,
    filecolor=blue,
    linkcolor=blue,
    urlcolor=blue,
    bookmarksopen=true,
    pdfstartview=FitH
}
\usepackage{caption,threeparttable}

\usepackage{amsthm,amsmath}
\usepackage{bbm}
\usepackage{amsfonts}%
\usepackage{amssymb}%

\usepackage{geometry} %
\usepackage{graphicx} %
\usepackage{framed}
\usepackage{caption}
\usepackage{subcaption}

\usepackage{booktabs} %
\usepackage{array} %
\usepackage{paralist} %
\usepackage{verbatim} %

\usepackage{lipsum}

\usepackage[singlelinecheck=false]{caption}
\usepackage{tabularx}
\usepackage{tikz}

\usepackage{multicol}
\usepackage{multirow}

\newtheorem{theorem}{Theorem}

\newtheorem{assump*}{Assumption}[section]
\newtheorem{lemma}{Lemma}

\newtheorem{proposition}{Proposition}

\theoremstyle{definition}
\newtheorem{definition}{Definition}

\newtheorem{remark}{Remark}

\newtheorem{assumption}{Assumption}

\def\equationautorefname~#1\null{(#1)\null} %

\usepackage{diagbox}
\usepackage{color-edits} %
\addauthor{a}{red}    %
\addauthor{c}{blue} %
\addauthor{b}{purple} %

\title{\LARGE Does p-Hacking Mitigate or Exacerbate the Effects of \\
Publication Bias?\thanks{We have benefited greatly from conversation with Kaspar Wuthrich. We also thank Jack Porter and Kohei Yata for helpful discussion, and the audience at the University of Iowa for useful comments. All errors are our own.}}

\author{ \parbox{0.45\linewidth}{\centering
\normalsize \setstretch{1.2}
Yong Cai \\
Department of Economics \\
University of Wisconsin-Madison \\
\url{yong.cai@wisc.edu}}
\parbox{0.45\linewidth}{\centering
\normalsize \setstretch{1.2}
Agathe Pernoud \\
Booth School of Business\\
University of Chicago\\
\url{agathe.pernoud@chicagobooth.edu}}\\ 
\\
\parbox{0.45\linewidth}{\centering
\normalsize \setstretch{1.2}
Boli Xu \\
Department of Economics \\
University of Iowa\\
\url{boli-xu@uiowa.edu} \hfill }
}

 \date{\parbox{\linewidth}{\centering%
 \hfill \\
  \today \endgraf}}

\begin{document}
\maketitle

\begin{abstract} \setstretch{1}
\noindent This paper studies the effects of p-hacking on the bias of published estimates when papers with statistically significant results are selectively published. We show that fast p-hacking---actions that lead to large changes in p-values---always exacerbates the bias from selective publication. On the other hand, slow p-hacking---actions that lead to small changes in p-values---exacerbates bias when selection is weak, but mitigates it when selection is strong. %
In a model featuring both types of p-hacking, we show that a normality assumption identifies the true distribution of effects as well as the counterfactual mean that would obtain under selective publication without p-hacking. Applying the model to meta-analyses on the effects of behavioral nudges and development aid, we find suggestive evidence that both mitigation and exacerbation can arise in practice.
\looseness=-1
\end{abstract}

\newpage 
\section{Introduction}

p-Hacking is widely considered to be a threat to the credibility of published research. Coined by \cite{simonsohn2014p}, the term refers to the collection of actions that empirical researchers take to increase the probability that their papers are published. A substantial body of work documents selective publication or publication bias \citep{sterling1959publication, rosenthal1979file}, in which papers that are more statistically significant are preferentially published (in economics, see, e.g., \citealt{brodeur2016star, christensen2018transparency, brodeur2020methods}). As such, the practice of p-hacking often entails increasing the statistical significance of empirical results, distorting reported results from the truth. 

\cite{simonsohn2020fast} further distinguishes between \emph{fast} and \emph{slow} p-hacking. Fast p-hacking refers to processes that change p-values by a substantial amount from analysis to analysis. It is often associated with actions that lead to new draws of essentially independent data. A leading example is file-drawering, in which a researcher with a pre-registered analysis plan discards insignificant results and repeats an experiment wholesale. In contrast, slow p-hacking refers to processes that change p-values by a small amount from analysis to analysis. It typically arises because to analyze messy observational data, researchers have to make many small decisions, most of which have little impact on the final results. For example, researchers may have to choose between different ways of transforming a variable, discretization points of continuous variables, or thresholds for defining outliers. Changing these decisions, therefore, offers a way for researchers to make small adjustments to their p-values. Less innocuously, a researcher may also falsify a small fraction of their data.

p-Hacking, in either manifestation, has been linked to exaggerated effect sizes and excess false positives in published research. It is therefore considered an important contributor to the recent ``replication crisis", which began in psychology \citep{open2015estimating} and biomedical research \citep{prinz2011believe,begley2012raise} but has since been documented across fields, including in economics (see, e.g., \citealt{ioannidis2013s,baker20161, camerer2016evaluating, camerer2018evaluating, christensen2018transparency}). Given the potentially severe repercussions, it is important to understand the effects of p-hacking on the accuracy of published research. While a large and growing body of work tackles this issue, existing papers primarily take either an empirical or simulation-based approach to quantifying the distortions arising from p-hacking (see, e.g., \citealt{leamer1983let, ioannidis2005most,simmons2011false, masicampo2012peculiar,head2015extent,brodeur2016star,brodeur2020methods,friese2020p,moss2023modelling}).

To better understand the settings to which prevailing results apply, this paper formally studies the effects of p-hacking in a simple model with selective publication. We take the form of selective publication as given and assume that it is known to researchers. Specifically, publication probability is assumed to be a step function that jumps at some thresholds of significance. We say that selection is \emph{two-sided} when positive and negative significant results are preferentially published with the same probability. We say that selection is \emph{one-sided} when positive significant results are selectively published, but insignificant and negative significant results are published with the same probability. We focus on two- and one-sided selection for our theoretical analyses. 

Selective publication results in biased estimates. It is well known, for example, that two-sided selective publication leads to an amplification bias: the published estimate has, in expectation, the same sign as the study-specific effect but a larger magnitude \citep{lane1978estimating,hedges1984estimation,iyengar1988selection}. For one-sided selective publication, the bias has the same sign as the direction of selection.

Building on these results, we evaluate the effects of p-hacking relative to a benchmark in which selective publication operates alone. This is motivated by the view that systematic p-hacking does not occur in the absence of selective publication. In our stylized model, researchers are endowed with an estimate and standard error. They can then increase the probability that their papers are published by engaging in p-hacking.
Researchers first decide whether or not to conduct fast p-hacking. This is modeled as the option to pay a fixed cost to redraw estimates from the same distribution as the original estimate.  Regardless of whether or not they draw new estimates, researchers then decide whether or not to conduct slow p-hacking. This is the option to pay a cost to report a manipulated estimate and/or standard error. The cost of slow p-hacking is convex and increasing in the amount of manipulation that researchers choose. 

Under two-sided selection, Theorem \ref{thm:slow} shows that slow p-hacking exacerbates the amplification bias when selection is sufficiently weak. On the other hand, slow p-hacking mitigates the amplification bias when selection is sufficiently strong, such that the expected published estimate is sandwiched between the true effect and the selection-only benchmark. Meanwhile, Theorem~\ref{thm:fast} shows that fast p-hacking always exacerbates amplification bias, leading to more distorted estimates. Theorem~\ref{thm:combined} studies a generalized model with both types of p-hacking, finding---as in the slow p-hacking only case---that both mitigation and exacerbation can arise, depending on whether selection is strong or weak. 
Similar results hold in the model with one-sided selection. 

The effect of p-hacking depends on the type of p-hacking that is undertaken, as well as the selection environment that produced it. Although our results are proven in a stylized setting, we argue that both mitigation and exacerbation can arise in environments that are relevant for empiricists.
To that end, we show that under a normality assumption, the mean study effect, as well as the counterfactual mean of published estimates under selective publication only, are both identified. This allows us to assess the extent to which mitigation or exacerbation may have occurred in a given literature. Under additional parametric assumptions, we estimate a censored version of our one-sided model using data from two meta-analyses.  
The first is \cite{mertens2022effectiveness}, which studies the effect of behavioral nudges on promoting desired choices. Here, we find that p-hacking may have mitigated the bias from selective publication by as much as 54\%. The second is \cite{doucouliagos2011ineffectiveness}, which concerns the effect of development aid on economic growth. Here, we find instead that p-hacking has exacerbated bias by 27\%. 

Discussions about p-hacking and selective publication often treat the two phenomenon separately. Our results show that it is important to understand p-hacking in the context of the selection environment that produced it. That p-hacking could mitigate the effects of publication bias under extreme selection also hints at a self-correcting tendency in the scientific process. This echoes the idea that literatures as a whole may be more reliable than their constituent papers, a view in philosophy of science that goes at least as far back as \cite{kuhn1962structure}. However, this does not imply that p-hacking is desirable, since our paper is silent on the other negative effects of p-hacking, such as inflating the number of false positives,
the erosion of public trust in science, or distortions in the allocation of research funds, among others (see, e.g., \citealt{ioannidis2005most,button2013power,Wasserstein02042016}).

The remainder of this paper proceeds as follows. Section \ref{section--setting} describes the formal setting. Sections \ref{section:twosided} and \ref{section:onesided} contain our main theoretical results under two- and one-sided publication bias, respectively. Section \ref{sec:emp} discusses the identification of our proposed model and presents empirical results based on \cite{mertens2022effectiveness} and \cite{doucouliagos2011ineffectiveness}. Section \ref{section--conclusion} concludes. 

\subsection{Related Literature}

This paper connects to the vast and growing literature on p-hacking. First, it speaks specifically to the line of work concerned with the effects of p-hacking on published estimates \citep{leamer1983let,ioannidis2005most, simmons2011false, gelman2013garden, brodeur2016star,ioannidis2017power,friese2020p,moss2023modelling,stefan2023big}. These papers typically provide quantitative evidence either empirically or through simulations. Complementing their approach, we develop a \textit{theoretical} framework to formally study how fast and slow p-hacking affects the bias of published estimates.\footnote{A contemporaneous paper, \cite{keane2026instrument}, also theoretically examines the effect of p-hacking on the bias of published estimates, focusing on the case where researchers p-hack instrumental-variable regressions.} We show that p-hacking can mitigate the effects of selective publication. This result contrasts with the prevailing views of p-hacking, with one notable exception: \cite{friese2020p}, which documents a similar mitigation effect in simulations of a statistical model in which researchers p-hack by drawing significant p-values from a triangular distribution. 

Relatedly, a separate strand of work studies the optimal design of the publication rule in the presence of p-hacking.\footnote{Some other papers study optimal publication rules without allowing for p-hacking. See, e.g., \cite{frankel2022findings, kitagawa2023optimal}.} \citet{jagadeesan2024publication} designs the publication rule when slow p-hacking can distort researchers' incentives at the research-design stage. \citet{spiess2025optimal} designs estimation procedures when researchers' preferences over reported results are misaligned with society's. In contrast to the normative focus of these two papers, we take a \textit{positive} approach, as we take the publication rule as given and investigate the effect of p-hacking. In addition, our framework distinguishes between fast and slow p-hacking, which is not considered by \citet{jagadeesan2024publication} and \citet{spiess2025optimal}.

Second, our identification result contributes to an active strand of research on the detection of and correction for p-hacking and selective publication. Recent work has focused on the use of p-curves \citep{simonsohn2014p, ulrich2015p, simonsohn2015better, bruns2016p, elliott2022detecting, elliott2025power} caliper tests \citep{gerber2008statistical,gerber2008publication, kudrin2024robust} among other methods \citep{brodeur2016star,moss2023modelling,faridani2025p,kudrin2025testing} for detecting p-hacking. Closely related are methods for detecting and correcting for selective publication. Main approaches include selection modeling \citep{hedges1984estimation,iyengar1988selection, vevea1995general, mccrary2016conservative, andrews2019identification}, funnel asymmetry and related meta-regression methods \citep{egger1997bias, duval2000trim, stanley2008meta,stanley2014meta}, as well as Bayesian model averageing \cite{bartovs2023robust}.
Methods geared towards the detection of either p-hacking or selective publication often do not distinguish between the two sources of bias. One exception is \cite{brodeur2016star}, which proposes bounds for the proportion of p-hacked studies under the assumption of monotone selection on significance. Our identification result focuses on the selection-only counterfactual mean, which we obtain under parametric restrictions. 

\section{Setting}\label{section--setting}

We adopt the setting of \cite{andrews2019identification} but allow for p-hacking of the estimate and standard error. A researcher faces a latent study $(\Theta, \Sigma) \sim \mu$, where $\Theta$ is the effect size and $\Sigma$ is the standard error. He goes through a potential publication process as follows.
\begin{enumerate}
\item \textbf{(Initial draw)} The researcher draws an estimate $X_1$ from the distribution $\mathcal{N}(\Theta, \Sigma^{2})$.\footnote{The assumption of normal distribution is consistent with empirical practice, which typically performs inference under the premise of asymptotic normality.} We let $\Phi(x)$ denote its cumulative distribution function (CDF) and $\phi(x)$ the corresponding probability density function (PDF). 

\item \textbf{(Fast p-hacking)} The researcher can redraw estimates independently from $\mathcal{N}(\Theta, \Sigma^2)$. The cost of the $n$-th draw is $\kappa_f\cdot c_n$, where $\kappa_f > 0$ is a scaling parameter that reflects the general difficulty of fast p-hacking. For notational consistency, we let $c_1 = 0$, i.e., the initial draw is free. We assume that $c_n$ strictly increases in $n$, with $\lim\limits_{n\rightarrow\infty}c_n = \infty$. For simplicity, we assume that the researcher is naive and believes $\Theta=0$ throughout the process.\footnote{A similar assumption also appears in \citet{mccloskey2024critical}. Its purpose is to rule out the possibility that the researcher learns about $\Theta$ during fast p-hacking, thereby simplifying the analysis. Alternatively, we could assume that the researcher knows the value of $\Theta$ from the outset or naively believes that $\Theta=\mathbbm{E}_{\mu}(\Theta)$. All of our theoretical results continue to hold under these alternative assumptions.} We let $\boldsymbol{X}=\{X_1, X_2, ..., X_N\}$ denote the set of estimates if the researcher ends up making $N$ draws.  

\item \textbf{(Slow p-hacking)} Among all the estimates in $\boldsymbol{X}$, the researcher can pick one estimate $X_n$, perform slow p-hacking, and report $(\hat{X}, \hat{\Sigma})$. The cost of slow p-hacking is $\kappa_s \cdot c(|\hat{X} - X_n|, |\hat{\Sigma} - \Sigma|)$, satisfying $\kappa_s>0$, $c(0,0)=c_1(0,0) = c_2(0,0) = 0$, $c_{11}>0$, $c_{22}>0$, and $c_{12}= 0$.\footnote{One example of slow p-hacking cost function is $c(|\hat{X} - X_n|, |\hat{\Sigma} - \Sigma|) = c_x \cdot (\hat{X} - X_n)^2 + c_\Sigma \cdot (\hat{\Sigma} - \Sigma)^2$.} Notice that the researcher has the option of not undertaking slow p-hacking by letting $(\hat{X}, \hat{\Sigma}) = (X_n, \Sigma)$. 

\item \textbf{(Publication)} The reported result $(\hat{X}, \hat{\Sigma})$ is published (denoted by $D=1$) with probability $p(\hat{X}/\hat{\Sigma})$, which only depends on $\hat{X}/\hat{\Sigma}$ (the test statistic for the null hypothesis $\Theta = 0$ when the reported result is not manipulated), as we specify below. This publication rule is known to the researcher from the beginning of the entire process. The researcher receives a reward of $V > 0$ upon publication and zero otherwise. 
\end{enumerate}

The researcher's expected payoff from making $N$ draws, using the $n$-th estimate, and reporting $(\hat{X}, \hat{\Sigma})$ is 
\[V\cdot p(\hat{X}/\hat{\Sigma}) - \kappa_s \cdot c(|\hat{X} - X_n|, |\hat{\Sigma} - \Sigma|) - \kappa_f \cdot \sum_{n=1}^{N} c_n .\]

\paragraph{\underline{Model of p-hacking.}} \ Our model of p-hacking follows \cite{simonsohn2020fast}, which distinguishes between \emph{fast} and \emph{slow} p-hacking.

Fast p-hacking refers to practices that change p-values substantially from analysis to analysis. A leading example of which is file-drawering, whereby researchers discard insignificant data and repeat an experiment unchanged. We therefore model fast p-hacking as independent redrawing of estimates. The fixed $\Sigma$ reflects the idea that researchers cannot change the design of their study, either because they have already pre-registered their experiment, or have already committed to some data collection strategy. We assume that the cost of drawing new estimates is increasing, reflecting factors such as the scarcity of experiment sites or objects, as well as the increasing reputation cost of continued p-hacking. The cost of drawing new estimates diverges to infinity, capturing the fact that researchers do not have unlimited resources for conducting studies. Our model of fast p-hacking resembles \citet{mccloskey2024critical}, in which the researcher can draw up to $L$ estimates. It is also consistent with \citet[Supplement 3]{simonsohn2014p}, which models file-drawering as researchers obtaining a sequence of independent p-values.

Slow p-hacking refers to practices that change the p-value by a small amount from analysis to analysis. It arises because researchers often have to make many decisions when analyzing messy data. For example, they choose the thresholds for discretizing continuous variables, which outliers to drop, or the transformation of variables. Individual decisions probably do not affect the p-values by much. Collectively, they offer a way for researchers to approach the significance threshold through a sequence of tiny steps. As a result, the outcome of slow p-hacking rarely exceeds the target threshold by much. This produces the bunching behavior most associated with p-hacking, and is what caliper tests are designed to detect. As such, we model slow p-hacking as the direct and deterministic manipulation of estimates and standard errors. The manipulation cost increases with the distance between the reported result and the original one,\footnote{Our model is in line with the literature on communication with lying costs (e.g., \citet{kartik2009strategic}).} reflecting factors such as the increasing amount of time spent on specification search. 
Our model of slow p-hacking is similar to \citet{jagadeesan2024publication}, in which researchers pay a linear cost to manipulate the estimate only. 

As a side note, our model assumes that researchers engage in fast p-hacking before resorting to slow p-hacking. This ordering is without loss of optimality, as the outcome of slow p-hacking is deterministic.

\paragraph{\underline{Publication rule.}} \ In our model, the publication probability of a reported result $(\hat{X}, \hat{\Sigma})$ only depends on $\hat{X}/\hat{\Sigma}$, the test statistic for the null hypothesis $\Theta = 0$ when the reported result is not manipulated. Throughout the paper, we focus on the following publication rule: 
\begin{equation}
    p(z) = \begin{cases}
        1 & \text{ if } z \geq t, \\
        1/\delta_0 & \text{ if } -t < z < t, \\
        1/\delta_- & \text{ if } z \leq -t~,
    \end{cases}
\label{eq:general}
\end{equation}
where $\delta_0 \geq \delta_{-} \geq 1$. Under this rule, a positive significant result at the significance level $t$ is published with probability normalized to one, (weakly) higher than a negative significant result's publication probability, which in turn is (weakly) higher than that of an insignificant result.\footnote{The assumption that a positive significant result is more likely to be published than a negative significant one is without loss of generality. In the empirical analysis, we normalize signs so that this assumption holds; that is, if a literature favors negative results, we flip the signs of the estimates.} This rule reflects two common phenomena in scientific publishing. First, empirical results significant at a certain level (such as 5\%) are more likely to be published compared to those that are marginally insignificant at the same level.\footnote{For instance, such a preference for significant results is seen in the fields of economics \citep{brodeur2016star, chopra2024null}, medicine \citep{ioannidis2007exploratory}, political science \citep{gerber2008statistical}, psychology \citep{masicampo2012peculiar}, sociology \citep{gerber2008publication}, and social sciences as a whole \citep{franco2014publication}.} Second, in certain literature, editors may preferentially publish results with a particular sign. Specifically, they may prefer either ``reasonable results,'' which support their prior beliefs, or ``surprising results,'' which contradict their prior beliefs. For example, in the minimum wage literature, \cite{card1995time} and \cite{andrews2019identification} find evidence that editors favor results showing negative effects of minimum wage on employment. 

Note that this publication rule encompasses two special cases. It corresponds to \textit{two-sided selective publication} if $\delta_{-} = 1$ and \textit{one-sided selective publication} if $\delta_{-} = \delta_0$. These two publication rules stand as the extreme cases of the general rule we consider in the paper. 

\paragraph{\underline{Assumption on p-hacking cost.}} \ Finally, we introduce a regularity assumption on the cost of slow p-hacking to make the researcher's slow p-hacking behavior consistent with reality. 

\begin{remark}
\label{remark:uniquemin}
For $X\geq 0$, the problem $\min\limits_{ \tilde{X} \geq t\tilde{\Sigma} > 0} c(|\tilde{X} - X|, |\tilde{\Sigma} - \Sigma|)$ has a unique solution. 
\end{remark}
\begin{proof}
See Appendix~\ref{pf:remark:uniquemin}.
\end{proof}

Remark~\ref{remark:uniquemin} implies the following. Under the publication rule we consider in this paper, if the researcher decides to slow p-hack a positive insignificant estimate to positive significance, there is a unique result that minimizes the p-hacking cost among all candidates. Hereafter, when $0\leq X<t\Sigma$, we let $(\tilde{X},\tilde{\Sigma})=(\xi(X), \xi(X)/t)$ denote this unique minimizer, where we fix $\Sigma$ and express it as a function of $X$. Based on that, we also let $C(X):=c(|\xi(X) - X|, |\xi(X)/t - \Sigma|)$ denote the minimal cost of slow p-hacking this positive insignificant estimate towards positive significance. 

We can extend these two definitions to $X<0$, with $(\xi(X), -\xi(X)/t)$ being the optimal result chosen by the researcher if he slow p-hacks a negative insignificant estimate to negative significance, and $C(X) > 0$ being the corresponding slow p-hacking cost. Due to the symmetric structure of the slow p-hacking cost function $c(\cdot,\cdot)$ and the fact that positive and negative significances require the same level, we have that $\xi(X)$ is anti-symmetric around zero, and $C(X)$ is symmetric around zero. 

We are now ready to state the regularity assumption. 

\begin{assumption}
\label{assump:largekappas}
We have $\kappa_s \cdot C(0) > V$. 
\end{assumption}

Assumption~\ref{assump:largekappas} states that $\kappa_s$ is sufficiently large such that it is never worthwhile to undertake slow p-hacking when $X=0$. This further implies that the researcher's optimal slow p-hacking strategy does not change the sign of the estimate; that is, it is not worthwhile to slow p-hack a negative estimate to positive significance or vice versa. Although this assumption is not crucial for the qualitative results of the paper, it is consistent with the empirical observation that insignificant results are also published. If, instead, we allow for $\kappa_s \cdot C(0) \leq V$, then it becomes possible that every published result is significant under two-sided selective publication. 

\section{Two-Sided Selective Publication}\label{section:twosided}

This section focuses on \textit{two-sided selective publication}, i.e., $\delta_{-} = 1$. To save notation, we let $\delta_0 = \delta$ in this section, so the publication rule becomes 
\begin{equation*}
        p(z) = \begin{cases}
1 & \text{ if } |z| \geq t, \\
1/\delta & \text{ if } |z| < t,
        \end{cases}
    \end{equation*}
with $\delta > 1$ being referred to as the \textbf{extent of selective publication}. 

Since we are interested in the interplay of p-hacking and selective publication, we fix $(\Theta, \Sigma, \kappa_s, \kappa_f)$ and consider $\delta$, the extent of selective publication, as the primary parameter in the analysis. We let $\gamma_h(\delta):=  \mathbbm{E}(\hat{X}|D=1,\Theta, \Sigma)$,  in which the subscript ``h'' stands for hacking, denote the expected published result in our model, in which selective publication and p-hacking co-exist. Meanwhile, we let $\gamma(\delta)$ denote the expected published result when there is selective publication but \textit{no} p-hacking; equivalently, we have $\gamma(\delta) = \lim\limits_{ \kappa_s, \kappa_f\rightarrow\infty}\gamma_h(\delta)$, as the researcher does not undertake any p-hacking when it is too costly. Similarly, we let $\gamma_f(\delta)= \lim\limits_{ \kappa_s\rightarrow\infty}\gamma_h(\delta)$ and $\gamma_s(\delta)= \lim\limits_{ \kappa_f\rightarrow\infty}\gamma_h(\delta)$ denote the expected published result under the scenario of fast-only and slow-only p-hacking, respectively. 

\begin{assumption}
\label{assump:tiebreak}
The researcher adopts the following tie-breaking rule in Section~\ref{section:twosided}. \\
(a) If he is indifferent about undertaking a certain p-hacking action (i.e., his continuation value remains the same with or without a certain p-hacking action), he does \textit{not} p-hack. \\ 
(b) If he is indifferent among multiple results at the time of reporting, he reports the one with the lowest p-value. 
\end{assumption}

Assumption~\ref{assump:tiebreak} specifies the researcher's behavior in certain scenarios where he is indifferent. The paper's results remain intact under alternative reasonable tie-breaking rules (see, e.g., Footnotes~\ref{footnote:tiebreak1} and \ref{footnote:tiebreak2}). 

Note that we defined our outcomes of interest conditional on $(\Theta,\Sigma)$, so all our theoretical results hold pointwise. Without loss of generality, we only consider the case where the effect size is positive, i.e., $\Theta > 0$; the case with $\Theta < 0$ follows by symmetry.\footnote{When $\Theta=0$, due to the symmetry of our model, neither selective publication nor p-hacking induces any publication bias.} 

\subsection{Benchmark: Selective Publication Without p-Hacking}

It is well known that published results under two-sided selective publication tend to overestimate the magnitude of the treatment effect. This is shown by \cite{hedges1984estimation} for $F$-tests of mean differences and \cite{iyengar1988selection} for the $t$-test. For completeness, Proposition~\ref{prop:nohack} confirms that the same result holds in our setting. 

\begin{proposition}
\label{prop:nohack}
When $\Theta > 0$, we have $\gamma(\delta)>\Theta$, $\forall \delta>1$. 
\end{proposition}
\begin{proof}
See Appendix~\ref{pf:prop:nohack}.
\end{proof}

Taking this publication bias as the benchmark, we now explore whether p-hacking mitigates or exacerbates the bias. 

\begin{definition}
\label{def:mora}
Let $\Theta>0$. We say that p-hacking \textbf{mitigates} publication bias if $\Theta < \gamma_h(\delta) < \gamma(\delta)$ and \textbf{exacerbates} publication bias if $\Theta < \gamma(\delta) < \gamma_h(\delta)$. Similar definitions apply to slow-only p-hacking using the term $\gamma_s(\delta)$ and fast-only p-hacking using the term $\gamma_f(\delta)$, respectively. 
\end{definition}

\subsection{Slow-Only p-Hacking}
\label{sec:slow}

In this subsection, we study the implications of slow p-hacking by letting $\kappa_f \rightarrow \infty$. In this case, the researcher only draws one estimate, $X_1$, and decides whether to undertake slow p-hacking. The following proposition summarizes the optimal slow p-hacking strategy. 

\begin{proposition}\label{prop:slow_strat}
The researcher's optimal slow p-hacking strategy can be represented by a threshold $\overline{X}_s\in (0, t\Sigma ]$ such that: \begin{enumerate}[(i)]
\item He undertakes slow p-hacking if and only if $|X_1| \in (\overline{X}_s, t\Sigma)$.
\item His reported result is 
\begin{align*} %
(\hat{X},\hat{\Sigma}) = \begin{cases}
            (\xi(X_1), |\xi(X_1)|/t) & \quad\text{if }|X_1|\in (\overline{X}_s, t\Sigma),  \\
            (X_1,\Sigma) & \quad\text{otherwise}.
        \end{cases}
\end{align*}
\end{enumerate}
Furthermore, the slow p-hacking threshold $\overline{X}_s$ (weakly) decreases in $\delta$ and converges to $t\Sigma$ as $\delta$ goes to one.
\end{proposition}
\begin{proof}
See Appendix~\ref{pf:prop:slow_strat}.
\end{proof}

Proposition~\ref{prop:slow_strat} shows that the researcher undertakes slow p-hacking only when the insignificant estimate $X_1$ is sufficiently close to being significant ($|X_1| \in (\overline{X}_s, t\Sigma)$). If publication is more selective (higher $\delta$), then the incentive to p-hack is stronger, leading to a lower threshold for slow p-hacking (lower $\overline{X}_s$). Importantly, the optimal p-hacking strategy is geared to the publication rule, and, in particular, is symmetric around zero. 

Having characterized the researcher's optimal slow p-hacking strategy, we can further explore its implications for publication bias. 

\begin{theorem}
\label{thm:slow}
There exist two thresholds $1<\underline{\delta}_s< \overline{\delta}_s<\infty$ such that slow p-hacking mitigates publication bias if $\delta > \overline{\delta}_s$ and exacerbates publication bias if $\delta < \underline{\delta}_s$.
\end{theorem}
\begin{proof}
See Appendix~\ref{pf:thm:slow}.
\end{proof}

For intuition, suppose researchers cannot p-hack $\Sigma$ and consider all published results with $X \geq 0$. We can think of the overall distribution as a mixture of insignificant and significant estimates, with the latter having a higher mean. Slow p-hacking introduces two effects. On the one hand, researcher manipulation reduces the mass of insignificant estimates, thereby increasing the bias. On the other hand, manipulated results land exactly on the threshold. They are published with higher probability than before and have the smallest possible value among significant estimates. This lowers the mean of estimates conditional on significance. The relative strength of these two effects determines the overall outcome. In the extreme case with most severe publication selection ($\delta = \infty$), no insignificant results are published, thereby shutting down the first effect. This explains why slow p-hacking can mitigate publication bias when selective publication is sufficiently severe. 

\paragraph{Quadratic Costs.} Theorem \ref{thm:slow} does not stipulate what happens when the extent of selective publication is intermediate $\delta\in [\underline{\delta}_s,\overline{\delta}_s]$. In that region, there are countervailing forces at play, and our general model does not yield a clear prediction for the effect of slow p-hacking. However, we can sharpen Theorem \ref{thm:slow} by imposing additional structure on the p-hacking technology. For the following result only, suppose that the slow p-hacking cost is quadratic: 
\[c(|\hat{X}-X|,|\hat{\Sigma}-\Sigma|)=c_X\cdot (\hat{X}-X)^2+c_\Sigma\cdot (\hat{\Sigma}-\Sigma)^2\,,\quad c_X,c_\Sigma>0\,.\]

\begin{proposition}
\label{thm:slow_quad}
Suppose the cost of slow p-hacking is quadratic and that the function $\phi(x)+\phi(-x)$ is log-concave on $[t\Sigma-\sqrt{V(c_Xt^2+c_\Sigma)/(\kappa_sc_Xc_\Sigma)}, t\Sigma]$. Then there exists a single threshold $\delta^*\in(1,\infty)$ such that slow p-hacking mitigates publication bias if $\delta > \delta^*$ and exacerbates publication bias if $\delta < \delta^*$.
\end{proposition}
\begin{proof}
See Online Appendix~\ref{pf:thm:slow_quad}.
\end{proof}

The log-concavity assumption on $\phi$ is always satisfied when $\Theta\leq \Sigma$. It also holds if the slow p-hacking threshold under the most selective publication rule is not too low: $\lim_{\delta\rightarrow\infty}\overline{X}_s\geq 0.663 \Sigma$. When $t=1.96$, this is the case when a researcher never moves an estimate by more than about 1.3 standard errors. That is, when researchers only slow p-hack results with p-values below 0.51.

\subsection{Fast-Only p-Hacking}
\label{sec:fast}

In this subsection, we investigate the implications of fast p-hacking by letting $\kappa_s\rightarrow \infty$. Essentially, the researcher chooses when to stop drawing new estimates. 

Let $\Delta_f(\delta):= \mathbb{P}_{X\sim\mathcal{N}(0,\Sigma^2)}(|X|\geq t\Sigma) \cdot (1-\frac{1}{\delta}) V$ be the researcher's expected payoff gain from exactly one more draw if all his current results are insignificant. Let $n^*_f:= \max\{n \mid \kappa_f c_n < \Delta_f(\delta) \}$. The researcher's optimal strategy when he can only fast p-hack is as follows.

\begin{proposition}
\label{prop:fast_strat}
The researcher draws up to $n^*_f$ estimates until obtaining a significant result. Furthermore, $n^*_f$ is weakly increasing in the extent of selective publication $\delta$. 
\end{proposition}
\begin{proof}
See Appendix~\ref{pf:prop:fast_strat}.
\end{proof}

Based on the researcher's optimal strategy, we have the following result. 

\begin{theorem}\label{thm:fast}
There exists a threshold $\overline{\delta}_f \in (1, \infty)$ such that fast p-hacking exacerbates publication bias if $\delta > \overline{\delta}_f$ and has no effect otherwise. Moreover, the effect of fast p-hacking converges to zero when $\delta\rightarrow \infty$; that is, $\lim\limits_{\delta\rightarrow\infty}[\gamma_f(\delta) - \gamma(\delta)]=0$. 
\end{theorem}
\begin{proof}
See Appendix~\ref{pf:thm:fast}.
\end{proof}

The intuition behind Theorem~\ref{thm:fast} is as follows. Fast p-hacking enables the researcher to hide insignificant estimates, making significant estimates more likely to be reported than insignificant ones. In terms of overestimating the magnitude of the treatment effect, this force works in the same direction as selective publication, thereby exacerbating the publication bias. Moreover, as $\delta\rightarrow\infty$, the effect of fast p-hacking on publication bias vanishes since all insignificant estimates are unpublished in that limit, regardless of the extent to which the researcher hides them. When $\delta \leq \overline{\delta}_f$, fast p-hacking is never undertaken and has no effect on the distribution of published results. 

\subsection{General Model}
\label{sec:joint}

In this section, we consider situations where the researcher can undertake both types of p-hacking. Beyond generalizing our previous results, we emphasize how the two types of p-hacking interact. 

As before, it is strictly optimal for the researcher to stop engaging in any form of p-hacking once he obtains a significant result. If all his current results are insignificant, the decision of whether to draw another estimate depends on the expected benefits. Let $X_k^*:= \text{arg}\max\limits_{X_\kappa: \kappa\leq k}|X_\kappa|$ denote the researcher's ``best current estimate'' if he has drawn $k$ estimates (i.e., the one with the lowest p-value).\footnote{Throughout the paper, we restrict attention to the generic situations where no two estimates have identical absolute values.} Additionally, we let $\Delta_h:=\mathbb{P}_{X\sim \mathcal{N}(0, \Sigma^{2})}\left[|X|>\overline{X}_s\right] \times \mathbb{E}_{X\sim \mathcal{N}(0, \Sigma^{2})}\left[V(1-\frac{1}{\delta}) - \kappa_s C(X)\mid |X|>\overline{X}_s \right]$ and $n^*:= \max\{n\mid k_f\cdot c_n < \Delta_h\}$. 

\begin{proposition}\label{prop:combined_strat}
The researcher draws at most $n^*$ estimates. His optimal p-hacking strategy can be represented by a slow threshold $\overline{X}_s$ and a decreasing series of fast thresholds $\{\overline{X}_f^n\}_{n=1}^{n^*-1}$ satisfying
\[0< \overline{X}_s <\overline{X}_f^{n^*-1} < \cdots < \overline{X}_f^{2} < \overline{X}_f^{1}< t\Sigma. \]
Suppose he has already drawn $k<n^*$ estimates. He draws a $(k+1)$-th estimate if and only if his best current estimate $X_k^*$ satisfies $|X_k^*|<\overline{X}_f^k$. After he stops drawing new estimates, he undertakes slow p-hacking if and only if $|X_k^*|\in (\overline{X}_s,t\Sigma)$ and reports the following result:
\begin{align*}
(\hat{X}, \hat{\Sigma}) = \begin{cases}
(\xi(X_k^*), |\xi(X_k^*)|/t) & \quad\text{if }|X_k^*|\in (\overline{X}_s, t\Sigma),  \\
            (X_k^*,\Sigma) & \quad\text{otherwise}.
        \end{cases}
\end{align*}
\end{proposition}
\begin{proof}
See Appendix~\ref{pf:prop:combined_strat}.
\end{proof}

Proposition~\ref{prop:combined_strat} shows that, at any moment, the researcher's continuation p-hacking strategy depends on his best current estimate and how many draws have been made. Two features of his optimal p-hacking strategy are worth highlighting. First, the fast p-hacking threshold, $\overline{X}_f^n$, decreases in $n$, indicating that the researcher is less likely to fast p-hack if he has already drawn more estimates. This is driven by the fact that the fast p-hacking cost $c_n$ increases in $n$. Second, the slow p-hacking threshold is independent of $n$. The intuition is that once the researcher undertakes slow p-hacking, he must have decided to stop drawing new estimates, and therefore, only the best current estimate matters. Indeed, this threshold is the same as the p-hacking threshold in the special case with only slow p-hacking (see Section~\ref{sec:slow}). 

\begin{theorem}
\label{thm:combined}
There exist two thresholds $1<\underline{\delta}_h< \overline{\delta}_h<\infty$ such that p-hacking mitigates publication bias if $\delta > \overline{\delta}_h$ and exacerbates publication bias if $\delta < \underline{\delta}_h$.
\end{theorem}
\begin{proof}
See Appendix~\ref{pf:thm:combined}.
\end{proof}

Theorem~\ref{thm:combined} can be understood as follows. When $\delta$ is sufficiently small, the researcher does not undertake fast p-hacking, so the overall effect of p-hacking follows its slow component. When $\delta$ is sufficiently large, the researcher engages in both fast and slow p-hacking. As we point out in Theorem~\ref{thm:fast}, the effect of fast p-hacking converges to zero as $\delta$ approaches infinity. Therefore, in the limit, the overall effect of p-hacking is also captured by the slow component. 

While discussions about p-hacking generally suppose that it is undesirable, our result suggests the importance of considering p-hacking within the context of selective publication and that there exist simple, plausible models in which p-hacking may mitigate publication bias. Nevertheless, it is worth emphasizing that our result does not speak to other potential damaging effects of p-hacking, such as reducing trust in science and distorting the incentives of career researchers.

Theorems \ref{thm:slow} to \ref{thm:combined} compare $\gamma(\delta)$ with $\gamma_s(\delta)$, $\gamma_f(\delta)$, and $\gamma_h(\delta)$, respectively, aiming to investigate the effect of p-hacking on publication bias. The next theorem compares $\gamma_h(\delta)$ with $\gamma_s(\delta)$ and $\gamma_f(\delta)$, which enables us to understand the effect of banning one type of p-hacking, taking the existence of both types of p-hacking as the status quo. 

\begin{theorem}
\label{thm:fastandslow}
(a) Banning fast p-hacking always reduces publication bias. That is, for any $\kappa_s$, we have $\gamma_h(\delta) \geq \gamma_s(\delta) > \Theta$, $\forall \delta$. The first inequality is strict if and only if $n^* > 1$ (i.e., the researcher sometimes draws more than one estimate when fast p-hacking is allowed). \\
(b) Banning slow p-hacking increases publication bias when publication is selective enough. That is,  for any $\kappa_f$, $\gamma_f(\delta) > \gamma_h(\delta)  > \Theta$ when $\delta$ is sufficiently large.
\end{theorem}
\begin{proof}
See Appendix~\ref{pf:thm:fastandslow}.
\end{proof}

As established in Theorem~\ref{thm:fastandslow}, holding the slow p-hacking technology fixed, banning fast p-hacking is unambiguously desirable in reducing publication bias. By contrast, for any given fast p-hacking technology, prohibiting slow p-hacking backfires when selective publication is sufficiently stringent. 

\section{One-Sided Selective Publication}
\label{section:onesided}

This section analyzes the other extremal publication rule, \textit{one-sided selective publication}. To save notation, we let $\delta_{-} = \delta_0 =\delta$ in this section, so the publication rule becomes  
\begin{equation*}
        p(z) = \begin{cases}
1 & \text{ if } z \geq t, \\
1/\delta & \text{ if } z < t,
        \end{cases}
    \end{equation*}
where the parameter $\delta > 1$ is the \textbf{extent of selective publication}. Unlike Section~\ref{section:twosided}, we do \textit{not} restrict attention to $\Theta > 0$ in this section, since the publication rule is not symmetric around zero. 

The purpose of this section is to show that the main results of Section~\ref{section:twosided} continue to hold under one-sided selective publication: allowing p-hacking mitigates publication bias when publication selection is sufficiently severe, but exacerbates it when publication selection is sufficiently mild. To keep the analysis brief, we directly consider our general model where both fast and slow p-hacking are allowed. Also, since the proofs in this section follow similar reasoning to those in Section~\ref{section:twosided}, we relegate them to Online Appendix~\ref{onlineapx_secondaryproof}. 

Similarly to Section~\ref{section:twosided}, we introduce the following tie-breaking assumption. As before, its purpose is to simplify the analysis, and all results remain intact if we relax it. 

\begin{assumption}
\label{assump:tiebreak_onesided}
The researcher adopts the following tie-breaking rule in Section~\ref{section:onesided}. \\
(a) If he is indifferent about undertaking a certain p-hacking action (i.e., his continuation value remains the same with or without a certain p-hacking action), he does \textit{not} p-hack. \\ 
(b) If he is indifferent among multiple results at the time of reporting, he reports the one with the highest t-statistics. 
\end{assumption}

The following proposition states the researcher's optimal p-hacking strategy under one-sided selective publication. Let $\breve{X}_k^*:=\max\limits_{l\leq k}X_l$ denote the researcher's best current estimate if he has drawn $k$ estimates. 

\begin{proposition}
\label{prop:onesided_strat}
Under one-sided selective publication, the researcher draws at most $n^*$ estimates. His optimal p-hacking strategy can be represented by a slow threshold $\overline{X}_s$ and a decreasing series of fast thresholds $\{\overline{X}_f^n\}_{n=1}^{n^*-1}$ satisfying\footnote{For notational simplicity, we reuse $n^*$, $\overline{X}_s$, and $\{\overline{X}_f^n\}_{n=1}^{n^*-1}$ for both the two-sided and one-sided selective publication settings. Note that while the optimal slow p-hacking threshold $\overline{X}_s$ remains identical across both models, the values for $n^*$ and $\{\overline{X}_f^n\}_{n=1}^{n^*-1}$ are specific to the publication environment and differ between the two sections. \label{footnote1}}
\[0< \overline{X}_s <\overline{X}_f^{n^*-1} < \cdots < \overline{X}_f^{2} < \overline{X}_f^{1}< t\Sigma. \]
Suppose the researcher has already drawn $k<n^*$ estimates. He draws a $(k+1)$-th estimate if and only if $\breve{X}_k^* < \overline{X}_f^k$. After he stops drawing a new estimate, he undertakes slow p-hacking if and only if $\breve{X}_k^*\in (\overline{X}_s, t\Sigma)$ and reports the following result:
\begin{align*}
(\hat{X}, \hat{\Sigma}) = \begin{cases}
                        (\xi(\breve{X}_k^*), \xi(\breve{X}_k^*)/t) & \quad\text{if }\breve{X}_k^*\in (\overline{X}_s, t\Sigma),  \\
            (\breve{X}_k^*,\Sigma) & \quad\text{otherwise}.
        \end{cases}
\end{align*}
\end{proposition}
\begin{proof}
See Online Appendix~\ref{pf:prop:onesided_strat}.
\end{proof}

We let $\breve{\gamma}(\delta)$ denote the expected published estimate with one-sided selective publication and no p-hacking, and $\breve{\gamma}_h(\delta)$ denote the expected published estimate with both selective publication and p-hacking. It is not difficult to see that one-sided selective publication gives rise to publication bias, i.e., $\breve{\gamma}(\delta) > \Theta$ for $\delta>1$. We say that p-hacking mitigates publication bias if $\breve{\gamma}(\delta) > \breve{\gamma}_h(\delta) > \Theta$ and exacerbates publication bias if $\breve{\gamma}_h(\delta) > \breve{\gamma}(\delta) > \Theta$. 

\begin{theorem}
\label{thm:onesided}
Under one-sided selective publication, there exist two thresholds $1<\underline{\delta}_1< \overline{\delta}_1<\infty$ such that p-hacking mitigates publication bias if $\delta > \overline{\delta}_1$ and exacerbates publication bias if $\delta < \underline{\delta}_1$.
\end{theorem}
\begin{proof}
See Online Appendix~\ref{pf:thm:onesided}.
\end{proof}

Theorem~\ref{thm:onesided} confirms that the main insights of Section~\ref{section:twosided} continue to hold under one-sided selective publication. 

\section{Empirical Demonstration}\label{sec:emp}

Our analysis shows p-hacking can mitigate or exacerbate the effects of publication bias, depending on the selection environment. These results were obtained under simplifying assumptions that may be restrictive. We show in this section that both exacerbation and mitigation can arise under data-generating processes that are empirically plausible.

To better match data, we generalize our model to allow slow p-hacking to be potentially stochastic in Section \ref{section:emp_model}. Section \ref{section:identification} shows that it is possible to identify the selective-publication-only counterfactual at the literature level under the assumption that $\Theta \mid \Sigma = s \sim N(\theta(s),\sigma^2(s))$. Nonetheless, identification and estimation appear tricky in practice. Section \ref{sec:emp_mle} introduces parametric assumptions so that the restricted model can be estimated by maximum likelihood. In Section \ref{section:emp_results}, we estimate our model using two meta-analyses, finding suggestive evidence that both mitigation and exacerbation can arise in practice.

\subsection{Empirical Model with Stochastic Slow p-Hacking}\label{section:emp_model}

This section describes our empirical model, which generalizes the theoretical model by allowing slow p-hacking to be potentially stochastic. As will become clear in the next two sections, our empirical applications are more suited to models of one-sided selection. Therefore, we develop the empirical model based on one-sided selection.\footnote{Online Appendix \ref{app:twosided_emp} treats a more general empirical model in which $\delta_- \in [1, \delta_0]$.} 

We modify the slow p-hacking setting as follows. Suppose a researcher slow p-hacks an estimate $X_n>0$ to positive significance. The associated cost is still $\kappa_s C(X_n)$. However, unlike the baseline model, we now let the generated output $(\hat{X}, \hat{\Sigma})$ be stochastic: it is drawn from an arbitrary distribution $H(\cdot, \cdot \mid X_n, \Sigma)$ and must satisfy $\hat{X}/\hat{\Sigma} \in[t, t_w]$ for some known $t_w \geq t$. %
Note that if we let $H$ be the Dirac-delta function that assigns probability one to $(\xi(X_n), \xi(X_n)/t)$, this modified slow p-hacking setting is identical to that in the theoretical model. As discussed in \cite{simonsohn2020fast}, slow and fast p-hacking differ in how much control researchers have over their p-values. Allowing slow p-hacking to be stochastic captures the fact that this control is imperfect. Researchers nonetheless retain substantially more control under slow than fast p-hacking, as encoded by the support restriction on the former. 

Notice that the researcher's optimal p-hacking strategy in this modified setting still follows Proposition~\ref{prop:onesided_strat}. Intuitively, allowing the outcome of slow p-hacking to be stochastic does not affect the researcher's incentives at any stage of the p-hacking process. Although the outcome of slow p-hacking is now stochastic, slow p-hacking yields the researcher the same payoff and incurs the same cost as in the theoretical model. Therefore, the researcher's slow p-hacking strategy remains unchanged. As a consequence, his incentive to engage in fast p-hacking also remains unchanged, since the continuation payoff from stopping drawing new estimates and potentially engaging in slow p-hacking is the same as in the theoretical model. Therefore, we can continue to use Proposition~\ref{prop:onesided_strat} for characterizing the researcher's p-hacking behavior in the empirical model. 

\subsection{Identification under Normality}\label{section:identification}

Up to this point, our analysis has operated at the \textit{study level}, as our theoretical results characterized $\gamma_h$ and $\gamma$ conditional on the study $(\Theta, \Sigma)$. Going forward, we will turn to identifying parameters at the \textit{literature level}. We focus on the joint distribution of $(\Theta, \Sigma)$ and the mean effect size $\mathbb{E}(\Theta)$, as well as the means under p-hacking and selective publication and selective publication only. With some abuse of notation, we define the latter quantities as:
\begin{equation*}
    \mathbb{E}(\gamma_h) := \frac{\mathbb{E}_{\Theta, \Sigma}(q_h(\Theta, \Sigma; \delta)\gamma_h(\delta))}{\mathbb{E}_{\Theta, \Sigma}(q_h(\Theta, \Sigma; \delta))} \quad \mbox{ and } \quad \mathbb{E}(\gamma) := \frac{\mathbb{E}_{\Theta, \Sigma}(q(\Theta, \Sigma; \delta)\gamma(\delta))}{\mathbb{E}_{\Theta, \Sigma}(q(\Theta, \Sigma; \delta))} ~,
\end{equation*}
where $q_h(\Theta, \Sigma; \delta)$ and $q(\Theta, \Sigma; \delta)$ are the publication probabilities $\mathbb{P}(D = 1 \mid \Theta, \Sigma; \delta)$ under under the respective scenarios. We suppress the dependence of the above parameters on $\delta$ to lighten notation.

To simplify the analysis, we will assume that $\Theta \mid \Sigma = s \sim N(\theta(s), \sigma^2(s))$. This assumption nests the classic random effects model, in which $\Theta$ has a normal marginal distribution and is independent of $\Sigma$. This allows us to identify the parameters of interest:

\begin{theorem}\label{thm:identification_onesided}
   Suppose $t$ and $t_w$ are known and that $\Theta \mid \Sigma = s \sim N(\theta(s), \sigma^2(s))$. Moreover, suppose $\mathbb{P}(\hat{X}_i/\hat{\Sigma}_i > t_w), \mathbb{P}(|\hat{X}_i/\hat{\Sigma}_i| < t) > 0$ and that $\sigma^2(s) > 0$ for all $s \in \text{Supp}(\Sigma)$. Then, the joint distribution of $(\Theta, \Sigma)$, as well as $\delta$ are identified. In particular, $\mathbb{E}(\Theta)$, $\mathbb{E}(\gamma)$ and $\mathbb{E}(\gamma_h)$ are identified. %
\end{theorem}
\begin{proof}
See Appendix~\ref{pf:thm:identification_onesided}.
\end{proof}

Our proof makes use of identification-at-infinity as well as deconvolution arguments. As such, identification may be weak and estimation may be challenging in practice. This motivates the additional restrictions that we describe in the next section. 

Theorem \ref{thm:identification_onesided} requires two parametric assumptions. The first is the normality of $\Theta \mid \Sigma$, which allows us to identify the full distribution of $(\Theta, \Sigma)$ using only observations for which $\hat{X}_i/\hat{\Sigma}_i > t_w$. The second is the assumption that selection takes the form of a step function. This allows us to separate the effects of p-hacking from that of selection, to back out $\delta$ and hence $\mathbb{E}(\gamma)$. $\mathbb{E}(\gamma_h)$ is identified since it is the mean of published papers. We can thus infer the amount of bias arising from selective publication and assess whether p-hacking exacerbates or mitigates it. 

Additionally, $t$ and $t_w$ need to be specified. Since a large empirical literature demonstrates the bunching of p-values at $5\%$ and $1\%$ across different fields of study \citep{masicampo2012peculiar, benjamin2018redefine, andrews2019identification}, it is natural to set $t = 1.96$. $t_w$ does not need to be exactly known. It only needs to be sufficiently large so that slow p-hacking is unlikely to exceed it. To accommodate p-hacking up to the 1\% level of significance, it is natural to set $t_w = 2.58 + \varepsilon$.

Note that the model is not fully identified. For example, the thresholds for slow p-hacking as well as the cost functions for fast and slow p-hacking are not identified. 

\subsection{Maximum Likelihood Estimation}\label{sec:emp_mle}

The previous section shows that the counterfactuals of interest are semi-parametrically identified. In practice, however, identification and estimation may be challenging. 

To make headway, we impose that researchers can draw at most one additional estimate, so that $n^*=2$. We set $t_w = 4$ to accommodate slow p-hacking up to the 1\% level of significance (critical value = 2.58) while ensuring sufficient mass in the tail for identification (see Theorem \ref{thm:identification_onesided}). The buffer allows researchers to overshoot their slow p-hacking target by a reasonable amount. Appendix \ref{app:emp_robust} shows that the results in the next section are not sensitive to the choice of $n^*$ and $t_w$.

Moreover, we will assume that $\Theta$ and $\Sigma$ are independent and let $(\theta, \sigma^2):=(\mathbb{E}(\Theta),\text{Var}(\Theta))$. \cite{andrews2019identification} require this assumption to identify the probability of selection using meta-analyses. In our case, independence reduces the nonparametric problem of estimating $(\theta(s), \sigma(s))$ to a parametric problem. Furthermore, as in the empirical sections of \cite{andrews2019identification}, we will assume that $\Sigma \sim \Gamma(\kappa, \lambda)$, where $\kappa$ and $\lambda$ are the shape and the rate parameters of the Gamma distribution to be estimated. 

Finally, while we previously treated $\Sigma$ as given and formulated the slow p-hacking cost essentially as $C(X;\Sigma)$, our estimation adopts a more structured specification. We assume that
\[C(X; \Sigma)=\tilde{C}(X/\Sigma),\]
so that the cost of slow p-hacking depends only on the unmanipulated p-value associated with the estimate being manipulated. This specification reflects the idea that what matters for the difficulty of obtaining statistical significance is not the estimate itself, but its statistical significance relative to its standard error. Intuitively, it is easier to manipulate $X$ in an environment that is noisier. Combined with the assumption that researchers believe $\Theta=0$, this specification implies that $\overline{X}_f^n(s) = \overline{X}_f^n \cdot s$ for constant $\overline{X}_f^n$. 

Under the above assumptions, we are able to specify a likelihood function for a censored version of the data with only a finite number of identified parameters. The likelihood function can be found in Appendix \ref{app:likelihood}. It describes censored data where observations in the bins $(-t,t)$ and $(t, t_w]$ are treated as indistinguishable from each other. As such, we do not have to specify the extent to which researchers slow p-hack $\hat{X}_i$ relative to $\hat{\Sigma}_i$, or any randomness that may be part of the slow p-hacking process, as long as the resulting t-statistic is upper bounded by $t_w$ in absolute value. 

Additionally, we impose the constraints that $-t \leq \overline{X}_f^{n*-1} \leq \cdots \leq \overline{X}_f^1 \leq t$ as well as other model-implied constraints that are described in Appendix \ref{app:likelihood}. To avoid constrained optimization, as well as issues with parameters at the boundary, we implement them as strict inequalities via appropriate transformations of the variables.

Because of weak identification and error from numerical integration, the Hessian can be poorly conditioned. To compute standard errors for $\mathbb{E}(\gamma)$, we implement the Moore-Penrose inverse for the Hessian, dropping eigenvalues that are smaller than $10^{-9} \times$ the largest eigenvalue. In all cases, we verify that the dropped directions are orthogonal to those needed to compute the variance of $\mathbb{E}(\gamma)$, so that omitting them have minimal effects. 

\subsection{Results}\label{section:emp_results}

We apply the above maximum likelihood procedure to data from two meta-analyses and find that both exacerbation and mitigation can be relevant in practice. Section \ref{section:emp_nudge} considers the effectiveness of nudges in promoting desired choices. We find that p-hacking may have mitigated the effects of publication bias in this literature. Section \ref{section:emp_aid} presents our results on the effectiveness of development aid on promoting growth, suggesting that p-hacking has exacerbated the effects of publication bias in this literature.

\subsubsection{Effect of Nudges on Choices}\label{section:emp_nudge}

The meta-analysis of \cite{mertens2022effectiveness} studies the effectiveness of nudges---interventions in the choice architecture---for promoting personally or socially desirable behavior. The analysis is based on 455 estimates from 334 studies\footnote{Studies refer to independent experiments, possibly belonging to the same paper. \cite{mertens2022effectiveness} treat studies as independent units of analysis, which is standard in meta-analysis in Psychology.}, obtained from 213 published and unpublished papers. The authors find that nudges are effective in altering behavior, with the headline result that $\mathbb{E}(\Theta) \approx 0.45$, measured in Cohen's $d$.\footnote{Cohen's $d$ is a standardized measure of effect size commonly used in meta-analysis.} \cite{maier2022no} argues instead that there is substantial selection in publication. Correcting for it using Bayesian model averaging, they estimate $\mathbb{E}(\Theta)$ to be between 0.04 and 0.11, with Bayes factor below $1$, implying weak evidence that the parameter is non-zero.

We apply our method to the data of \cite{maier2022no}.\footnote{This is the corrected version of \cite{mertens2022effectiveness}, though it assigns 2 papers to \texttt{publication\_id} 95. Fixing this discrepancy leads to a total of 213 papers. %
}  Since it does not distinguish between published and unpublished papers, we use ChatGPT 5.6 Sol High to obtain the publication status of each paper, dropping the 6 unpublished papers from our sample\footnote{The unpublished papers have \texttt{publication\_id} 17,18,19,20,22 and 87.}. Following \cite{maier2022no}, we treat studies as independent units of analysis. Some studies present multiple estimates based on the same data. We choose the most significant estimates as our preferred sample. This corresponds to the notion that a paper needs at least one significant result to be ``publishable". Our estimate for $\mathbb{E}(\gamma_h)$ is then the simple mean of these preferred estimates. 

Our final sample comprises 328 studies (from 207 published papers), each contributing one estimate. Histograms of the t-statistics for the full sample and our preferred sample are presented in Figure \ref{fig:emp_nudge}. The apparent bunching of the t-statistics at the threshold of 1.96 suggests appreciable slow p-hacking in this literature. We might therefore expect p-hacking to have a mitigative effect here. Comparing the two plots, we also see that retaining only the most significant estimate does not qualitatively alter the distribution of t-statistics. 

\begin{figure}
\centering
\includegraphics[width=0.48\linewidth]{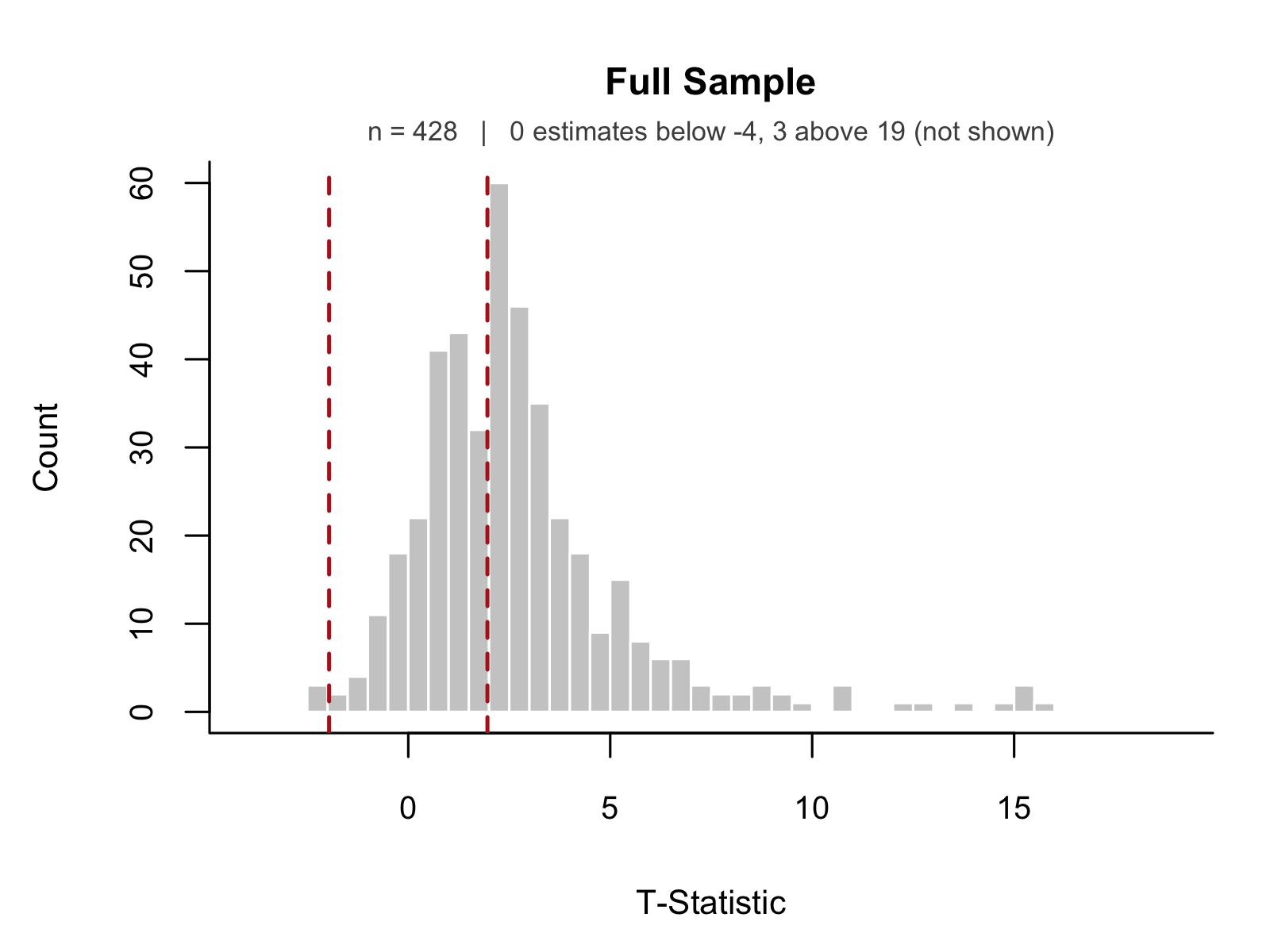}
\includegraphics[width=0.48\linewidth]{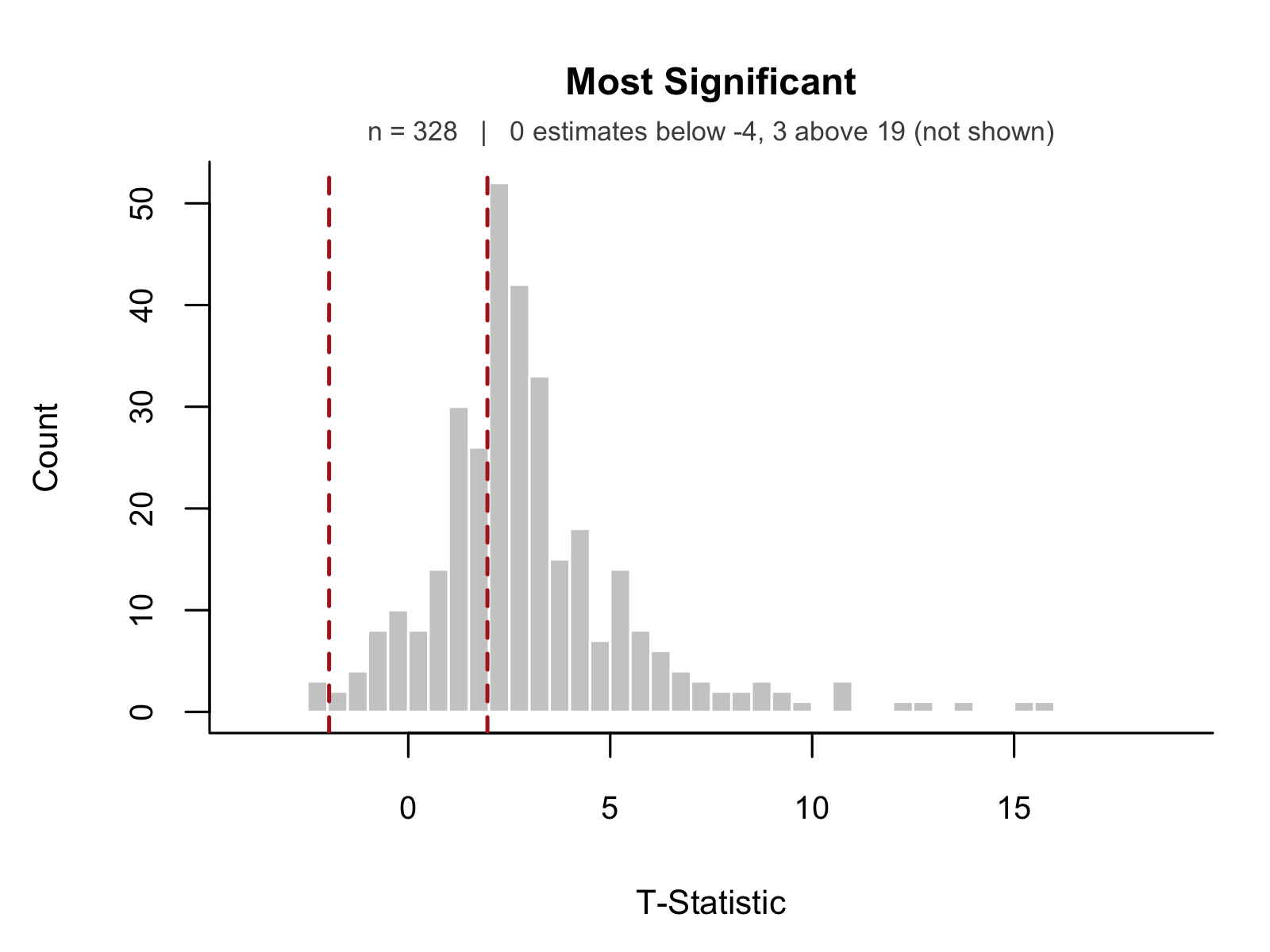}
\caption{Histograms of the t-statistics from the full published sample with multiple estimates per study, and from our preferred sample, which keeps only the most significant estimates. Red dashed lines indicate $\pm 1.96$.}\label{fig:emp_nudge}
\end{figure}

As discussed in Section \ref{sec:emp_mle}, we set $n^* = 2$ and $t_w = 4$. Because the left-tail is non-existent, the two-sided selection model is not identified and we impose the one-sided selection model. Appendix \ref{appendix:emp_robust_nudge} shows that the results below are not sensitive to the choices of $n^*, t_w$ or our choice of preferred sample.

Results are presented in Table \ref{tab:emp_nudge}. We find that negative significant and insignificant results are published at 8.3\% the rate at which positive and significant results are published. Correcting for selective publication and p-hacking, we estimate $\mathbb{E}(\Theta)$ to be 0.171. It is insignificant at the 5\% level and comparable to the upper end of the values obtained by \cite{maier2022no}. We estimate $\mathbb{E}(\gamma)$ to be $0.959$ and $\mathbb{E}(\gamma_h)$ to be $0.531$, suggesting that p-hacking has mitigated the bias from selective publication by 54\%. That the mitigation from slow p-hacking dominates the exacerbation effects in this literature is in line with our intuition from the bunching in Figure \ref{fig:emp_nudge}.

\begin{table}[htbp]
  \centering
    \begin{tabular}{ccccc}\hline\hline
    $1/\delta$ & $\mathbb{E}(\Theta)$ & $\mathbb{E}(\gamma_h)$ & $\mathbb{E}(\gamma)$ & $\mathbb{E}(\gamma) - \mathbb{E}(\gamma_h)$ \\
    \hline
    0.083 & 0.171 & 0.531 & 0.959 & 0.428 \\
    (0.031) & (0.130) & (0.028) & (0.180) & (0.178) \\
    \hline\hline
    \end{tabular}%
    \caption{MLE on the dataset of \cite{mertens2022effectiveness} allowing for selective publication and p-hacking. $1/\delta$ is the relative publication probabilities for negative significant and insignificant results compared to positive and significant results. $\mathbb{E}(\Theta)$ is the mean effect size under the true distribution of latent studies. $\mathbb{E}(\gamma_h)$ is the observed mean under both selective publication and p-hacking. $\mathbb{E}(\gamma)$ is the counterfactual mean under selective publication only. Standard errors computed using Moore-Penrose inverse are presented in round brackets.} \label{tab:emp_nudge}%
\end{table}%

\subsubsection{Effect of Development Aid on Economic Growth}\label{section:emp_aid}

\cite{doucouliagos2011ineffectiveness} conducts meta-analysis on the effect of development aid on economic growth. Using a dataset of 105 papers and 1217 estimates, the authors find that positive and significant results are selectively published, and that correcting for these results using meta-regressions leads to estimates of partial correlations\footnote{Partial correlations are obtained by standardizing regression estimates to make them more comparable across studies. They are commonly used in meta-analysis.} that are small in magnitude---between 0.02 to 0.04---and often statistically insignificant. 

We apply our method to the data of \cite{doucouliagos2011ineffectiveness}, which was updated in \cite{doucouliagos2013robust} to include a total of 113 papers and 1347 estimates. As before, we select the most significant estimate in each study as the preferred estimate. We plot the t-statistics from the full sample as well as our preferred sample in Figure \ref{fig:emp_aid}. In contrast to the literature on nudges, there is no pronounced excess mass at the significance thresholds, so that we might expect less mitigative effects from slow p-hacking. Comparing the two plots in Figure \ref{fig:emp_aid}, we see that the distribution of our preferred estimates puts substantially less mass in the insignificant region than the full sample. This follows mechanically from our selection criterion. If we believe that a paper is assessed based on its strongest results, then our preferred estimates more closely reflect the statistics that are marginal for an editor's decision, and which are therefore candidates for p-hacking. Our theoretical analysis also suggests that reducing the mass of insignificant results will make exacerbation less likely. As such, if the most significant sample is not the marginal sample, we loosely interpret the results below as a lower bound on the exacerbation effect of p-hacking in this literature. 

\begin{figure}
\centering
\includegraphics[width=0.48\linewidth]{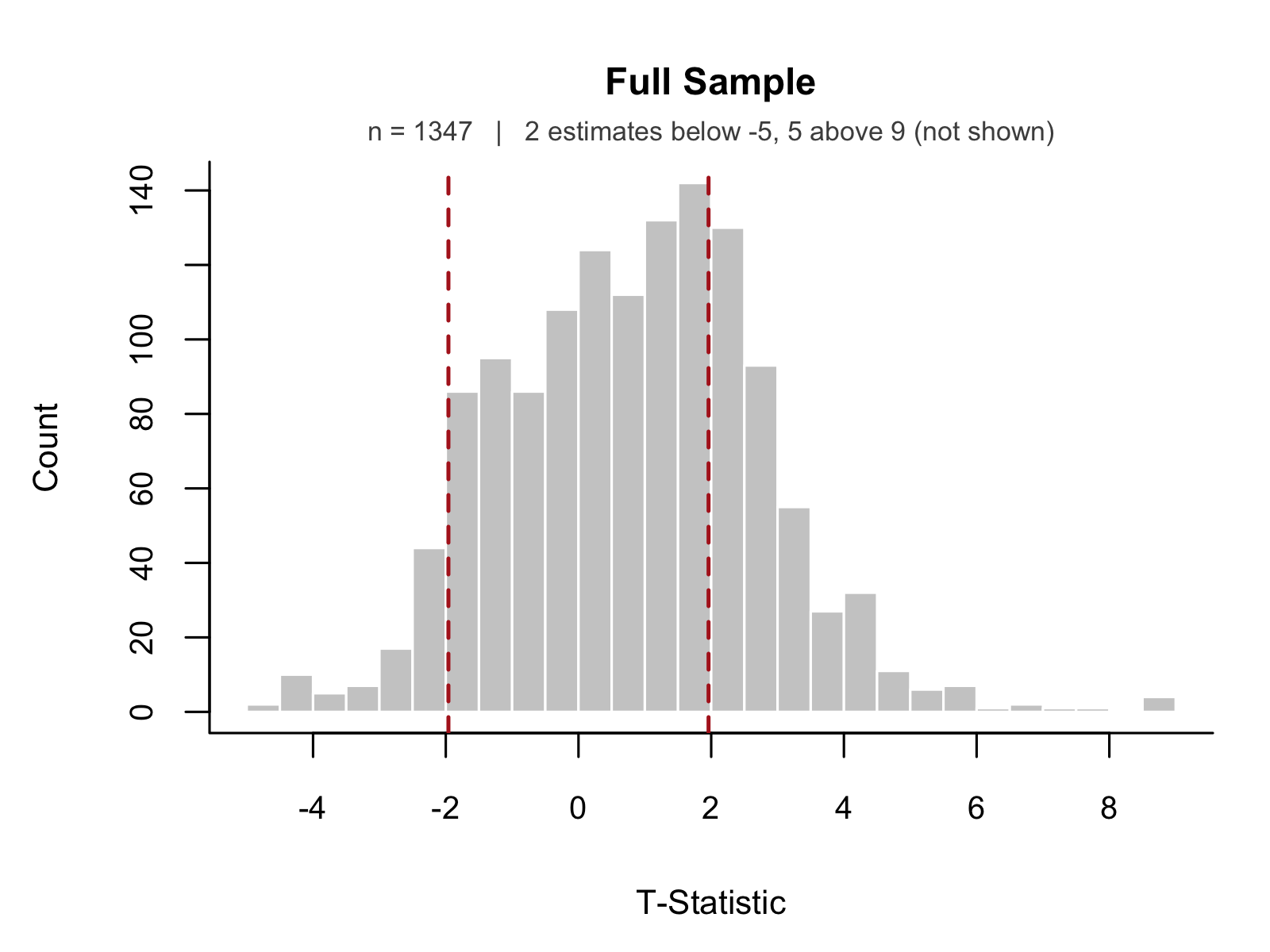}
\includegraphics[width=0.48\linewidth]{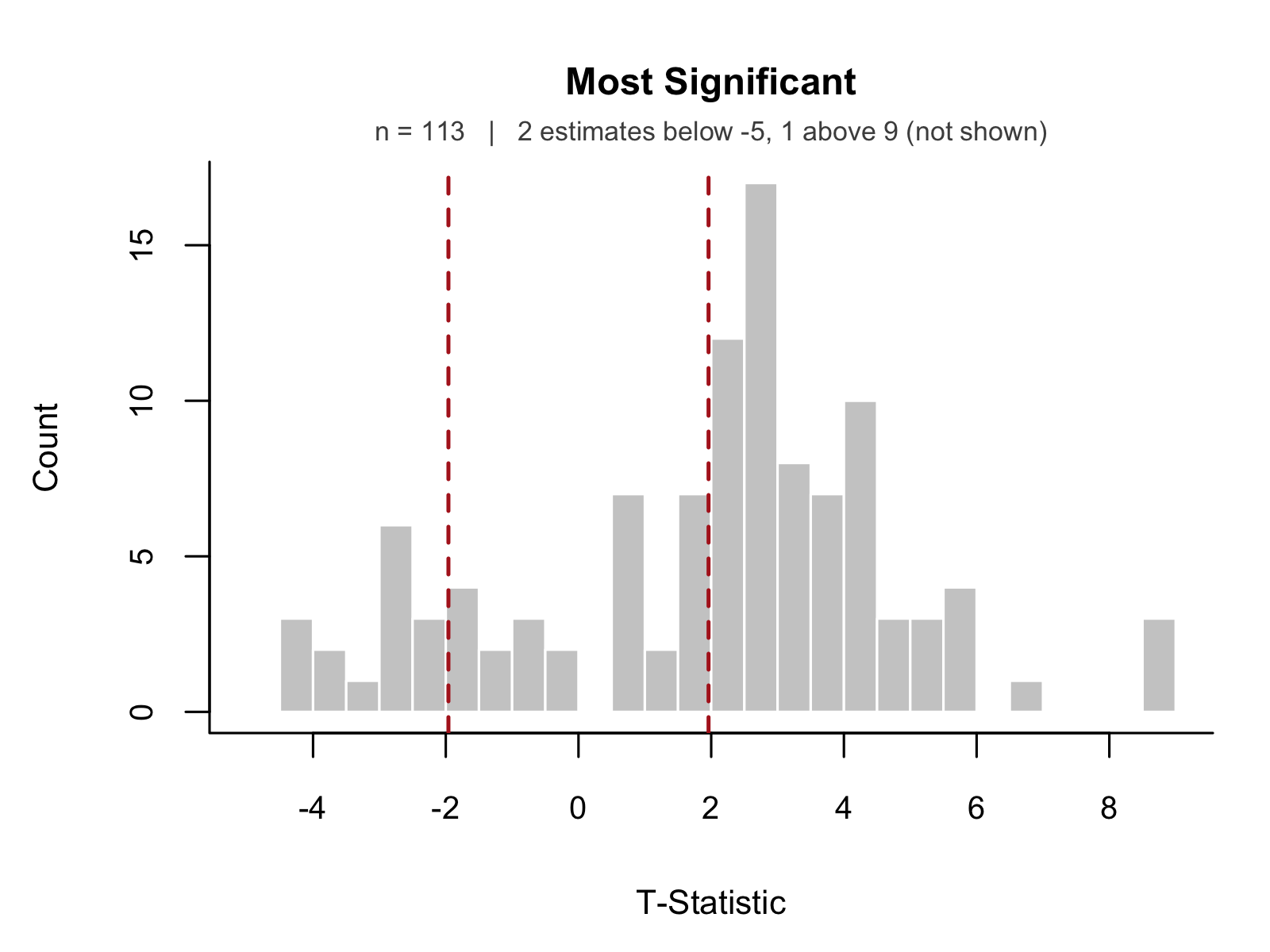}
\caption{Histograms of the t-statistics from the full published sample with multiple estimates per paper, and from our preferred sample, which keeps only the most significant estimates. Red dashed lines indicate $\pm 1.96$.}\label{fig:emp_aid}
\end{figure}

As before, we fix $n^* = 2$, $t_w = 4$. We also impose one-sided selection, since \cite{doucouliagos2011ineffectiveness} argues that the two leading sources of publication bias in this literature are (1) researcher idealism and (2) incentives for increasing development funding, both of which favor positive estimates. Appendix \ref{appendix:emp_robust_aid} shows that the results below are not sensitive to the aforementioned choices, although fitting a general model that estimates $\delta_-$ leads to estimates with substantially larger standard errors.

Results are presented in Table \ref{tab:emp_aid}. We find that negative significant and insignificant results are published at 15\% the rate at which positive and significant results are published. Correcting for selection and p-hacking, estimate for $\mathbb{E}(\Theta)$ to be $-0.113$,  which is negative but insignificant at the 5\% level. We estimate $\mathbb{E}(\gamma)$ to be $0.121$ and $\mathbb{E}(\gamma_h)$ to be $0.185$, suggesting that p-hacking has exacerbated the bias from selective publication by 27\%. \cite{doucouliagos2013robust} argues that researchers in the aid literature cherry-pick control variables in order to achieve positive and significant estimates. To the extent that this behavior is captured by our stylized model of fast p-hacking, it would explain why exacerbation dominates mitigation in this literature. 

\begin{table}[htbp]
  \centering
    \begin{tabular}{ccccc}\hline\hline
     $1/\delta$ & $\mathbb{E}(\Theta)$ & $\mathbb{E}(\gamma_h)$ & $\mathbb{E}(\gamma)$ & $\mathbb{E}(\gamma) - \mathbb{E}(\gamma_h)$ \\
    \hline
     0.149 & -0.113 & 0.185 & 0.121 & -0.064 \\
    (0.057) & (0.061) & (0.032) & (0.038) & (0.017) \\
    \hline\hline
    \end{tabular}%
    \caption{MLE on the dataset of \cite{doucouliagos2011ineffectiveness} allowing for selective publication and p-hacking. $1/\delta$ is the relative publication probabilities for negative significant and insignificant results compared to positive and significant results. $\mathbb{E}(\Theta)$ is the mean effect size under the true distribution of latent studies. $\mathbb{E}(\gamma_h)$ is the observed mean under both selective publication and p-hacking. $\mathbb{E}(\gamma)$ is the counterfactual mean under selective publication only. Standard errors computed using Moore-Penrose inverse are presented in round brackets.} \label{tab:emp_aid}%
\end{table}%

\section{Conclusion}\label{section--conclusion}

This paper studies the effects of p-hacking on the bias of published estimates in environments with selective publication. We show that under both one- and two-sided selection, fast p-hacking unambiguously exacerbates the effects of publication bias, whereas slow p-hacking can either exacerbate or mitigate it, depending on the extent of selection in the publication process. The possibility of mitigation echoes the work of \cite{kuhn1962structure}: communities of biased researchers can nonetheless produce useful work due to self-correcting dynamics in the scientific process.

Although these results are developed in a stylized setting, we show that the selection-only counterfactual mean is identified under a normality assumption. Taking the model to the data, we find suggestive evidence that p-hacking mitigates the effects of selective publication in the literature on the effects of nudges, but exacerbates it in the literature on the effects of aid on growth.

Discussions about p-hacking and selective publication often treat the two phenomena separately. Our results show that it is important to understand p-hacking in the context of the selection environment that produced it, echoing \cite{friese2020p}, which observes their interacting effects in simulations. We stress, however, that our results concern only the bias of published estimates. They say nothing about the other pernicious effects of p-hacking, such as inflating the number of false positives, eroding trust in science, or distorting the incentives of career researchers. What is clear is that p-hacking and selective publication cannot be understood in isolation. Only by studying them together can we hope to model the scientific process faithfully and design publication policies that preserve its credibility.

\newpage
\begin{singlespace}
\bibliographystyle{chicago}
\bibliography{phacking.bib}
\end{singlespace}

\clearpage
\newpage

\appendix

\section{Proofs of Main Results}

\subsection{Proof of Remark \ref{remark:uniquemin}}
\label{pf:remark:uniquemin}

If $X\geq t\Sigma$, the unique minimizer is $(\tilde{X}, \tilde{\Sigma}) = (X, \Sigma)$, i.e., it is unnecessary to undertake slow p-hacking, as the estimate is already significant. If $X\in\left[0, t\Sigma\right)$, then the optimization problem is equivalent to $\min\limits_{\tilde{\Sigma}} c(t\tilde{\Sigma} - X, \Sigma - \tilde{\Sigma})$. This objective function is strictly convex in $\tilde{\Sigma}$, since 
\[\frac{d^2 c(t\tilde{\Sigma} - X, \Sigma - \tilde{\Sigma})}{(d \tilde{\Sigma})^2} = c_{11} t^2 - 2c_{12} t + c_{22},\]
which is strictly positive because $c_{11}>0$, $c_{22}>0$, and $c_{12} = 0$. Therefore, we can infer that the optimization problem admits a unique minimizer for $X\in\left[0, t\Sigma\right)$. 

\subsection{Proof of Proposition \ref{prop:nohack}}
\label{pf:prop:nohack}

We begin with the following lemma, which will be used in several proofs in the appendix. For a given symmetric non-negative function $m(x)$, we define $K(m):=\frac{\int_{-\infty}^{\infty} x\cdot m(x)\cdot \phi(x) d x}{\int_{-\infty}^{\infty} m(x)\cdot \phi(x) d x}$ as the expected value of $x$ under the ``kernel'' $m(x)$. This definition applies to several terms in this paper. For instance, the true treatment effect $\Theta$ can be expressed as $K(m_u)$ where $m_u(x)\equiv 1$ is a uniform kernel function, and $\gamma(\delta)$ can be expressed as $K(m_n)$ where 
\[ m_n(x)=\begin{cases}
1 & \text{ if } |x| \geq t\Sigma \\
1/\delta & \text{ if } |x| < t\Sigma
        \end{cases}
\]
is a kernel function adapted from the publication rule.

\begin{lemma}
\label{lem:key}
If two symmetric non-negative functions $m_1(x)$ and $m_2(x)$ are such that $\frac{m_1(x)}{m_2(x)}$ weakly increases and is not constant in $x$ for $x \in \mathcal{X}:=\{x\geq 0 \mid m_1(x)+m_2(x)>0\}$, then we have $K(m_1) > K(m_2)$.\footnote{Slightly abusing the notation, we define the value of $\frac{m_1(x)}{m_2(x)}$ on the extended real line, so $\frac{m_1(x)}{m_2(x)} = +\infty$ when $m_1(x) > 0$ and $m_2(x) = 0$.}
\end{lemma}
\begin{proof}
Since $m_i(x)$ is symmetric, we have 
\begin{eqnarray*}
K(m_i) &=& \frac{\int_{0}^{\infty} x \cdot m_i(x) \cdot [\phi(x)-\phi(-x)] d x}{\int_{0}^{\infty} m_i(x) \cdot [\phi(x)+\phi(-x)] d x} \\
&=& \frac{\int_{0}^{\infty} x \cdot \rho(x) \cdot m_i(x) \cdot [\phi(x)+\phi(-x)] d x}{\int_{0}^{\infty} m_i(x) \cdot [\phi(x)+\phi(-x)] d x} \\
&=& \frac{\int_{0}^{\infty} x \cdot \rho(x) \cdot w_i(x) d x}{\int_{0}^{\infty} w_i(x) d x}
\end{eqnarray*}
where $\rho(x) := \frac{\phi(x)-\phi(-x)}{\phi(x)+\phi(-x)}$ and $w_i(x) = m_i(x) \cdot [\phi(x)+\phi(-x)]$. 

Using the PDF of a normal distribution, we can show that $\rho(x)=tanh\left(\frac{\Theta x}{\Sigma^2}\right)$ is strictly increasing in $x$ as $\Theta>0$, which further indicates that $x\cdot \rho(x)$ strictly increases in $x$. Hence, one sufficient condition for $K(m_1) > K(m_2)$ is that the corresponding weight functions are such that $\frac{w_1(x)}{w_2(x)}$ is weakly increasing and not constant in $x$ for $x \in \mathcal{X}$, i.e., the weight function $w_1(x)$ places relatively more weight on the larger values of $x$, which is associated with larger values of $x\cdot\rho(x)$. 

Finally, notice that $\frac{w_1(x)}{w_2(x)} = \frac{m_1(x)}{m_2(x)}$, which concludes the proof. 
\end{proof}

We now come back to proving Proposition~\ref{prop:nohack}, which is equivalent to $K(m_n) > K(m_u)$, as suggested by the first paragraph of the proof. By Lemma~\ref{lem:key}, it suffices to show that $\frac{m_n(x)}{m_u(x)}$ weakly increases in $x$ for $x\geq 0$, which is true for $\delta > 1$. 

\subsection{Proof of Proposition \ref{prop:slow_strat}}
\label{pf:prop:slow_strat}

First, if the original result $X_1$ is already significant (i.e., $|X_1|\geq t\Sigma$), it is strictly optimal not to p-hack. Second, if the original result is not significant, and the researcher decides to undertake slow p-hacking, then the optimal reported result is $(\xi(X_1), |\xi(X_1)|/t)$. The question boils down to whether it is worthwhile to undertake slow p-hacking, i.e., whether
\begin{equation}
V \left(1-\frac{1}{\delta}\right) > \kappa_s \cdot C(X_1), \label{eqn:Xs}
\end{equation}
where the left-hand side (LHS) is the benefit from having a significant result, and the right-hand side (RHS) is the minimal cost of doing so. Notice that the RHS is symmetric in $X_1$ around zero and strictly decreasing in $|X_1|$ for $|X_1|\in \left[0, t \Sigma\right)$. Moreover, the inequality fails when $X_1=0$ according to Assumption~\ref{assump:largekappas} and holds when $X_1=t\Sigma$ since $C(t\Sigma)=0$. By continuity, there exists a threshold $\overline{X}_s \in (0,t\Sigma]$ such that the inequality holds if and only if $|X_1|>\overline{X}_s$. 

Furthermore, the LHS is strictly increasing in $\delta$ while the RHS does not depend on $\delta$. Thus, the indifference threshold $\overline{X}_s$ is strictly decreasing in $\delta$. As $\delta$ goes to one, the LHS goes to zero, and $\overline{X}_s$ goes to $t\Sigma$.

\subsection{Proof of Theorem \ref{thm:slow}}
\label{pf:thm:slow}

Throughout Section~\ref{section:twosided}, we categorize published results by the value of the original estimate (i.e., the value of $X_n$ chosen by the researcher in the timeline). Category (a) consists of published results whose original estimate satisfies $|X_n|\geq t\Sigma$; category (b) consists of published results whose original estimate satisfies $|X_n|\in (\overline{X}_s, t\Sigma)$; category (c) consists of published results whose original estimate satisfies $|X_n| \leq \overline{X}_s$. Based on this categorization, we can express the expected published results of the four scenarios as follows. 
\begin{eqnarray}
\gamma(\delta) &=& e_a q_a + e_b q_b + e_c q_c = e_a q_a + e_{bc} q_{bc} \label{eq:nohack} \\
\gamma_s(\delta) &=& e_a^s q_a^s + e_b^s q_b^s + e_c^s q_c^s \label{eq:slow} \\
\gamma_f(\delta) &=& e_a^f q_a^f + e_b^f q_b^f + e_c^f q_c^f = e_a^f q_a^f + e_{bc}^f q_{bc}^f \label{eq:fast} \\
\gamma_h(\delta) &=& e_a^h q_a^h + e_b^h q_b^h + e_c^h q_c^h. \label{eq:both}
\end{eqnarray}
As an illustration, in (\ref{eq:nohack}), $e_a$, $e_b$, and $e_c$ are the expected published results of the three categories when p-hacking is not allowed; $q_a$, $q_b$, and $q_c$ represent the proportion of each category among all published results when p-hacking is not allowed. Similar definitions apply to the other three scenarios. Notice that in the scenarios with fast-only p-hacking and no p-hacking, categories (b) and (c) are not distinguished, so we can also combine these two categories, as indicated in (\ref{eq:nohack}) and (\ref{eq:fast}). 

\paragraph{\underline{Part 1}: p-hacking mitigates publication bias for $\delta > \overline{\delta}_s$.} For the scenario with no p-hacking, the reported result is $X_1$. Following the notation in (\ref{eq:nohack}), we know that $e_a:=\mathbbm{E}(X_1 \,|\, |X_1|\geq t \Sigma)$, $e_b:=\mathbbm{E}(X_1 \,|\, |X_1|\in (\overline{X}_s, t \Sigma))$, and $e_c:=\mathbbm{E}(X_1 \,|\, |X_1| \in [0, \overline{X}_s])$; meanwhile, letting $p_a:=\mathbbm{Pr}(|X_1|\geq t \Sigma)$, $p_b:=\mathbbm{Pr}(|X_1|\in (\overline{X}_s, t \Sigma))$, and $p_c:=\mathbbm{Pr}(|X_1| \in [0, \overline{X}_s])$, then we have $q_a = \frac{p_a}{p_a + p_b/\delta + p_c/\delta}$, $q_b = \frac{p_b/\delta}{p_a + p_b/\delta + p_c/\delta}$, and $q_c = \frac{p_c/\delta}{p_a + p_b/\delta + p_c/\delta}$. Hence, $\gamma(\delta)$ can be rewritten as follows: 
\begin{equation}
\gamma(\delta) = \frac{e_a p_a + e_b p_b/\delta + e_c p_c/\delta}{p_a + p_b/\delta + p_c/\delta}\,. \label{eq:a2}
\end{equation}
For the scenario with only slow p-hacking, we can infer that $e_a^s=e_a$ and $e_c^s = e_c$, as these results are not subject to slow p-hacking; by contrast, we have $e_b^s =\mathbbm{E}(\xi(X_1) \,|\, |X_1|\in (\overline{X}_s, t \Sigma))$. Moreover, the probabilities of the three categories are $q_a^s = \frac{p_a}{p_a + p_b + p_c/\delta}$, $q_b^s = \frac{p_b}{p_a + p_b + p_c/\delta}$, and $q_c^s = \frac{p_c/\delta}{p_a + p_b + p_c/\delta}$. Hence, following (\ref{eq:slow}), we can rewrite $\gamma_s(\delta)$ as follows:
\begin{equation}
\gamma_s(\delta) = \frac{e_a p_a + e_b^s p_b + e_c p_c/\delta}{p_a + p_b + p_c/\delta}\,.\label{eq:a3}
\end{equation}

We begin with two lemmas. %
\begin{lemma}
\label{lemma:slowhack}
The function $\xi(x)$ is (i) anti-symmetric around zero, (ii) satisfies $x < \xi(x) < t\Sigma$, $\forall x\in(\overline{X}_s,t\Sigma)$, (iii) and strictly increases in $x$ for $x\in(\overline{X}_s, t\Sigma)$.  
\end{lemma}
\begin{proof}
(i) The anti-symmetry of $\xi(x)$ stems from the symmetry of the publication rule and the slow p-hacking cost function.\\ 
(ii) By contradiction, if $\xi(x) > t\Sigma$ for some $x\in(\overline{X}_s,t\Sigma)$, the researcher's reported result $(\xi(x), |\xi(x)|/t)$ must have incurred a higher cost than reporting $(t\Sigma, \Sigma)$, which is already significant. Similarly, if $\xi(x) < x$, then the researcher's reported result $(\xi(x), |\xi(x)|/t)$ is strictly dominated by $(x, |\xi(x)|/t)$. 

We have left to establish that $\xi(x) \neq x$ and $\xi(x) \neq t\Sigma$.  From the first-order condition, $\xi(x)$ must satisfy 
\begin{equation}
c_1\left(\xi(x)-x, \Sigma - \frac{\xi(x)}{t}\right) - \frac{1}{t}c_2\left(\xi(x)-x, \Sigma - \frac{\xi(x)}{t}\right) = 0. \label{eq:a8}
\end{equation}
At $\xi(x)=x$, the first term is $c_1(0,\Sigma-x/t) = c_1(0,0) = 0$ since $c_{12}=0$. However, the second term is $c_2(0,\Sigma-x/t)>0$ since $c_{22}>0$. Thus, the FOC cannot hold and $\xi(x) \neq x$. A similar argument shows that $\xi(x) \neq t\Sigma$.\\
(iii) We apply the implicit function theorem to (\ref{eq:a8}) and get 
\[
c_{11} \cdot (\xi'(x) - 1) - \frac{c_{12}}{t} \cdot \xi'(x) - \frac{1}{t} \left[ c_{21} \cdot (\xi'(x) - 1) - \frac{c_{22}}{t} \cdot \xi'(x) \right] = 0.
\]
Rearranging the terms to isolate $\xi'(x)$ yields
\begin{equation}
\xi'(x) = \frac{c_{11} - c_{12}/t}{c_{11} - 2c_{12}/t + c_{22}/t^2}. \label{eq:a9}
\end{equation}
By assumption, we know that $c_{11}>0$, $c_{22}>0$, and $c_{12}=0$. Therefore, we know from (\ref{eq:a9}) that $\xi'(x) > 0$. This concludes the proof. 

\end{proof}

\begin{lemma}
\label{lem:keyinequality}
Let $\mathcal{X}\subseteq [0,\infty)$ with $\underline{x}:=\inf\{\mathcal{X}\}$ and $\overline{x}:=\sup\{\mathcal{X}\}$ being its exact lower and upper bounds. Suppose a function $f(x)$ satisfies $0 < M_1\leq f(x) \leq M_2$, $\forall x\in \mathcal{X}$, for some constants $M_1, M_2$, then we have
\begin{equation}
    M_1 \cdot \frac{\phi(\underline{x})-\phi(-\underline{x})}{\phi(\underline{x})+\phi(-\underline{x})}\leq\frac{\int_{x\in\mathcal{X}} f(x) [\phi(x)-\phi(-x)]dx}{\int_{x\in\mathcal{X}}[\phi(x)+\phi(-x)]dx} \leq M_2 \cdot \frac{\phi(\overline{x})-\phi(-\overline{x})}{\phi(\overline{x})+\phi(-\overline{x})}. \label{eq:keyinequality}
\end{equation}
The inequalities are strict if $\mathcal{X}$ is non-degenerate. 
\end{lemma}
\begin{proof}
Using the PDF of normal distributions, we know that $\rho(x):=\frac{\phi(x)-\phi(-x)}{\phi(x)+\phi(-x)}=tanh\left(\frac{\Theta x}{\Sigma^2}\right)$ is strictly increasing in $x$ when $\Theta>0$. Hence, the first inequality holds because 
\begin{eqnarray}
\frac{\int_{x\in\mathcal{X}} f(x) [\phi(x)-\phi(-x)]dx}{\int_{x\in\mathcal{X}}[\phi(x)+\phi(-x)]dx} &\geq & \frac{\int_{x\in\mathcal{X}} M_1 [\phi(x)-\phi(-x)]dx}{\int_{x\in\mathcal{X}}[\phi(x)+\phi(-x)]dx} \nonumber \\
& = & M_1 \cdot \frac{\int_{x\in\mathcal{X}} \rho(x) [\phi(x)+\phi(-x)]dx}{\int_{x\in\mathcal{X}}[\phi(x)+\phi(-x)]dx} \nonumber \\
&\geq& M_1 \cdot \rho(\underline{x}) \nonumber \\
&=& M_1 \cdot \frac{\phi(\underline{x})-\phi(-\underline{x})}{\phi(\underline{x})+\phi(-\underline{x})}. \nonumber
\end{eqnarray}
The first inequality holds because $f(x)\geq M_1$ and $\phi(x)>\phi(-x)$ for any $x\in\mathcal{X}$. The first equality uses the definition of $\rho(x)$. The second inequality holds because $\rho(x)$ strictly increases in $x$ while the term $\frac{\int_{x\in\mathcal{X}} \rho(x) [\phi(x)+\phi(-x)]dx}{\int_{x\in\mathcal{X}}[\phi(x)+\phi(-x)]dx}$ can be regarded as a weighted average of $\rho(x)$ across $x\in \mathcal{X}$; moreover, it holds strictly if $\mathcal{X}$ is non-degenerate. The second equality uses the definition of $\rho(\underline{x})$. We thus proved the first inequality in (\ref{eq:keyinequality}); the second inequality follows by similar reasoning. 
\end{proof}

Now, we are ready to show Part 1 of the theorem. Since $\gamma_s(\cdot)$ and $\gamma(\cdot)$ are continuous, it suffices to show that 
\begin{equation}
\lim\limits_{\delta\rightarrow \infty} \gamma(\delta) > \lim\limits_{\delta\rightarrow\infty}\gamma_s(\delta) > \Theta. \label{eqn1}
\end{equation}
The first inequality of (\ref{eqn1}) boils down to $e_a > e_b^s$ when $\delta\rightarrow\infty$, which, by the anti-symmetry of $\xi(x)$ as indicated by Lemma~\ref{lemma:slowhack}, is equivalent to
\begin{eqnarray}
\frac{\int_{t\Sigma}^{\infty} x [\phi(x)-\phi(-x)]dx}{\int_{t\Sigma}^{\infty}[\phi(x)+\phi(-x)]dx} > \frac{\int_{\overline{X}_s}^{t\Sigma} \xi(x) [\phi(x)-\phi(-x)]dx}{\int_{\overline{X}_s}^{t\Sigma}[\phi(x)+\phi(-x)]dx}\,. \nonumber %
\end{eqnarray}
Using Lemma~\ref{lem:keyinequality} twice, together with the finding that $\xi(x)< t\Sigma$ from Lemma~\ref{lemma:slowhack}, we have 
\begin{eqnarray}
\frac{\int_{t\Sigma}^{\infty} x [\phi(x)-\phi(-x)]dx}{\int_{t\Sigma}^{\infty}[\phi(x)+\phi(-x)]dx} \geq t\Sigma \cdot \frac{\phi(t\Sigma)-\phi(-t\Sigma)}{\phi(t\Sigma)+\phi(-t\Sigma)} > \frac{\int_{\overline{X}_s}^{t\Sigma} \xi(x) [\phi(x)-\phi(-x)]dx}{\int_{\overline{X}_s}^{t\Sigma}[\phi(x)+\phi(-x)]dx}. \nonumber 
\end{eqnarray}

The second inequality of (\ref{eqn1}) holds because 
\begin{equation}
\lim\limits_{\delta\rightarrow\infty}\gamma_s(\delta) = \frac{e_a p_a + e_b^s p_b}{p_a + p_b} > \frac{e_a p_a + e_b p_b}{p_a + p_b}  > \Theta, \nonumber
\end{equation}
where the first inequality holds because $\xi(x)>x$ according to Lemma~\ref{lemma:slowhack}, and the second inequality follows directly from Proposition~\ref{prop:nohack}, which suggests that selective publication introduces a bias. 

\paragraph{\underline{Part 2}: p-hacking exacerbates publication bias for $\delta < \underline{\delta}_s$.} From (\ref{eq:a2}) and (\ref{eq:a3}), we have 
$\gamma_s(\delta)\cdot(p_a+p_b+p_c/\delta) - \gamma(\delta)\cdot(p_a+p_b/\delta+p_c/\delta) = (e_b^s-e_b/\delta)p_b$. Thus, 
$$\gamma_s(\delta)-\gamma(\delta) = \frac{p_b}{p_a+p_b+p_c/\delta} \cdot \Lambda(\delta)\,,$$
where $\Lambda(\delta):=(e_b^s-e_b/\delta) - (1-1/\delta)\gamma(\delta)$. Since $\frac{p_b}{p_a+p_b+p_c/\delta}$ is strictly positive, we want to show that $\Lambda(\delta) > 0$ when $\delta$ is sufficiently close to $1$. Further, since $\lim\limits_{\delta\rightarrow 1}\Lambda(\delta)=0$, as p-hacking becomes meaningless when selective publication vanishes, it suffices to show that $\lim\limits_{\delta\rightarrow 1}\Lambda'(\delta) > 0$, by continuity of $\Lambda(\cdot)$. Equivalently, we let $\lambda:=1/\delta \in (0,1)$ and rewrite $\Lambda(\delta)$ as $\tilde{\Lambda}(\lambda) := (e_b^s- \lambda e_b) - (1-\lambda) \tilde{\gamma}(\lambda)$ where $\tilde{\gamma}(\lambda):= \frac{e_a p_a + \lambda e_b p_b + \lambda e_c p_c}{p_a + \lambda p_b +\lambda p_c}$. We want to show that $\lim\limits_{\lambda\rightarrow 1}\tilde{\Lambda}'(\lambda) < 0$. Notice that 
\begin{eqnarray}
\lim\limits_{\lambda\rightarrow 1}\tilde{\Lambda}'(\lambda) &=& \lim\limits_{\lambda\rightarrow 1}\frac{\partial e_b^s}{\partial \lambda} - \lim\limits_{\lambda\rightarrow 1} \left(\lambda \frac{\partial e_b}{\partial \lambda}\right)  - \lim\limits_{\lambda\rightarrow 1} e_b + \lim\limits_{\lambda\rightarrow 1} \tilde{\gamma}(\lambda) - \lim\limits_{\lambda\rightarrow 1} (1-\lambda)\tilde{\gamma}'(\lambda) \nonumber \\
&=& \lim\limits_{\lambda\rightarrow 1}\frac{\partial (e_b^s-e_b)}{\partial \lambda} - \frac{\phi(t\Sigma) - \phi(-t\Sigma)}{\phi(t\Sigma) + \phi(-t\Sigma)} \cdot t\Sigma + \Theta - 0. \label{eqn4} 
\end{eqnarray}
\begin{lemma}
\label{lemma:tech}
We have $\lim\limits_{\lambda\rightarrow 1}\frac{\partial (e_b^s-e_b)}{\partial \lambda}=-\infty$. 
\end{lemma}
\begin{proof}
Notice that $e_b^s-e_b = \frac{\int_{\overline{X}_s}^{t\Sigma}(\xi(x)-x)(\phi(x)-\phi(-x))dx}{\int_{\overline{X}_s}^{t\Sigma}(\phi(x)+\phi(-x))dx}$. We let $N(\lambda)$ denote its numerator and ${D(\lambda)}$ its denominator, so 
\begin{eqnarray}
\frac{\partial (e_b^s-e_b)}{\partial \lambda}  
&=& \frac{N'(\lambda)}{D(\lambda)} - \frac{N(\lambda)\cdot D'(\lambda)}{[D(\lambda)]^2} \nonumber \\
&=& \frac{\partial \overline{X}_s}{\partial \lambda} \cdot \left\{\frac{[\overline{X}_s-\xi(\overline{X}_s)][\phi(\overline{X}_s)-\phi(-\overline{X}_s)]}{\int_{\overline{X}_s}^{t\Sigma}(\phi(x)+\phi(-x))dx} \right. \nonumber \\ 
&& \hspace{2cm}\left. + \frac{[\phi(\overline{X}_s)+\phi(-\overline{X}_s)]\int_{\overline{X}_s}^{t\Sigma}(\xi(x)-x)(\phi(x)-\phi(-x))dx}{[\int_{\overline{X}_s}^{t\Sigma}(\phi(x)+\phi(-x))dx]^2}  \right\}. \label{eqn5}
\end{eqnarray}
Since $\overline{X}_s\rightarrow t\Sigma$ as $\lambda\rightarrow 1$, we apply the L'Hospital rule to the terms in the curly bracket of (\ref{eqn5}) when $\overline{X}_s\rightarrow t\Sigma$. It ends up being equal to
\begin{eqnarray}
&& \frac{[\phi(t\Sigma)-\phi(-t\Sigma)]\cdot \left(\lim\limits_{x\rightarrow t\Sigma}\xi'(x) - 1\right)}{\phi(t\Sigma)+\phi(-t\Sigma)} + \frac{[\phi(t\Sigma)-\phi(-t\Sigma)]\cdot \left(\lim\limits_{x\rightarrow t\Sigma}\xi'(x) - 1\right)}{-2[\phi(t\Sigma)+\phi(-t\Sigma)]} \nonumber \\
 &=& \lim\limits_{x\rightarrow t\Sigma}(\xi'(x) - 1) \cdot \frac{\phi(t\Sigma) - \phi(-t\Sigma)}{2[\phi(t\Sigma)+\phi(-t\Sigma)]}. \nonumber
\end{eqnarray}
Hence, we have 
\begin{equation}
\lim\limits_{\lambda\rightarrow 1}\frac{\partial (e_b^s-e_b)}{\partial \lambda} = \lim\limits_{\lambda\rightarrow 1} \frac{\partial \overline{X}_s}{\partial \lambda} \cdot \left(\lim\limits_{x\rightarrow t\Sigma}\xi'(x) - 1\right) \cdot \frac{\phi(t\Sigma) - \phi(-t\Sigma)}{2[\phi(t\Sigma)+\phi(-t\Sigma)]}. \label{eqn6}
\end{equation}
Next, we show that the first term on the RHS of (\ref{eqn6}) equals $+\infty$. From Proposition~\ref{prop:slow_strat}, we know that 
\begin{equation}
V \left(1-\lambda \right) = \kappa_s \cdot c(\xi(\overline{X}_s) - \overline{X}_s, \Sigma - \xi(\overline{X}_s)/t ). \label{eqn7}
\end{equation}
Applying the implicit function theorem and the envelope theorem to (\ref{eqn7}), we have 
\begin{equation}
V = \kappa_s \cdot c_1(\xi(\overline{X}_s) - \overline{X}_s, \Sigma - \xi(\overline{X}_s)/t ) \cdot \frac{\partial \overline{X}_s}{\partial \lambda}. \nonumber
\end{equation}
As $\lambda\rightarrow 1$, we have $c_1(\xi(\overline{X}_s) - \overline{X}_s, \Sigma - \xi(\overline{X}_s)/t )\rightarrow c_1(0,0) = 0$, and therefore, 
\begin{equation}
\lim\limits_{\lambda\rightarrow 1}\frac{\partial \overline{X}_s}{\partial \lambda} = \frac{V}{\kappa_s \cdot c_1(0,0)} = +\infty. \label{eqn8}
\end{equation}
We can also show that the second term on the RHS of (\ref{eqn6}) is negative. From (\ref{eq:a9}), we have 
\begin{equation}
\xi'(x) = \frac{c_{11}-c_{12}/t}{c_{11}-2c_{12}/t+c_{22}/(t^2)} = \frac{c_{11}}{c_{11}+c_{22}/(t^2)} < 1, \label{eqn9}
\end{equation}
as $c_{11}>0$, $c_{12}=0$, and $c_{22}>0$.

Combining (\ref{eqn8}), (\ref{eqn9}), and the fact that the third term on the RHS of (\ref{eqn6}) is a positive constant, we finally conclude that $\lim\limits_{\lambda\rightarrow 1}\frac{\partial (e_b^s-e_b)}{\partial \lambda} = -\infty$. 
\end{proof}

Applying Lemma~\ref{lemma:tech} to (\ref{eqn4}), we have $\lim\limits_{\lambda\rightarrow 1}\tilde{\Lambda}'(\lambda) =-\infty <0$. This completes the proof of Part 2 of the theorem. 

\subsection{Proof of Proposition \ref{prop:fast_strat}}
\label{pf:prop:fast_strat}

It is straightforward that the researcher optimally stops drawing new estimates once he obtains a significant result. If all his current results are insignificant, his decision of whether to draw a new estimate depends on the cost of redrawing. By the definition of $n^*_f$, the researcher will stop drawing new estimates after $n^*_f$ insignificant results, as the gain from drawing an additional estimate cannot offset the cost for any future draw. However, he will still find it optimal to draw a new estimate after $n^*_f-1$ insignificant results, as the gain from one additional draw is higher than the cost. By backward induction, we can infer that he is also willing to draw a new estimate if the current number of insignificant results is smaller than $n^*_f-1$. Finally, the argument that $n^*_f$ weakly increases in $\delta$ follows directly from the fact that $\Delta_f(\delta)$ strictly increases in $\delta$. 

\subsection{Proof of Theorem \ref{thm:fast}}
\label{pf:thm:fast}

Let $\overline{\delta}_f:= \Delta_f^{-1}(\kappa_f c_2)$, we know that $n^*_f = 1$  for $\delta < \overline{\delta}_f$. In other words, it is not worthwhile to undertake fast p-hacking if the extent of selective publication is low. Therefore, the opportunity of fast p-hacking does not affect publication bias if $\delta < \overline{\delta}_f$. 

Next, we show that fast p-hacking, if ever performed (i.e., $n^*_f \geq 2$), exacerbates publication bias. Let $p_{\text{0}}:=\text{Pr}_{X\sim\mathcal{N}(\Theta,\Sigma^2)}(|X| < t\Sigma)$ be the probability of an estimate being insignificant. Following the notation in (\ref{eq:fast}), we infer that $e_a^f = e_a$, since the researcher will optimally stop drawing new estimates once a significant one is obtained. Fix $\delta$ and suppose the researcher draws up to $n$ estimates. We know that $q_a^f = \frac{1-(p_{0})^{n}}{(p_{0})^{n} / \delta + 1-(p_{0})^{n}}$ and $q_{bc}^f = 1-q_a^f$, and therefore, the expected published result is $\hat{\gamma}_f(\delta, n) = \frac{e_{bc}^f \cdot (p_{0})^{n} / \delta + e_a \cdot (1-(p_{0})^{n})}{(p_{0})^{n} / \delta + 1-(p_{0})^{n}}$. We also know that $\gamma_f(\delta) = \hat{\gamma}_f(\delta, n^*_f)$ and $\gamma(\delta) = \hat{\gamma}_f(\delta, 1)$. Hence, it suffices to show that $\hat{\gamma}_f(\delta, n)$ increases in $n$ for any $\delta$. Notice that $e_a$ is constant in $n$. Thus, it suffices to show that: (i) $e_a > e_{bc}^f$ for any $n$; (ii) $e_{bc}^f$ strictly increases in $n$. 

We can use Lemma~\ref{lem:key} to show $e_a > e_{bc}^f$ for any $n$. Specifically, $e_a = K(m_a)$ with the kernel function $m_a(x) = \mathbbm{1}(|x|\geq t\Sigma)$ and $e_{bc}^f = K(m_b)$ with the kernel function $m_b(x)$ only placing weight on the insignificant results. Hence, $\frac{m_a(x)}{m_b(x)}$ is a weakly increasing step function that switches from zero to infinity as $x$ crosses $t\Sigma$. 

We can also use Lemma~\ref{lem:key} to show that $e_{bc}^f$ strictly increases in $n$. Since the researcher reports the result with the lowest p-value,  if $n>n'$, the distribution of the absolute value of the reported result with $n$ first-order stochastically dominates that with $n'$, i.e., $\frac{m_b(x, n)}{m_b(x, n')}$ strictly increases in $x \geq 0$.\footnote{\label{footnote:tiebreak1}In fact, Theorem \ref{thm:fast} also holds when the researcher adopts an alternative tie-breaking rule when facing multiple insignificant estimates, such as (i) reporting an estimate drawn uniformly at random, (ii) reporting the initial estimate $X_1$, or (iii) reporting the latest estimate $X_{n^*_f}$. Under these alternative tie-breaking rules, we have $e_{bc}^f = e_{bc}$ for any $n$, so Theorem~\ref{thm:fast} remains true.} 

Finally, we have $\lim\limits_{\delta\rightarrow\infty}\gamma_f(\delta) = \lim\limits_{\delta\rightarrow\infty}\gamma(\delta) = e_a$, which proves the last argument of the theorem. 

\subsection{Proof of Proposition \ref{prop:combined_strat}}
\label{pf:prop:combined_strat}

We characterize the optimal p-hacking strategy by backward induction. 

\paragraph{Step 1: Optimal slow p-hacking.} 

When the researcher decides to undertake slow p-hacking, he must have stopped drawing new estimates. Hence, his slow p-hacking decision is identical to that in Proposition~\ref{prop:slow_strat}. Suppose he has already drawn $k$ estimates, he slow p-hacks if and only if the current best estimate $|X^*_{k}| \in (\overline{X}_s, t\Sigma)$. Based on his slow p-hacking strategy, we define his continuation value upon stopping drawing new estimates as 
\[U_{s}(x) = \frac{V}{\delta}+V\left(1-\frac{1}{\delta}\right)\mathbbm{1}\{|x| >\overline{X}_s\}- \kappa_s C(x)\mathbbm{1}\{\overline{X}_s< |x| < t\Sigma\}\]
when the value of his best current estimate is $x$. 

\paragraph{Step 2: Optimal fast p-hacking.} 

~\\
\noindent \underline{Step 2.1: Existence of $n^*$.} Next, we work backward to determine the researcher's optimal fast p-hacking strategy. To begin with, there is an upper bound on the number of estimates that the researcher is willing to draw. This is because the cost of drawing a new estimate is eventually strictly greater than $V$, and thus strictly greater than any benefit the researcher might get. Let $n^*$ be the maximal number of estimates that the researcher ever draws under the optimal p-hacking strategy. We will analyze the researcher's fast p-hacking strategy given $n^*$, and come back to determine the unique value of $n^*$ at the end of the proof. 

We introduce two functions. Let $U_n(x)$ denote the researcher's continuation value under his optimal strategy when he has drawn $n$ estimates, with the best current estimate being $x$. By the definition of $n^*$, we have $U_{n^*}(x) = U_s(x)$, as the researcher stops drawing new estimates when he has already made $n^*$ draws. Let $V_n(x)$ denote the researcher's continuation value if he has drawn $n$ estimates, with the best current estimate being $x$ \textit{and} decides to draw one more estimate. By definition, we have 
\begin{equation}
V_n(x) = \mathbb{E}_{X_{n+1}\sim\mathcal{N}(0,\Sigma^2)}[U_{n+1}(\max\{|x|, |X_{n+1}|\})] - \kappa_f c_{n+1}. \nonumber
\end{equation}
Also, since the researcher's continuation value of not drawing a new estimate is $U_s(x)$, we know that 
\begin{equation}
U_n(x) = \max\{U_s(x), V_n(x)\}. \nonumber
\end{equation}
Moreover, we let $\Delta_n(x):=\left[V_n(x)+\kappa_f  c_{n+1} \right] - U_s(x)$ denote the benefit of drawing one more estimate (without taking into account the cost). Note that $V_n(x) > U_s(x)$ is equivalent to $\Delta_n(x) > \kappa_f  c_{n+1}$. Also, from the symmetry of the publication rule and the researcher's cost functions, we can infer that $U_n(x)$, $V_n(x)$, and $\Delta_n(x)$ are all symmetric around zero. It remains to show two statements: \\
(a) For any $n\leq n^*-1$, $\exists \overline{X}_{f}^{n} \in (\overline{X}_s, t\Sigma)$ such that $V_n(x) > U_s(x)$ if and only if $|x|<\overline{X}_{f}^{n}$. \\
(b) For $n\leq  n^*-1$, the threshold $\overline{X}_{f}^{n}$ strictly decreases in $n$. 

\noindent \underline{Step 2.2: Optimal fast p-hacking for $n = n^*-1$.} 

When $n=n^*-1$, the researcher draws another estimate if and only if $V_{n^*-1}(x) > U_s(x)$, which is equivalent to $\Delta_{n^*-1}(x) > \kappa_f c_{n^*}$. The closed form of $\Delta_{n^*-1}(x)$ is
\begin{align*}
\Delta_{n^*-1}(x) = \begin{cases} 0 \quad&\text{if } |x| \geq t\Sigma\\
\kappa_s \int_{|x|}^{\infty}[C(x)-C(\tilde{x})] \left[\frac{1}{\Sigma}\phi\left(\frac{\tilde{x}}{\Sigma}\right)+\frac{1}{\Sigma}\phi\left(\frac{-\tilde{x}}{\Sigma}\right)\right]d\tilde{x} \quad&\text{if } |x| \in (\overline{X}_s, t\Sigma)\\
\int_{\overline{X}_s}^{\infty} \left[V\left(1-\frac{1}{\delta}\right)-\kappa_s C(\tilde{x})\right] \left[\frac{1}{\Sigma}\phi\left(\frac{\tilde{x}}{\Sigma}\right)+\frac{1}{\Sigma}\phi\left(\frac{-\tilde{x}}{\Sigma}\right)\right]d\tilde{x}\quad&\text{if } |x| \leq\overline{X}_s.
\end{cases}
\end{align*}
From its closed form, we can infer that: $\Delta_{n^*-1}(x)$ is constant for $x\in [0,\overline{X}_s]$, strictly decreases for $x\in (\overline{X}_s, t\Sigma)$, and becomes constant again at zero for $x\geq t\Sigma$. Since $\Delta_{n^*-1}(t\Sigma) = 0 < \kappa_f  c_{n^*} < \Delta_{n^*-1}(\overline{X}_s)$, where the last inequality holds as the researcher may draw one more estimate with $n=n^*-1$, we know that there exists $\overline{X}_{f}^{n^*-1} \in (\overline{X}_s, t\Sigma)$ such that $V_{n^*-1}(x) > U_s(x)$ if and only if $|x|<\overline{X}_{f}^{n^*-1}$. Accordingly, 
\begin{align*}
U_{n^*-1}(x) = \begin{cases} U_s(x) \quad&\text{if } |x| \geq \overline{X}_{f}^{n^*-1}\\
V_{n^*-1}(x) \quad&\text{if } |x| < \overline{X}_{f}^{n^*-1}.
\end{cases}
\end{align*}

\noindent \underline{Step 2.3: Optimal fast p-hacking for $n<n^*-1$.} 

We now use induction to prove the two statements identified in Step 2.1 for any $n<n^*-1$. 

\begin{lemma}
\label{lemma:delta}
Given $n<n^*-1$. Suppose $V_{n+1}(x)$, $\Delta_{n+1}(x)$, and $c_{n+2}$ satisfy the following. \\
(i) $V_{n+1}(x)$ is constant for $x\in[0, \overline{X}_s]$, strictly increases for $x\in (\overline{X}_s, t\Sigma)$, and becomes constant again for $x\geq t\Sigma$; \\
(ii)  $\Delta_{n+1}(x)$ is constant for $x\in[0, \overline{X}_s]$, strictly decreases for $x\in (\overline{X}_s, t\Sigma)$, and becomes constant again at zero for $x\geq t\Sigma$; \\
(iii) $\Delta_{n+1}(\overline{X}_s) > \kappa_f c_{n+2}$.\\
Then $V_{n}(x)$, $\Delta_n(x)$, and $c_{n+1}$ also satisfy these conditions. 
\end{lemma}
\begin{proof}
\textit{\underline{Part (i).}} Since both $V_{n+1}(x)$ and $U_s(x)$ satisfy condition (i), we know that $U_{n+1}(x) = \max\{U_s(x),V_{n+1}(x) \}$ also satisfies condition (i). As $V_n(x) = \mathbb{E}_{X_{n+1}\sim\mathcal{N}(0,\Sigma^2)}[U_{n+1}(\max\{|x|, |X_{n+1}|\})] -\kappa_f  c_{n+1}$, we can further infer that $V_n(x)$ is also constant for $x\in[0, \overline{X}_s]$, strictly increases for $x\in (\overline{X}_s, t\Sigma)$, and becomes constant again for $x\geq t\Sigma$, proving condition (i) for $V_n(x)$. 

\noindent \textit{\underline{Part (ii).}}  The above result, together with the fact that $U_s(x)$ satisfies condition (i), immediately indicates that $\Delta_n(x) = V_n(x) - U_s(x) +\kappa_f  c_{n+1}$ is constant for $x\in[0, \overline{X}_s]$. We also know that $\Delta_n(x) = 0$ for $x \geq t\Sigma$, as the researcher gains nothing from one more draw if his best current estimate is already significant. It remains to show that $\Delta_n(x)$ strictly decreases for $x\in (\overline{X}_s, t\Sigma)$. To show this, let $\overline{X}_s < x' < x'' < t\Sigma$, and it suffices to show $\Delta_n(x') > \Delta_n(x'')$, which is equivalent to 
\[ \mathbb{E}_{X_{n+1}\sim\mathcal{N}(0,\Sigma^2)}[U_{n+1}(\max\{|x'|, |X_{n+1}|\})] - U_s(x') > \mathbb{E}_{X_{n+1}\sim\mathcal{N}(0,\Sigma^2)}[U_{n+1}(\max\{|x''|, |X_{n+1}|\})] - U_s(x'') \]
by definition, which, with terms rearranged, is 
\begin{equation}
\mathbb{E}_{X_{n+1}\sim\mathcal{N}(0,\Sigma^2)}[U_{n+1}(\max\{|x''|, |X_{n+1}|\}) - U_{n+1}(\max\{|x'|, |X_{n+1}|\})] < U_s(x'') - U_s(x'). \label{eq:a3}
\end{equation}
Notice that $U_{n+1}(\max\{|x''|, |X_{n+1}|\}) - U_{n+1}(\max\{|x'|, |X_{n+1}|\}) \leq U_{n+1}(|x''|) - U_{n+1}(|x'|)$, with the inequality being strict if $|X_{n+1}| \in (x', x'')$. Hence, the LHS of (\ref{eq:a3}) satisfies 
\begin{equation}
\mathbb{E}_{X_{n+1}\sim\mathcal{N}(0,\Sigma^2)}[U_{n+1}(\max\{|x''|, |X_{n+1}|\}) - U_{n+1}(\max\{|x'|, |X_{n+1}|\})] < U_{n+1}(x'') - U_{n+1}(x'). \label{eq:a4}
\end{equation}
Meanwhile, we have 
\begin{align*}
U_{n+1}(x) - U_s(x) = \begin{cases} 0 \quad&\text{if } |x| \geq \overline{X}_{f}^{n+1}\\
V_{n+1}(x)-  U_s(x) = \Delta_{n+1}(x) - \kappa_f c_{n+2} \quad&\text{if } |x| < \overline{X}_{f}^{n+1}.
\end{cases}
\end{align*}
Since $\Delta_{n+1}(x)$ is strictly decreasing for $x\in (\overline{X}_s, t\Sigma)$ according to the premise of this lemma, we know that $U_{n+1}(x) - U_s(x)$ is weakly decreasing for $x\in (\overline{X}_s, t\Sigma)$, and therefore, 
\[U_{n+1}(x') - U_s(x') \geq U_{n+1}(x'') - U_s(x'') \]
which is equivalent to 
\begin{equation}
U_s(x'') - U_s(x') \geq U_{n+1}(x'') - U_{n+1}(x'). \label{eq:a5}
\end{equation}
Combining (\ref{eq:a4}) and (\ref{eq:a5}), we finally show (\ref{eq:a3}) and thus conclude the proof of (ii) for $\Delta_n(x)$. 

\noindent \textit{\underline{Part (iii).}} We prove this statement by showing 
\begin{equation}
\Delta_n(x) \geq \Delta_{n+1}(x) > \kappa_f c_{n+2} > \kappa_f c_{n+1}. \nonumber
\end{equation}
The first inequality holds because 
\begin{eqnarray*}
\Delta_n(x) &=& \mathbb{E}_{X_{n+1}\sim\mathcal{N}(0,\Sigma^2)}[U_{n+1}(\max\{|x|, |X_{n+1}|\})] - U_s(x) \\
\Delta_{n+1}(x) &=& \mathbb{E}_{X_{n+2}\sim\mathcal{N}(0,\Sigma^2)}[U_{n+2}(\max\{|x|, |X_{n+2}|\})] - U_s(x),
\end{eqnarray*}
while $U_{n+1}(x) \geq U_{n+2}(x)$ for any $x$, which intuitively holds because, fixing the researcher's best current estimate, he is better off when he has drawn fewer estimates given the increasing cost of drawing new estimates. The second inequality is the premise of this lemma. The third inequality follows from the increasing cost of drawing new estimates. 
\end{proof}

Combining statements (ii) and (iii) of Lemma~\ref{lemma:delta} and the symmetry of $\Delta_n(x)$, we conclude that there exists $\overline{X}_{f}^{n} \in (\overline{X}_s, t\Sigma)$ such that $V_{n}(x) > U_s(x)$ if and only if $|x|<\overline{X}_{f}^{n}$. This proves statement (a) specified in Step 2.1 for any $n < n^*-1$. 

Finally, we prove statement (b) in Step 2.1. For any $n< n^*-1$, we have  
\begin{equation}
U_s(\overline{X}_{f}^{n}) = V_{n}(\overline{X}_{f}^{n}) > V_{n+1}(\overline{X}_{f}^{n}), \label{eq:a6}
\end{equation}
where the equality follows from the definition of $\overline{X}_{f}^{n}$ and the inequality holds because $V_{n}(x) > V_{n+1}(x)$ for any $x\in (\overline{X}_s, t\Sigma)$, as previously proved. Hence, (\ref{eq:a6}) implies that $\overline{X}_{f}^{n} > \overline{X}_{f}^{n+1}$ by the definition of $\overline{X}_{f}^{n+1}$. Therefore, we conclude that for $n\leq  n^*-1$, the threshold $\overline{X}_{f}^{n}$ strictly decreases in $n$. 

\noindent \underline{Step 2.4: Characterization of $n^*$.} 

The analysis in Steps 2.1 to 2.3 takes $n^*$ as given. We wrap up Step 2 of the proof by characterizing the value of $n^*$. Notice that the function $\Delta_{n^*-1}(x)$ indeed does not depend on the value of $n^*$. Hence, $\Delta_h:= \Delta_{n^*-1}(0)$ does not depend on $n^*$ as well, and therefore, $n^*=\max\{n\mid \kappa_f c_{n}<\Delta_h\}$ uniquely exists. This is the maximal number of draws, since a best current estimate of $0$ is most likely to induce fast p-hacking. 

\subsection{Proof of Theorem \ref{thm:combined}}
\label{pf:thm:combined}

When $\delta$ is sufficiently small, the researcher never undertakes fast p-hacking, and therefore, the result that p-hacking exacerbates publication bias immediately follows from Theorem \ref{thm:slow}. 

To show that p-hacking mitigates publication bias when $\delta$ is sufficiently large, we focus on the limit $\delta\longrightarrow \infty$. In that limit, following the notation in (\ref{eq:nohack}), we have $\lim\limits_{\delta\longrightarrow \infty} \gamma(\delta) = e_a = \mathbbm{E}_{X\sim\mathcal{N}(\Theta, \Sigma^2)}(X \,|\, |X|\geq t \Sigma)$. Under p-hacking, using the notation in (\ref{eq:both}), we can infer that $e_a^h = e_a$, as the researcher optimally stops drawing new estimates once a significant estimate is obtained. Hence, $\lim\limits_{\delta\longrightarrow \infty}\gamma_h(\delta)$ is the weighted average of $e_a$ and $e_b^h$. To show that $\lim\limits_{\delta\longrightarrow \infty} \gamma(\delta)>\lim\limits_{\delta\longrightarrow \infty} \gamma_h(\delta)$, it suffices to show that $e_b^h < e_a$.

Suppose the researcher ends up drawing $k$ estimates under his optimal p-hacking strategy in a realization of $\{X_k\}_{k\leq n^*}$, so $X_k^*$ is the original estimate that the researcher adopts for submission. We let $\Psi(x):=\text{Pr}(X_k^*\leq x)$ denote the CDF of $X_k^*$ when both forms of p-hacking are allowed, with $\psi(x)$ being its PDF. 

\begin{lemma}
\label{lemma:eta}
The function $\eta(x):=\frac{\psi(x)}{\phi(x)}$ is symmetric around zero. 
\end{lemma}
\begin{proof}
Consider $x > 0$. Let $k(x)$ satisfy $x\in [\overline{X}_f^{k(x)}, \overline{X}_f^{k(x)-1})$, then by definition, the researcher must have drawn at least $k(x)$ estimates before he is willing to use $X^*_{k(x)} = x$ for slow p-hacking and reporting. Therefore, we can infer that 
\[\psi(x) = \binom{k(x)}{1} \phi(x) \left[\Phi(x)-\Phi(-x)\right]^{k(x)-1} + \sum_{l = k(x)}^{n^*-1} \phi(x) \left[\Phi(\overline{X}_f^{l})-\Phi(-\overline{X}_f^{l})\right]^{l},\]
and similarly, 
\[\psi(-x) = \binom{k(x)}{1} \phi(-x) \left[\Phi(x)-\Phi(-x)\right]^{k(x)-1} + \sum_{l = k(x)}^{n^*-1} \phi(-x) \left[\Phi(\overline{X}_f^{l})-\Phi(-\overline{X}_f^{l})\right]^{l}.\]
We can thus confirm that $\eta(x) = \eta(-x)$. 
\end{proof}

Following similar reasoning to Lemma~\ref{lem:keyinequality}, we have 
\begin{eqnarray}
e_b^h &=& \frac{\int_{\overline{X}_s}^{t\Sigma} \xi(x) \eta(x) [\phi(x)-\phi(-x)]dx}{\int_{\overline{X}_s}^{t\Sigma} \eta(x) [\phi(x)+\phi(-x)]dx} \nonumber \\
&=& \frac{\int_{\overline{X}_s}^{t\Sigma} \xi(x) \tanh\left(\frac{\Theta x}{\Sigma^2}\right) \eta(x) [\phi(x)+\phi(-x)]dx}{\int_{\overline{X}_s}^{t\Sigma} \eta(x) [\phi(x)+\phi(-x)]dx} \nonumber \\
&\leq & t\Sigma \cdot \frac{\int_{\overline{X}_s}^{t\Sigma} \tanh\left(\frac{\Theta x}{\Sigma^2}\right) \eta(x) [\phi(x)+\phi(-x)]dx}{\int_{\overline{X}_s}^{t\Sigma} \eta(x) [\phi(x)+\phi(-x)]dx} \nonumber \\
&< & t\Sigma \cdot \tanh\left(\frac{\Theta\cdot t\Sigma}{\Sigma^2}\right), \label{eq:a7}
\end{eqnarray}
where the first equality holds by definition given the symmetry of $\eta(x)$ (see Lemma~\ref{lemma:eta}) and the anti-symmetry of $\xi(x)$ (see Lemma~\ref{lemma:slowhack}), the second equality exploits the PDF of a normal distribution, the first inequality uses Lemma~\ref{lemma:slowhack}, and the second inequality is based on the fact that $\tanh\left(\frac{\Theta x}{\Sigma^2}\right) = \frac{\phi(x)-\phi(-x)}{\phi(x)+\phi(-x)}$ strictly increases in $x$ given that $\Theta>0$. Meanwhile, from Lemma~\ref{lem:keyinequality}, we already know that $e_a > t\Sigma \cdot \tanh\left(\frac{\Theta\cdot t\Sigma}{\Sigma^2}\right)$. Combining this with (\ref{eq:a7}), we can conclude that $e_b^h < e_a$, and thus $\lim\limits_{\delta\longrightarrow \infty} \gamma(\delta)>\lim\limits_{\delta\longrightarrow \infty} \gamma_h(\delta)$. 

To complete the proof, we also need to show that $\lim\limits_{\delta\longrightarrow \infty} \gamma_h(\delta) \geq \Theta$. This is true because we have $\lim\limits_{\delta\longrightarrow \infty} \gamma_h(\delta) \geq \lim\limits_{\delta\longrightarrow \infty} \gamma_s(\delta) \geq \Theta$, where the first inequality is proved in Appendix~\ref{pf:thm:fastandslow} and the second one in Appendix~\ref{pf:thm:slow}. 

\subsection{Proof of Theorem \ref{thm:fastandslow}} 
\label{pf:thm:fastandslow}

\paragraph{Proof of (a).} If the researcher never finds it optimal to undertake fast p-hacking (i.e., $n^*=1$), then $\gamma_h(\delta) = \gamma_s(\delta)$ and the claim holds trivially. From now on, suppose $n^*>1$. 

We use the notations from (\ref{eq:slow}) and (\ref{eq:both}). In the previous appendices, we already showed that $e_a^h=e_a^s=e_a$ and $e_c^s = e_c$. Hence, we have 
\[\gamma_h(\delta) =q_a^h e_a + q_b^h e_b^h + q_c^h e_c^h \quad\text{and}\quad \gamma_s(\delta) = q_a^s e_a + q_b^s e_b^s + q_c^s e_c.\]
To show that $\gamma_h(\delta) > \gamma_s(\delta)$ for any $\delta$, we begin with two lemmas. 

\begin{lemma}
\label{lemma:thm4_1}
We have $e_b^h > e_b^s$ and $e_c^h > e_c$. 
\end{lemma}
\begin{proof}
Notice that the results in category (c) are not slow p-hacked, regardless of whether fast p-hacking is allowed. The only difference between the two scenarios is that, when fast p-hacking is allowed, a researcher may have several insignificant estimates in category (c) to choose from and is indifferent about which one to report. Per Assumption~\ref{assump:tiebreak}, in this case, the researcher reports the estimate with the lowest p-value.\footnote{\label{footnote:tiebreak2}In fact, Theorem \ref{thm:fastandslow} also holds when the researcher adopts an alternative tie-breaking rule in this case, such as (i) reporting an estimate drawn uniformly at random, (ii) reporting the initial estimate $X_1$, or (iii) reporting the latest estimate $X_{n^*}$. Under these alternative tie-breaking rules, we have $e_c^h = e_c$, so Theorem~\ref{thm:fastandslow} remains valid.} As a consequence, the distribution of the published result conditional on its belonging to category (c) with the fast p-hacking option first-order stochastically dominates that without the option, leading to $e_c^h > e_c$. 

By similar reasoning, we know that the distribution of the original estimate for category (b) results with the fast p-hacking option first-order stochastically dominates that without the option. By Lemma~\ref{lemma:slowhack}, we know that $\xi(x)$ strictly increases in $x$ for $x\in(\overline{X}_s, t\Sigma)$, which implies that $e_b^h > e_b^s$. 
\end{proof}

\begin{lemma}
\label{lemma:thm4_2}
We have $e_a > e_b^s > e_c > 0$, $q_a^h > q_a^s$, and $q_c^h < q_c^s$. 
\end{lemma}
\begin{proof}
The inequality $e_a > e_b^s$ is already proved in Appendix~\ref{pf:thm:slow}, and $e_b^s > e_c > 0$ can be proved by the same reasoning using Lemma \ref{lem:keyinequality}. 

Next, we show $q_c^h < q_c^s$. Recall the notations $p_a:=\mathbbm{Pr}_{X\sim\mathcal{N}(\Theta,\Sigma^2)}(|X|\geq t \Sigma)$, $p_b:=\mathbbm{Pr}_{X\sim\mathcal{N}(\Theta,\Sigma^2)}(|X|\in (\overline{X}_s, t \Sigma))$, and $p_c:=\mathbbm{Pr}_{X\sim\mathcal{N}(\Theta,\Sigma^2)}(|X| \in [0, \overline{X}_s])$ from Appendix~\ref{pf:thm:slow}. We have $q_c^s = \frac{p_c/\delta}{1-(1-1/\delta)p_c}$ and $q_c^h = \frac{(p_c)^{n^*}/\delta}{1-(1-1/\delta)(p_c)^{n^*}}$. This further shows the following: 
\begin{align*}
q_c^h<q_c^s \iff & \frac{(p_c)^{n^*}/\delta}{1-(1-1/\delta)(p_c)^{n^*}} < \frac{p_c/\delta}{1-(1-1/\delta)p_c} \\
\iff &(p_c)^{n^*} [1-(1-1/\delta)p_c ]<p_c [1-(1-1/\delta)(p_c)^{n^*}]\\
\iff & (p_c)^{n^*-1}  - (1-1/\delta)](p_c)^{n^*} < 1 - (1-1/\delta)(p_c)^{n^*}\\
\iff & (p_c)^{n^*-1}<1,
\end{align*}
which holds if and only if $n^*>1$. 

Finally, we show $q_a^h > q_a^s$. Notice that 
\[q_a^s = \frac{p_a}{1-(1-1/\delta)p_c}\quad\text{and}\quad q_a^h = \frac{p_a^h}{1-(1-1/\delta)(p_c)^{n^*}}\,,\]
where 
\begin{eqnarray}
p_a^h &:=&  p_a \cdot \left[1+\Pr(|X_1|<\overline{X}_f^1)+\Pr(|X_1|<\overline{X}_f^2)\Pr(|X_2|<\overline{X}_f^2)+\dots\right] \nonumber \\
&=& p_a \cdot \left[1+\sum_{n=1}^{n^*-1}\prod_{k\leq n}\Pr\left(|X_k|<\overline{X}_f^{n}\right)\right] \nonumber \\
&>& p_a \cdot \left(1+\sum_{n=1}^{n^*-1}(p_c)^{n}\right) \nonumber \\
&=& p_a \cdot \frac{1-(p_c)^{n^*}}{1-p_c}, \nonumber
\end{eqnarray}
where the inequality comes from $\overline{X}_f^n>\overline{X}_s$. Therefore, we have 
\begin{eqnarray}
q_a^h > q_a^s &\iff& \frac{p_a^h}{1-(1-1/\delta)(p_c)^{n^*}} > \frac{p_a}{1-(1-1/\delta)p_c} \nonumber \\
&\Longleftarrow&  \frac{[1-(p_c)^{n^*}]/(1-p_c)}{1-(1-1/\delta) (p_c)^{n^*}} \geq \frac{1}{1-(1-1/\delta )p_c} \nonumber \\
&\iff &[1-(p_c)^{n^*}][1-(1-1/\delta)p_c] \geq (1-p_c)[1-(1-1/\delta)(p_c)^{n^*}] \nonumber\\
&\iff& (p_c)^{n^*} + p_c(1-1/\delta) \leq p_c + (p_c)^{n^*}(1-1/\delta) \nonumber\\
&\iff& p_c\geq (p_c)^{n^*}, \nonumber 
\end{eqnarray}
which holds if $n^*>1$. 
\end{proof}

With these two lemmas, we can finally show Part (a) of the theorem as follows. 
\begin{align*}
\gamma_h(\delta) - \gamma_s(\delta) &> (q_a^h e_a + q_b^h e_b^s + q_c^h e_c) - (q_a^s e_a + q_b^s e_b^s + q_c^s e_c) \\ 
&= (q_a^h-q_a^s)(e_a-e_b^s) + (q_c^s-q_c^h)(e_b^s-e_c) \\
&> 0, 
\end{align*}
where the first inequality follows from Lemma~\ref{lemma:thm4_1}, the equality holds by rearranging the terms, and the second inequality exploits Lemma~\ref{lemma:thm4_2}. 

\paragraph{Proof of (b).} By continuity of $\gamma_f$ and $\gamma_h$, it suffices to show that 
\[\lim_{\delta\longrightarrow\infty}\gamma_f(\delta)>\lim_{\delta\longrightarrow\infty}\gamma_h(\delta).\]
We know from Theorem \ref{thm:fast} that 
\[\lim_{\delta\longrightarrow\infty}\gamma_f(\delta) = \lim_{\delta\longrightarrow\infty}\gamma(\delta).\]
Furthermore, we know from Theorem \ref{thm:combined} that
\[\lim_{\delta\longrightarrow\infty}\gamma_h(\delta) <\lim_{\delta\longrightarrow\infty}\gamma(\delta).\]
The claim is thus established. 

\subsection{Proof of Theorem~\ref{thm:identification_onesided}}
\label{pf:thm:identification_onesided}

In Section \ref{pf:thm:identification_onesided}, $\Phi(x)$ and $\phi(x)$ will refer to the CDF and PDF of the standard normal distribution respectively. We will make explicit the dependence on $\Theta$ and $\Sigma$ whenever we reference the distribution of a $N(\Theta, \Sigma^2)$ random variable.

Let $\hat{Z}=\hat{X}/\hat{\Sigma}$ be the reported Z-statistic and define
\begin{align*}
    Q =
    \begin{cases}
        2 & \mbox{ if } \hat{Z} > t_w, \\
        1 & \mbox{ if } t \leq \hat{Z} \leq t_w, \\
        0 & \mbox{ if } -t < \hat{Z} < t, \\
        -2 & \mbox{ if } \hat{Z} \leq -t.
    \end{cases}
\end{align*}
Let $(-\infty,\overline{X}_f^{l}]$ be the continuation region for fast p-hacking after the $l$-th draw. These thresholds do not depend on $\Theta$, because researchers neither know $\Theta$ nor update their beliefs about $\Theta$ during the p-hacking process.

Conditional on $\Theta=b$ and $\Sigma=s$, a report with $Q=2$ is generated when the researcher reaches some draw $j+1$, all previous draws lie in the relevant continuation region, and the $(j+1)$th draw is positive and exceeds $t_w s$. Hence
\begin{equation*}
    \mathbb{P}(Q=2\mid \Theta=b,\Sigma=s)
    =
    \left(1-\Phi\left(\frac{t_w s-b}{s}\right)\right)p_{f,1}(b,s)~,
\end{equation*}
where
\begin{align*}
    p_{f,1}(b,s)
    &=
    \sum_{j=0}^{n^*-1}
    \left(
        \Phi\left(\frac{\overline{X}_f^j-b}{s}\right)
    \right)^j~.
\end{align*}
and we define $\overline{X}_f^{n^*}:=-\infty$ and $\overline{X}_f^0:=\infty$.

Observe that the above function satisfies
\begin{equation}\label{equation--identification_lim_pf_onesided}
    \lim_{b\to\infty}p_{f,1}(b,s)
    =
    1~.
\end{equation}
Moreover,
\begin{equation*}
    \frac{\partial}{\partial b}p_{f,1}(b,s)
    =
    -\sum_{j=1}^{n^*-1}
    \frac{j}{s}
    \left(
        \Phi\left(\frac{\overline{X}_f^j-b}{s}\right)
    \right)^{j-1}
	\cdot
        \phi\left(\frac{\overline{X}_f^j-b}{s}\right)~,
\end{equation*}
Therefore,
\begin{equation}\label{equation--identification_derivative_pf_onesided}
    \lim_{|b|\to\infty}\frac{\partial}{\partial b}p_{f,1}(b,s)
    =
    0.
\end{equation}

We now identify the distribution of latent effects. Let $f_{\hat{X}}$ denote the density of published estimates. Conditional on $Q=2$, $\Sigma=s$, and $\Theta=b$,
\begin{equation*}
    f_{\hat{X}\mid\Theta,\Sigma,Q}(x\mid b,s,2)
    =
    \frac{
        \frac{1}{s}\phi\left(\frac{x-b}{s}\right)
    }{
        1-\Phi\left(\frac{t_w s-b}{s}\right)
    }
    \mathbbm{1}\{x>t_w s\}~.
\end{equation*}
Meanwhile,
\begin{equation*}
    \mu_{\Theta\mid\Sigma,Q}(b\mid s,2)
    =
    \frac{1}{p_{Q,2}(s)}
    \left(1-\Phi\left(\frac{t_w s-b}{s}\right)\right)
    p_{f,1}(b,s)
    \mu_{\Theta\mid\Sigma}(b\mid s)~,
\end{equation*}
where
\begin{equation*}
    p_{Q,2}(s)
    =
    \int
    \left(1-\Phi\left(\frac{t_w s-b}{s}\right)\right)
    p_{f,1}(b,s)
    \mu_{\Theta\mid\Sigma}(b\mid s)
    \; db~.
\end{equation*}
Thus, for $x > t_w s$,
\begin{equation*}
    f_{\hat{X}\mid\Sigma,Q}(x\mid s,2)
    =
    \frac{1}{p_{Q,2}(s)}
    \int
    \frac{1}{s}\phi\left(\frac{x-b}{s}\right)
    p_{f,1}(b,s)
    \mu_{\Theta\mid\Sigma}(b\mid s)
    \; db~.
\end{equation*}
The Gaussian convolution on the right-hand side is real analytic in $x$, so its values on the open set $x>t_w s$ determine it on $\mathbb{R}$. Moreover, Gaussian convolution is injective because the characteristic function of the Gaussian kernel is everywhere nonzero. Partial deconvolution therefore identifies
\begin{equation*}
    g_1(b,s)
    :=
    \frac{p_{f,1}(b,s)}{p_{Q,2}(s)}
    \mu_{\Theta\mid\Sigma}(b\mid s)
    =
    \frac{p_{f,1}(b,s)}{p_{Q,2}(s)}
    \frac{1}{\sigma(s)}
    \phi\left(\frac{b-\theta(s)}{\sigma(s)}\right)~.
\end{equation*}
Differentiating,
\begin{equation*}
    \frac{\partial}{\partial b}g_1(b,s)
    =
    \frac{1}{p_{Q,2}(s)}
    \left[
        -p_{f,1}(b,s)\frac{b-\theta(s)}{\sigma^2(s)}
        +
        \frac{\partial}{\partial b}p_{f,1}(b,s)
    \right]
    \frac{1}{\sigma(s)}
    \phi\left(\frac{b-\theta(s)}{\sigma(s)}\right)~.
\end{equation*}
Next, form
\begin{equation*}
    h(b,s)
    =
    \frac{
        \partial_b g_1(b,s)
    }{
        g_1(b,s)
    }
    =
    -\frac{b-\theta(s)}{\sigma^2(s)}
    +
    \frac{
        \partial_b p_{f,1}(b,s)
    }{
        p_{f,1}(b,s)
    }~.
\end{equation*}
Since $p_{f,1}(b,s) \geq 1$,
\eqref{equation--identification_derivative_pf_onesided} implies that the second term
converges to zero. Hence
\begin{equation*}
    -\frac{1}{\sigma^2(s)}
    =
    \lim_{b\to\infty}\left(h(b+1,s)-h(b,s)\right)~,
\end{equation*}
and
\begin{equation*}
    \frac{2\theta(s)}{\sigma^2(s)}
    =
    \lim_{b\to\infty}\left(h(b,s)+h(-b,s)\right)~.
\end{equation*}
Therefore $\theta(s)$ and $\sigma^2(s)$ are identified, and so is
$\mu_{\Theta\mid\Sigma}(\cdot\mid s) = N\left(\theta(s),\sigma^2(s)\right)$.

In turn,
\begin{equation*}
    \frac{g_1(b,s)}{\mu_{\Theta\mid\Sigma}(b\mid s)}
    =
    \frac{p_{f,1}(b,s)}{p_{Q,2}(s)}
\end{equation*}
is identified. Equation
\eqref{equation--identification_lim_pf_onesided} gives
\begin{equation*}
    p_{Q,2}(s)
    =
    \left(
        \lim_{b\to\infty}
        \frac{g_1(b,s)}{\mu_{\Theta\mid\Sigma}(b\mid s)}
    \right)^{-1}~.
\end{equation*}
Substituting $p_{Q,2}(s)$ back into the preceding equation then identifies
$p_{f,1}(b,s)$.

The density of $\Sigma$ conditional on $Q=2$ satisfies
\begin{equation}\label{equation--sigma_cond_significance_onesided}
    \mu_{\Sigma\mid Q}(s\mid2)
    =
    \frac{1}{\overline p_{Q,2}}
    \mu_\Sigma(s)p_{Q,2}(s),
\end{equation}
where
\begin{equation*}
    \overline p_{Q,2}
    =
    \int p_{Q,2}(s)\mu_\Sigma(s)\;ds.
\end{equation*}
Normalization gives
\begin{equation*}
    \overline p_{Q,2}
    \int
    \frac{\mu_{\Sigma\mid Q}(s\mid2)}{p_{Q,2}(s)}
    \;ds
    =1,
\end{equation*}
so $\overline p_{Q,2}$ is identified. Equation
\eqref{equation--sigma_cond_significance_onesided} then identifies
$\mu_\Sigma(s)$ and therefore the joint distribution $\mu$.

We next identify the fast p-hacking thresholds. Fix $s$ and define
\begin{equation*}
    W_j(b,s)
    :=
    \Phi\left(\frac{\overline X_f^j-b}{s}\right)~.
\end{equation*}
Suppose that two sets of fast-p-hacking parameters,
\begin{equation*}
    \left(
        n,\left\{\overline X_f^j\right\}_{j=1}^{n-1}
    \right)
    \quad\text{and}\quad
    \left(
        n^\dagger,
        \left\{\overline X_f^{\dagger,j}\right\}_{j=1}^{n^\dagger-1}
    \right)~,
\end{equation*}
generate the same identified function $p_{f,1}(b,s)$ for every
$b\in\mathbb{R}$. We write
\begin{equation}\label{eq:pf1_onesided}
    p_{f,1}(b,s)
    =
    1+\sum_{j=1}^{n-1}W_j(b,s)^j~.
\end{equation}
Define $W_j^\dagger(b,s)$ analogously for the second parameterization.

We show inductively that the two sequences of thresholds coincide. Suppose
that the thresholds have already been shown to agree for every $j<k$ and that
the $k$-th term appears in both sets of parameters. Subtract their common
contributions and define
\begin{align*}
    R_k(b,s)
    &:={}
    p_{f,1}(b,s)
    -1
    -\sum_{j=1}^{k-1}W_j(b,s)^j,\\
    R_k^\dagger(b,s)
    &:={}
    p_{f,1}^\dagger(b,s)
    -1
    -\sum_{j=1}^{k-1}W_j^\dagger(b,s)^j~.
\end{align*}
Observational equivalence and the induction hypothesis imply that
\begin{equation}\label{eq:residual_equality_onesided}
    R_k(b,s)=R_k^\dagger(b,s)
    \qquad\text{for every }b\in\mathbb{R}.
\end{equation}

As $b\to\infty$, the asymptotic approximation for the Mills ratio gives
\begin{equation*}
    W_j(b,s)
    =
    \frac{s}{b-\overline X_f^j}
    \phi\left(\frac{b-\overline X_f^j}{s}\right)
    (1+o(1))~~.
\end{equation*}
Taking logs,
\begin{equation}\label{equation--logW_onesided}
    \log W_j(b,s)
    =
    -\frac{(b-\overline X_f^j)^2}{2s^2}
    -\log\left(\frac{b-\overline X_f^j}{s}\right)
    -\frac{1}{2}\log(2\pi)
    +o(1)~.
\end{equation}
Expanding the quadractic term above, we have that for every $j>k$,
\begin{equation*}
    \log\left(
        \frac{W_j(b,s)^j}{W_k(b,s)^k}
    \right)
    =
    -\frac{(j-k)b^2}{2s^2}+O(b)~,
\end{equation*}
which converges to $-\infty$. Thus the $k$-th term dominates the residual:
\begin{align*}
    R_k(b,s)
    &=W_k(b,s)^k(1+o(1))~,\\
    R_k^\dagger(b,s)
    &=W_k^\dagger(b,s)^k(1+o(1))~.
\end{align*}
Equation \eqref{eq:residual_equality_onesided} therefore implies
\begin{equation}\label{eq:positive_tail_ratio_onesided}
    \frac{W_k(b,s)^k}{W_k^\dagger(b,s)^k}
    \to 1
    \quad\mbox{as} \quad b\to\infty~.
\end{equation}

Suppose, toward a contradiction, that
$\overline X_f^k\neq\overline X_f^{\dagger,k}$. Equation
\eqref{equation--logW_onesided} gives
\begin{equation*}
    \log\left(
        \frac{W_k(b,s)}{W_k^\dagger(b,s)}
    \right)
    =
    \frac{\overline X_f^k-\overline X_f^{\dagger,k}}{s^2}b
    +O(1).
\end{equation*}
The right-hand side then converges to either $+\infty$ or $-\infty$, contradicting
\eqref{eq:positive_tail_ratio_onesided}. Therefore
$\overline X_f^k=\overline X_f^{\dagger,k}$ and common term
$W_k(b,s)^k$ can consequently be subtracted from both representations. We can now repeat the argument for the next threshold.

Finally, suppose without loss of generality that $n<n^\dagger$. After the first
$n-1$ common terms have been subtracted, the residual under the first
parameterization is zero, whereas the residual under the second is
\begin{equation*}
    \sum_{j=n}^{n^\dagger-1}W_j^\dagger(b,s)^j.
\end{equation*}
Every term is strictly positive for finite $b$, so this residual cannot vanish
identically. Hence $n=n^\dagger$. Therefore $n^*$ and
$\left\{\overline X_f^j\right\}_{j=1}^{n^*-1}$ are identified as functions of $s$. 

We can now identify the selective publication parameter. Let
$\overline p_{Q,1}:=\mathbb{P}(Q=1)$. Because positive significant results
are published with the same probability,
\begin{equation*}
    \mathbb{P}\left(
        t\leq\hat Z \leq t_w
        \mid
        \hat Z\geq t,D=1
    \right)
    =
    \frac{\overline p_{Q,1}}
    {\overline p_{Q,1}+\overline p_{Q,2}}.
\end{equation*}
Hence $\overline p_{Q,1}$ is identified whenever
$\overline p_{Q,2}>0$.

Recall that
negative significant and insignificant results are published with relative probability
$\lambda=1/\delta$. As such,
\begin{align*}
    \mathbb{P}\left(Q > 0\mid Q \geq 0, D=1\right)
    =
    \frac{
        \overline p_{Q,1}+\overline p_{Q,2}
    }{
        \overline p_{Q,1}+\overline p_{Q,2}
        +
        \lambda
        (1-\overline p_{Q,1}-\overline p_{Q,2})
    }~.
\end{align*}
$\delta$ is thus identified as long as $(1-\overline p_{Q,1}-\overline p_{Q,2}) > 0$, a sufficient condition for which is that $\mathbb{P}(\hat{Z} < t) > 0$. This then gives us $\mathbb{E}[\gamma(\delta)]$, the counterfactual mean under selective
publication only.

Finally, the mass of final reports that have been slow p-hacked into positive
significance is also identified. In particular,
\begin{equation*}
    \overline p_{s,1}
    =
    \overline p_{Q,1}
    -
    \int
    \left[
        \Phi\left(\frac{t_w s-b}{s}\right)
        -
        \Phi\left(\frac{ts-b}{s}\right)
    \right]
    p_{f,1}(b,s)
    \;d\mu_{\Theta,\Sigma}(b,s)~.
\end{equation*}
The integral above is the probability that the final report lies in $Q=1$ without slow p-hacking. Any excess mass must therefore have been generated by slow p-hacking.

\section{One-Sided Likelihood Function}\label{app:likelihood}

The censored likelihood function for the one-sided model is presented below. It is expressed in terms of the finite dimensional parameter $$\beta = (\theta, \sigma, \kappa, \lambda, \{\overline{X}_f^n\}_{n=1}^{n^*-1}, r_{s}, \delta)~.$$
With the exception of $r_s$, the components of $\beta$ are parameters from our model. $r_s \in [0,1]$ is the fraction of the mass in $(-t,t)$ that is slow p-hacked into significance. The proof of Theorem \ref{thm:identification_onesided} shows that $\beta$ is identified. %

The likelihood divides the observations into 4 regions, indexed by $Q_i$:
\begin{equation*}
    Q_i = \begin{cases}
        2 & \mbox{ if } \hat{X}_i/\hat{\Sigma}_i \geq t_w, \\
        1 & \mbox{ if } t \leq \hat{X}_i/\hat{\Sigma}_i < t_w, \\
        0 & \mbox{ if } -t < \hat{X}_i/\hat{\Sigma}_i < t, \\
        -2 & \mbox{ if } \hat{X}_i/\hat{\Sigma}_i \leq -t~,
    \end{cases}
\end{equation*}
Observations for which $Q_i \in \{1,0\}$ are censored. Their contributions to the log-likelihood are the same regardless of the exact values of $(\hat{X}_i, \hat{\Sigma}_i)$. Observations for which $Q_i \in \{2,-2\}$ are uncensored. Their contributions to the log-likelihood depends on the density of $(\hat{X}_i, \hat{\Sigma}_i)$. 

For the uncensored regions, the densities of a point $(x,s)$ are:
\begin{align*}
    f_{Q,2}(x, s \mid \beta ) &=  \int \mu_\Sigma(s, \kappa, \lambda) \, p_{f,1}(b,s) \,\frac{1}{\sigma}\phi\left(\frac{b-\theta}{\sigma}\right)\frac{1}{s}\phi\left(\frac{x-b}{s}\right) \; db  \\
    f_{Q,-2}(x, s \mid \beta ) &=  \frac{1}{\delta}  \int \mu_\Sigma(s, \kappa, \lambda) \, p_{f,-1}(x,b,s) \, \frac{1}{\sigma}\phi\left(\frac{b-\theta}{\sigma}\right)\frac{1}{s}\phi\left(\frac{x-b}{s}\right) \; db~,
\end{align*}
where
\begin{align*}
    p_{f,1}(b,s)  = \sum_{j = 0}^{n^*-1} \Phi\left( \frac{\overline{X}_f^{j}\cdot s -b}{s}\right)^j \quad \mbox{and} \quad p_{f,-1}(x,b,s)  = n^* \cdot \Phi\left(\frac{x-b}{s}\right)^{n^*-1}~.
\end{align*}
The point masses of the censored regions are:
\begin{align*}
    \overline{p}_{Q,1} & =   \int \mu_\Sigma(s, \kappa, \lambda) \, p_{f,1}(b,s) \, \frac{1}{\sigma}\phi\left(\frac{b-\theta}{\sigma}\right)\left(   \Phi\left(\frac{t_w\cdot s - b}{s}\right) - \Phi\left(\frac{t\cdot s - b}{s}\right) \right) \; db \; ds \; + \overline{p}_{s,1}  \\
    \overline{p}_{Q,0} & =  \frac{1}{\delta} \cdot (1-r_s) \int \mu_\Sigma(s, \kappa, \lambda)
    \frac{1}{\sigma}\phi\left(\frac{b-\theta}{\sigma}\right)
    \, p_{f,0}(b,s) \; db \;ds  ~,
\end{align*}
\noindent where
\begin{equation*}
    \overline{p}_{s,1}  = r_s \cdot \int \mu_\Sigma(s, \kappa, \lambda)
    \frac{1}{\sigma}\phi\left(\frac{b-\theta}{\sigma}\right)
    \, p_{f,0}(b,s) \; db\; ds
\end{equation*}
\noindent and 
\begin{align*}
    p_{f,0}(b,s) &= \sum_{j=1}^{n^*-1}
    \left[ \Phi\left(\frac{\overline{X}_f^{j-1}\cdot s-b}{s}\right)^{j-1} \cdot \Phi\left(\frac{t\cdot s-b}{s}\right)
    - \Phi\left(\frac{\overline{X}_f^{j}\cdot s-b}{s}\right)^{j} \right] \\
    & \qquad + \Phi\left(\frac{\overline{X}_f^{n^*-1}\cdot s-b}{s}\right)^{n^*-1} \cdot \Phi\left(\frac{t\cdot s-b}{s}\right)
    - \Phi\left(\frac{-t\cdot s-b}{s}\right)^{n^*} ~.
\end{align*}
As in Appendix \ref{pf:thm:identification_onesided}, we define $\overline{X}^{n^*}_f : = -\infty$ and $\overline{X}^{0}_f := \infty$. Appendix \ref{pf:thm:identification_onesided} also derives all of the above expressions. The exception is $p_{f,-1}$, which is the density for the largest estimate among $n^*$ estimates. This follows from our tie-breaking rule in Assumption \ref{assump:tiebreak_onesided}.

We can write the total mass from all four values of $Q$ as:
\begin{align*}
    \overline{p}_{TOT}  & =  \overline{p}_{Q,1} +  \overline{p}_{Q,0} + \int_{x/s\geq t_w} f_{Q,2}(x, s \mid \beta ) \; dx \; ds +  \int_{x/s\leq -t} f_{Q,-2}(x, s \mid \beta ) \; dx \; ds
\end{align*}
Finally, define:
\begin{align*}
    f(\hat{X}_i, \hat{\Sigma}_i \mid \beta) = \begin{cases}
        f_{Q,2}(\hat{X}_i, \hat{\Sigma_i} \mid \beta ) & \mbox{ if } \hat{X}_i/\hat{\Sigma}_i \geq t_w \\
         \overline{p}_{Q,1} & \mbox{ if } t \leq \hat{X}_i/\hat{\Sigma}_i  < t_w \\
        \overline{p}_{Q,0} & \mbox{ if } -t < \hat{X}_i/\hat{\Sigma}_i  < t \\
        f_{Q,-2}(\hat{X}_i, \hat{\Sigma}_i \mid \beta ) & \mbox{ if } \hat{X}_i/\hat{\Sigma}_i  \leq -t
    \end{cases}
\end{align*}
Then, the likelihood function is
\begin{equation*}
    L(\beta) = \prod_{i=1}^n \frac{{f}(\hat{X}_i, \hat{\Sigma}_i \mid \beta)}{\overline{p}_{TOT}} ~.
\end{equation*}

Additionally, we impose the following constraints: 
\begin{itemize}
    \item $-t \leq \overline{X}_f^{n^*-1} \leq ... \leq \overline{X}_f^1 \leq t$
    \item $0 \leq r_s \leq 1$ 
    \item $\sigma, \delta_0, \delta_-, \kappa, \lambda > 0$
\end{itemize}
The first constraint follows from Proposition \ref{prop:onesided_strat} except the lower bound for the thresholds is set to $-t$ rather than $0$. This is strict relaxation of the model since it is now agnostic about researcher behavior over a wider interval. 
The next two constraints follow from the definition of the variables. 
We implement all of the weak inequality constraints as strict inequality constraints via variable transformations. 

\newpage

\section{Online Appendix: Sensitivity Analysis for Empirical Results}\label{app:emp_robust}

\subsection{Effect of Nudges on Choices}\label{appendix:emp_robust_nudge}

Tables \ref{tab:emp_nudge_nstar} and \ref{tab:emp_nudge_tw} presents results for \cite{mertens2022effectiveness} under different choices of $n^*$ and $t_w$ respectively. In Table \ref{tab:emp_nudge_p}, instead of choosing the most significant estimate, we consider using the $50$-th percentile most significant estimate as the preferred value.\footnote{In this application, each study has at most 3 estimates, so that the 75-th percentile sample is identical to the 100-th percentile sample.} In each of the above tables, the specification in the main text is highlighted in bold. Finally, in Table \ref{tab:emp_nudge_pub}, we consider selecting one estimate per published paper instead of one estimate per study. Table \ref{tab:emp_nudge_pub} should be compared with Table \ref{tab:emp_nudge_p}. Across the board, we find that p-hacking mitigates publication bias and that our results are not sensitive to the choice of tuning parameters. 

\begin{table}[htbp]
  \centering
    \begin{tabular}{cccccc}\hline\hline
    $n^*$ & $1/\delta$ & $\mathbb{E}(\Theta)$ & $\mathbb{E}(\gamma_h)$ & $\mathbb{E}(\gamma)$ & $\mathbb{E}(\gamma) - \mathbb{E}(\gamma_h)$ \\ \hline
    \multirow{2}[0]{*}{1} & 0.095 & 0.284 & 0.531 & 1.009 & 0.478 \\
          & (0.032) & (0.112) & (0.028) & (0.158) & (0.154) \\
    \multirow{2}[0]{*}{\textbf{2}} & \textbf{0.083} & \textbf{0.171} & \textbf{0.531} & \textbf{0.959} & \textbf{0.428} \\
          & \textbf{(0.031)} & \textbf{(0.130)} & \textbf{(0.028)} & \textbf{(0.180)} & \textbf{(0.178)} \\
    \multirow{2}[0]{*}{3} & 0.079 & 0.123 & 0.531 & 0.928 & 0.397 \\
          & (0.024) & (0.124) & (0.028) & (0.143) & (0.138) \\
    \multirow{2}[0]{*}{4} & 0.077 & 0.094 & 0.531 & 0.911 & 0.380 \\
          & (0.026) & (0.132) & (0.028) & (0.151) & (0.148) \\
    \hline\hline
    \end{tabular}%
    \caption{Sensitivity analysis of the \cite{mertens2022effectiveness} application to the choice of $n^*$. Our preferred specification, $n^* = 2$ is in bold.}\label{tab:emp_nudge_nstar}
\end{table}%

\begin{table}[htbp]
  \centering
    \begin{tabular}{cccccc}\hline\hline
    $t_w$ & $1/\delta$ & $\mathbb{E}(\Theta)$ & $\mathbb{E}(\gamma_h)$ & $\mathbb{E}(\gamma)$ & $\mathbb{E}(\gamma) - \mathbb{E}(\gamma_h)$ \\ \hline
    \multirow{2}[0]{*}{3} & 0.148 & 0.165 & 0.531 & 0.619 & 0.088 \\
          & (0.072) & (0.113) & (0.028) & (0.069) & (0.063) \\
    \multirow{2}[0]{*}{\textbf{4}} & \textbf{0.083} & \textbf{0.171} & \textbf{0.531} & \textbf{0.959} & \textbf{0.428} \\
          & \textbf{(0.031)} & \textbf{(0.130)} & \textbf{(0.028)} & \textbf{(0.180)} & \textbf{(0.178)} \\
    \multirow{2}[0]{*}{5} & 0.072 & 0.209 & 0.531 & 1.205 & 0.674 \\
          & (0.029) & (0.159) & (0.028) & (0.297) & (0.296) \\
          \hline\hline
    \end{tabular}%
    \caption{Sensitivity analysis of the \cite{mertens2022effectiveness} application to the choice of $t_w$. Our preferred specification, $t_w = 4$ is in bold.}\label{tab:emp_nudge_tw}
\end{table}%

\begin{table}[htbp]
  \centering
    \begin{tabular}{cccccc} \hline\hline
    $p$   & $1/\delta$ & $\mathbb{E}(\Theta)$ & $\mathbb{E}(\gamma_h)$ & $\mathbb{E}(\gamma)$ & $\mathbb{E}(\gamma) - \mathbb{E}(\gamma_h)$ \\ \hline
    \multirow{2}[0]{*}{\textbf{1}} & \textbf{0.083} & \textbf{0.171} & \textbf{0.531} & \textbf{0.959} & \textbf{0.428} \\
          & \textbf{(0.031)} & \textbf{(0.130)} & \textbf{(0.028)} & \textbf{(0.180)} & \textbf{(0.178)} \\
    \multirow{2}[0]{*}{0.5} & 0.104 & 0.211 & 0.521 & 0.906 & 0.386 \\
          & (0.030) & (0.109) & (0.028) & (0.116) & (0.114) \\
    \hline\hline
    \end{tabular}%
    \caption{Sensitivity analysis of the \cite{mertens2022effectiveness} application to the choice of $p$. Our preferred specification, $p = 1$ (i.e. selecting the most significant estimate) is in bold. Because studies have 3 or fewer estimates each, $p =0.75$ leads to the same preferred sample as $p = 1$ and is omitted.}\label{tab:emp_nudge_p}
\end{table}%

\begin{table}[htbp]
  \centering
    \begin{tabular}{cccccc}\hline\hline
    $p$   & $1/\delta$ & $\mathbb{E}(\Theta)$ & $\mathbb{E}(\gamma_h)$ & $\mathbb{E}(\gamma)$ & $\mathbb{E}(\gamma) - \mathbb{E}(\gamma_h)$ \\
    \hline
    \multirow{2}[0]{*}{1} & 0.080 & 0.176 & 0.567 & 0.938 & 0.371 \\
          & (0.032) & (0.151) & (0.037) & (0.155) & (0.149) \\
    \multirow{2}[0]{*}{0.75} & 0.082 & 0.181 & 0.569 & 0.921 & 0.352 \\
          & (0.032) & (0.150) & (0.037) & (0.130) & (0.122) \\
    \multirow{2}[0]{*}{0.5} & 0.080 & 0.099 & 0.529 & 0.937 & 0.408 \\
          & (0.033) & (0.158) & (0.036) & (0.203) & (0.194) \\
    \hline\hline
    \end{tabular}%
    \caption{Results for \cite{mertens2022effectiveness} using one estimate per published paper instead one estimate per study.}\label{tab:emp_nudge_pub}\end{table}%

\subsection{Effect of Development Aid on Economic Growth}\label{appendix:emp_robust_aid}

Tables \ref{tab:emp_aid_nstar} and \ref{tab:emp_aid_tw} presents results for \cite{doucouliagos2011ineffectiveness} under different choices of $n^*$ and $t_w$ respectively. In Table \ref{tab:emp_aid_p}, instead of choosing the most significant estimate, we consider using the $75$-th and $50$-th percentile most significant estimate as the preferred value. In the above tables, the specification in the main text is highlighted in bold. Across the board, we find that p-hacking exacerbates publication bias and that our results are robust to the choice of tuning parameters.  

Finally, Table \ref{tab:emp_aid_2s} presents results from fitting the general model in which $\delta_-$ is estimated. See Appendix \ref{app:twosided_emp}. We find that the estimate of $\mathbb{E}(\gamma)$---and hence $\mathbb{E}(\gamma)-\mathbb{E}(\gamma_h)$---is very close to that in the one-sided model. However, standard errors are much larger, likely because the two-sided model adds 3 parameters (including $\delta_-$) and observations are censored over a larger set ($(-4,4]$ instead of $(-1.96, 4]$). Taken together, the point estimates do not appear to contradict the conclusion of exacerbation in this literature.

\begin{table}[htbp]
  \centering
    \begin{tabular}{cccccc}
    \hline\hline
    $n^*$ & $1/\delta$ & $\mathbb{E}(\Theta)$ & $\mathbb{E}(\gamma_h)$ & $\mathbb{E}(\gamma)$ & $\mathbb{E}(\gamma) - \mathbb{E}(\gamma_h)$ \\
	\hline
    \multirow{2}[0]{*}{1} & 0.140 & -0.069 & 0.185 & 0.160 & -0.026 \\
          & (0.112) & (0.075) & (0.032) & (0.043) & (0.032) \\
    \multirow{2}[0]{*}{\textbf{2}} & \textbf{0.149} & \textbf{-0.113} & \textbf{0.185} & \textbf{0.121} & \textbf{-0.064} \\
          & \textbf{(0.057)} & \textbf{(0.061)} & \textbf{(0.032)} & \textbf{(0.038)} & \textbf{(0.017)} \\
    \multirow{2}[0]{*}{3} & 0.130 & -0.147 & 0.185 & 0.110 & -0.076 \\
          & (0.049) & (0.061) & (0.032) & (0.039) & (0.019) \\
    \multirow{2}[0]{*}{4} & 0.119 & -0.168 & 0.185 & 0.102 & -0.083 \\
          & (0.047) & (0.063) & (0.032) & (0.041) & (0.019) \\
        \hline\hline
    \end{tabular}%
    \caption{Sensitivity analysis of the \cite{doucouliagos2011ineffectiveness} application to the choice of $n^*$. Our preferred specification, $n^* = 2$ is in bold.}\label{tab:emp_aid_nstar}
\end{table}%

\begin{table}[htbp]
  \centering
    \begin{tabular}{cccccc}\hline\hline
    $t_w$ & $1/\delta$ & $\mathbb{E}(\Theta)$ & $\mathbb{E}(\gamma_h)$ & $\mathbb{E}(\gamma)$ & $\mathbb{E}(\gamma) - \mathbb{E}(\gamma_h)$\\
    \hline
    \multirow{2}[1]{*}{3} & 0.151 & -0.110 & 0.185 & 0.131 & -0.054 \\
          & (0.056) & (0.062) & (0.032) & (0.039) & (0.013) \\
    \multirow{2}[0]{*}{\textbf{4}} & \textbf{0.149} & \textbf{-0.113} & \textbf{0.185} & \textbf{0.121} & \textbf{-0.064} \\
          & \textbf{(0.057)} & \textbf{(0.061)} & \textbf{(0.032)} & \textbf{(0.038)} & \textbf{(0.017)} \\
    \multirow{2}[0]{*}{5} & 0.356 & -0.050 & 0.185 & 0.072 & -0.113 \\
          & (0.449) & (0.125) & (0.032) & (0.056) & (0.050) \\
          \hline\hline
    \end{tabular}%
    \caption{Sensitivity analysis of the \cite{doucouliagos2011ineffectiveness} application to the choice of $t_w$. Our preferred specification, $t_w = 4$ is in bold.}\label{tab:emp_aid_tw}
\end{table}%

\begin{table}[htbp]
  \centering
    \begin{tabular}{cccccc} \hline\hline
    $p$   & $1/\delta$ & $\mathbb{E}(\Theta)$ & $\mathbb{E}(\gamma_h)$ & $\mathbb{E}(\gamma)$ & $\mathbb{E}(\gamma) - \mathbb{E}(\gamma_h)$ \\
    \hline
    \multirow{2}[0]{*}{\textbf{1}} & \textbf{0.149} & \textbf{-0.113} & \textbf{0.185} & \textbf{0.121} & \textbf{-0.064} \\
          & \textbf{(0.057)} & \textbf{(0.061)} & \textbf{(0.032)} & \textbf{(0.038)} & \textbf{(0.017)} \\
    \multirow{2}[0]{*}{0.75} & 0.581 & 0.000 & 0.144 & 0.059 & -0.085 \\
          & (0.571) & (0.098) & (0.029) & (0.046) & (0.035) \\
    \multirow{2}[0]{*}{0.5} & 1.149 & 0.103 & 0.124 & 0.091 & -0.033 \\
          & (0.544) & (0.068) & (0.024) & (0.037) & (0.027) \\
    \hline\hline
    \end{tabular}%
    \caption{Sensitivity analysis of the \cite{doucouliagos2011ineffectiveness} application to the choice of $p$. Our preferred specification, $p=1$ (i.e. selecting the most significant estimate) is in bold.}\label{tab:emp_aid_p}
\end{table}%

\begin{table}[htbp]
  \centering
    \begin{tabular}{cccccc}\hline\hline
    $1/\delta_-$ & $1/\delta_0$ & $\mathbb{E}(\Theta)$ & $\mathbb{E}(\gamma_h)$ & $\mathbb{E}(\gamma)$ & $\mathbb{E}(\gamma) - \mathbb{E}(\gamma_h)$ \\
    \hline
    0.024 & 0.155 & -0.265 & 0.185 & 0.120 & -0.065 \\
    (0.023) & (0.287) & (0.138) & (0.032) & (0.174) & (0.171) \\
    \hline\hline
    \end{tabular}%
    \caption{Results for \cite{doucouliagos2011ineffectiveness} under two-sided selection. Here, $n^* = 2$, $t_w = 4$ and $ p = 1$.}\label{tab:emp_aid_2s}
\end{table}%

\newpage

\section{Online Appendix: Proofs of Secondary Results}
\label{onlineapx_secondaryproof}

\subsection{Proof of Proposition  \ref{thm:slow_quad}}\label{pf:thm:slow_quad}

Under quadratic cost, When undertaking slow p-hacking, a researcher who undertakes slow p-hacking manipulates the point estimate to
\[\xi(X) = \frac{c_Xt^2}{c_Xt^2+c_\Sigma}\cdot X + \frac{c_\Sigma}{c_Xt^2+c_\Sigma}\cdot t\Sigma=\alpha \cdot X +(1-\alpha)\cdot t\Sigma\,,\]
where $\alpha:= c_Xt^2/(c_Xt^2+c_\Sigma)$. Recall our previous notation $\lambda = 1/\delta$. The slow p-hacking threshold satisfies 
\[V(1-\lambda) = \kappa_s \frac{c_Xc_\Sigma}{c_Xt^2+c_\Sigma}(t\Sigma-\overline{X}_s)^2\iff \frac{t\Sigma-\overline{X}_s}{\sqrt{1-\lambda}}=\sqrt{\frac{V[c_Xt^2+c_\Sigma]}{\kappa_sc_Xc_\Sigma}}\,.\]

Let $\varepsilon:=t\Sigma - \overline{X}_s$. Under quadratic costs, we thus have
\[\varepsilon=t\Sigma-\overline{X}_s=\sqrt{\frac{(1-\lambda)V[c_Xt^2+c_\Sigma]}{\kappa_sc_Xc_\Sigma}}.\]
Note that the RHS is a strictly decreasing bijection from $\lambda\in (0,1)$ to $\varepsilon\in (0,\overline{\varepsilon})$ where $\overline{\varepsilon}:=\sqrt{V[c_Xt^2+c_\Sigma]/(\kappa_sc_Xc_\Sigma)}$. Therefore, we can use $\varepsilon$ instead of $\lambda$ as the key parameter in our analysis. We have:
\[\lambda = 1-\frac{\varepsilon^2}{\overline{\varepsilon}^2}\quad\text{and}\quad \frac{d\lambda}{d\varepsilon}=-\frac{2\varepsilon}{\overline{\varepsilon}^2}<0\,.\]

Recall from the proof of Theorem \ref{thm:slow} that 
\[\text{sign}(\tilde{\gamma}_s(\lambda) - \tilde{\gamma}(\lambda)) = \text{sign}(\tilde{\Lambda}(\lambda))\,,\]
where $\tilde{\Lambda}(\lambda) := (e_b^s- \lambda e_b) - (1-\lambda) \tilde{\gamma}(\lambda)$. Furthermore, for $\lambda<1$, $ \text{sign}(\tilde{\Lambda}(\lambda))= \text{sign}(\tilde{\Lambda}(\lambda)/(1-\lambda))$. Throughout, let $w_+(x) := \phi(x)+\phi(-x)$ and $w_-(x) = \phi(x)-\phi(-x)$. We have
\begin{align*}
\frac{\tilde{\Lambda}(\lambda)}{1-\lambda} = \frac{(e_b^s- \lambda e_b)}{1-\lambda}-\tilde{\gamma}(\lambda)&=\int_{\overline{X}_s}^{t\Sigma} \frac{(\xi(x) -\lambda x)}{(1-\lambda)}\frac{w_-(x)}{\int_{\overline{X}_s}^{t\Sigma}w_+(x')dx'}dx-\tilde{\gamma}(\lambda)\\
&=\int_{\overline{X}_s}^{t\Sigma} \left(\frac{(\xi(x) -\lambda x)}{(1-\lambda)}\frac{w_-(x)}{w_+(x)}-\tilde{\gamma}(\lambda)\right)\frac{w_+(x)}{\int_{\overline{X}_s}^{t\Sigma}w_+(x')dx'}dx.
 \end{align*}
 Let 
 \[\Psi(\varepsilon):=\int_{t\Sigma-\varepsilon}^{t\Sigma} \left[\left(\frac{\overline{\varepsilon}^2}{\varepsilon^2}\left(\xi(x) -x\right)+x\right)\frac{w_-(x)}{w_+(x)}-\tilde{\gamma}\left(1-\frac{\varepsilon^2}{\overline{\varepsilon}^2}\right)\right]w_+(x)dx.\]
 We thus have 
\[\text{sign}(\tilde{\gamma}_s(\lambda) - \tilde{\gamma}(\lambda)) = \text{sign}(\Psi(\varepsilon))\,.\]
We want to show that $\Psi(\varepsilon)=0$ implies $\Psi'(\varepsilon)<0$ for any $\varepsilon\in (0,\overline{\varepsilon})$. Differentiating $\Psi$, we get 
\begin{align}
\Psi'(\varepsilon) &=  \left(\frac{\overline{\varepsilon}^2}{\varepsilon^2}(\xi(\overline{X}_s)-\overline{X}_s) + \overline{X}_s\right)w_-(\overline{X}_s)-\tilde{\gamma}\left(1-\frac{\varepsilon^2}{\overline{\varepsilon}^2}\right)w_+(\overline{X}_s) \nonumber \\
&-\frac{2\overline{\varepsilon}^2}{\varepsilon^3}\int_{t\Sigma-\varepsilon}^{t\Sigma} \left(\xi(x) -x\right)w_-(x)dx+\frac{2\varepsilon}{\overline{\varepsilon}^2}\tilde{\gamma}'\left(1-\frac{\varepsilon^2}{\overline{\varepsilon}^2}\right) \cdot \int_{t\Sigma-\varepsilon}^{t\Sigma} w_+(x)dx. \label{eq:sc1}
\end{align}
Note that $\tilde{\gamma}$ is strictly decreasing in $\lambda$. Indeed,
\[\tilde{\gamma}'(\lambda) = \frac{p_a(p_b+p_c)}{(p_a+\lambda(p_b+p_c))^2}\cdot\left(\frac{p_be_b+p_ce_c}{p_b+p_c}-e_a\right)<0,\]
where the inequality comes from Lemma \ref{lem:key}. Also, $\int_{t\Sigma-\varepsilon}^{t\Sigma} w_+(x)dx > 0$. We can thus infer that the fourth term of (\ref{eq:sc1}) is negative. 

Take any $\varepsilon$ such that $\Psi(\varepsilon)=0$. By definition of $\Psi$, we have 
\begin{equation}
\int_{t\Sigma-\varepsilon}^{t\Sigma} \left(\frac{\overline{\varepsilon}^2}{\varepsilon^2}(\xi(x)-x)+x\right)w_-(x)dx=\tilde{\gamma}\left(1-\frac{\varepsilon^2}{\overline{\varepsilon}^2}\right)\int_{t\Sigma-\varepsilon}^{t\Sigma}w_+(x)dx. \label{eq:sc2}
\end{equation}
Then, 
\begin{align}
\Psi'(\varepsilon) < &\left( \frac{\overline{\varepsilon}^2}{\varepsilon^2} (\xi(\overline{X}_s)- \overline{X}_s)+\overline{X}_s\right)w_-(\overline{X}_s)-w_+(\overline{X}_s)\int_{t\Sigma-\varepsilon}^{t\Sigma}\left(\frac{\overline{\varepsilon}^2}{\varepsilon^2}(\xi(x)-x)+x \right)\frac{w_-(x)}{\int_{t\Sigma-\varepsilon}^{t\Sigma}w_+(x')dx'}dx \nonumber  \\
&-\frac{2\overline{\varepsilon}^2}{\varepsilon^3}\int_{t\Sigma-\varepsilon}^{t\Sigma} \left(\xi(x) -x\right)w_-(x)dx \nonumber \\
=&w_+(\overline{X}_s)\times \nonumber \\ 
&\quad\int_{t\Sigma-\varepsilon}^{t\Sigma}\left[ \frac{\overline{\varepsilon}^2}{\varepsilon^2}\left(\xi(\overline{X}_s) - \overline{X}_s)+ \overline{X}_s\right)\frac{w_-(\overline{X}_s)}{w_+(\overline{X}_s)}-\left(\frac{\overline{\varepsilon}^2}{\varepsilon^2}(\xi(x)-x)+ x\right)\frac{w_-(x)}{w_+(x)}\right]\frac{w_+(x)}{\int_{t\Sigma-\varepsilon}^{t\Sigma}w_+(x')dx'}dx\nonumber \\
&-\frac{2\overline{\varepsilon}^2}{\varepsilon^3} \int_{t\Sigma-\varepsilon}^{t\Sigma} \left(\xi(x) -x\right)w_-(x)dx \nonumber \\
\leq&w_-(\overline{X}_s)\int_{t\Sigma-\varepsilon}^{t\Sigma} \left[\left(\frac{\overline{\varepsilon}^2}{\varepsilon^2}(\xi(\overline{X}_s) -\overline{X}_s)+\overline{X}_s\right)-\left(\frac{\overline{\varepsilon}^2}{\varepsilon^2} (\xi(x)-x)+x\right)\right]\frac{w_+(x)}{\int_{t\Sigma-\varepsilon}^{t\Sigma}w_+(x')dx'}dx \nonumber \\
&-\frac{2\overline{\varepsilon}^2}{\varepsilon^3} \frac{w_-(\overline{X}_s)}{w_+(\overline{X}_s)}\int_{t\Sigma-\varepsilon}^{t\Sigma} \left(\xi(x) -x\right)w_+(x)dx. \label{eq:sc3}
\end{align}
where the first inequality uses (\ref{eq:sc2}) and the fact that the fourth term of (\ref{eq:sc1}) is negative, the equality rearranges terms, and the last inequality uses the fact that $w_-(x)/w_+(x)$ is increasing and $\xi(x) > x$ over the relevant range of $x$.

Let $F(x):=\int^{t\Sigma}_xw_+(x')dx'$. Integrating by part and using the fact that $\xi(x)-x = (1-\alpha)(t\Sigma-x)$ under quadratic costs, we have 
\begin{align}
& \int_{t\Sigma-\varepsilon}^{t\Sigma} \left[\left(\frac{\overline{\varepsilon}^2}{\varepsilon^2}(\xi(\overline{X}_s) -\overline{X}_s)+\overline{X}_s\right)-\left(\frac{\overline{\varepsilon}^2}{\varepsilon^2} (\xi(x)-x)+x\right)\right]w_+(x)dx \nonumber \\
=& -\int_{t\Sigma-\varepsilon}^{t\Sigma}\frac{d}{dx}\left(\frac{\overline{\varepsilon}^2}{\varepsilon^2} (\xi(x)-x)+x\right)F(x)dx\nonumber \\
=& -\left(1-\frac{\overline{\varepsilon}^2}{\varepsilon^2} (1-\alpha)\right)\int_{t\Sigma-\varepsilon}^{t\Sigma}F(x)dx, \label{eq:sc4}
\end{align}
and 
\begin{align}
\int_{t\Sigma-\varepsilon}^{t\Sigma} \left(\xi(x) -x\right)w_+(x)dx&=(\xi(\overline{X}_s)-\overline{X}_s))F(\overline{X}_s)-(1-\alpha)\int_{t\Sigma-\varepsilon}^{t\Sigma}F(x)dx\nonumber \\
&=(1-\alpha)\varepsilon F(\overline{X}_s)-(1-\alpha)\int_{t\Sigma-\varepsilon}^{t\Sigma}F(x)dx. \label{eq:sc5}
\end{align}
Substitute (\ref{eq:sc4}) and (\ref{eq:sc5}) into (\ref{eq:sc3}), we know that $\Psi'(\varepsilon) < 0$ if 
\begin{align}
-\left(1-\frac{\overline{\varepsilon}^2}{\varepsilon^2} (1-\alpha)\right)&\frac{\int_{t\Sigma-\varepsilon}^{t\Sigma}F(x)dx}{F(\overline{X}_s)} - \frac{2\overline{\varepsilon}^2}{\varepsilon^3} \frac{(1-\alpha)}{w_+(\overline{X}_s)}\left(\varepsilon F(\overline{X}_s)-\int_{t\Sigma-\varepsilon}^{t\Sigma}F(x)dx\right) \leq 0. \label{eq:sc6}
\end{align}
Dividing (\ref{eq:sc6}) by $\varepsilon>0$, it collapses to 
\[\left[\frac{\overline{\varepsilon}^2}{\varepsilon^2} (1-\alpha)-1\right]\times r(\varepsilon) \leq \frac{2\overline{\varepsilon}^2}{\varepsilon^2}(1-\alpha)\times\frac{1-r(\varepsilon)}{z(\varepsilon)}\]
where
\[r(\varepsilon):=\frac{\int_{t\Sigma-\varepsilon}^{t\Sigma}F(x)dx}{\varepsilon F(\overline{X}_s)}\quad\text{and}\quad z(\varepsilon):=\frac{\varepsilon w_+(\overline{X}_s)}{F(\overline{X}_s)}\,.\]
Since $\frac{\overline{\varepsilon}^2}{\varepsilon^2} (1-\alpha)-1\leq \frac{\overline{\varepsilon}^2}{\varepsilon^2}(1-\alpha)$, it suffices to show that 
\[r(\varepsilon) \leq 2\times\frac{1-r(\varepsilon)}{z(\varepsilon)}\iff r(\varepsilon)\leq \frac{2}{2+z(\varepsilon)}\,.\]
By assumption, $w_+$ is log-concave on $[\overline{X}_s, t\Sigma]$. Thus, for each $s\geq 0$, the function $x\mapsto \log(w_+(x+s)) - \log (w_+(x))$ is non-increasing over the relevant range. Exponentiating, this implies that $w_+(x+s)/w_+(x)$ is non-increasing as well. 

By a change of variable, we can write
\[\frac{F(x)}{w_+(x)} = \frac{\int_x^{t\Sigma}w_+(x')dx'}{w_+(x)} = \frac{\int_0^{t\Sigma-x}w_+(x+s)dx'}{w_+(x)}\,.\]
Thus, $F(x)/w_+(x)$ is also non-increasing in $x$. Note that 
\[-\frac{d\log (F(x))}{dx} = \frac{w_+(x)}{F(x)}\,,\]
which is non-decreasing. Thus, the decay rate of $F$ at $\overline{X}_s$ is a lower bound on its decay rate at any $x\in (\overline{X}_s,t\Sigma)$: 
\[-\frac{d\log (F(x))}{dx}\geq \frac{w_+(\overline{X}_s)}{F(\overline{X}_s)}=\frac{z(\varepsilon)}{\varepsilon}\quad\forall x\in (\overline{X}_s,t\Sigma).\]
Integrating and exponentiating gives 
\[F(x)\leq F(\overline{X}_s)\exp\left(-(x-\overline{X}_s)\frac{z(\varepsilon)}{\varepsilon}\right)\,.\]
Thus, 
\begin{align*}
r(\varepsilon)\leq \frac{\int_{t\Sigma-\varepsilon}^{t\Sigma}F(\overline{X}_s)\exp\left(-(x-\overline{X}_s)\frac{z(\varepsilon)}{\varepsilon}\right)dx}{\varepsilon F(\overline{X}_s)}&=\frac{1-\exp(-z(\varepsilon))}{z(\varepsilon)}\,.
\end{align*}
To finish the proof, we have left to show that, for $z(\varepsilon)>0$,
\[\frac{1-\exp(-z(\varepsilon))}{z(\varepsilon)}\leq \frac{2}{2+z(\varepsilon)}\iff 0\leq (2+z(\varepsilon))\exp(-z(\varepsilon))-2+z(\varepsilon)\,.\]
The RHS is increasing in $z(\varepsilon)$ and equal to zero at $z(\varepsilon)=0$, so the inequality holds. As a result, for any  $\varepsilon$ such that $\Psi(\varepsilon)=0$, we have $\Psi'(\varepsilon)<0$. Thus, the sign of $\tilde{\gamma}_s(\lambda) - \tilde{\gamma}(\lambda)$ only changes once: there exists $\lambda^*\in (0,1)$ such that $\tilde{\gamma}_s(\lambda)\geq \tilde{\gamma}(\lambda)$ iff $\lambda \geq \lambda^*$. Equivalently, there exists a single threshold $1<\delta^*<\infty$ such that slow p-hacking mitigates publication bias if $\delta > \delta^*$ and exacerbates publication bias if $\delta < \delta^*$.

\subsection{Proof of Proposition \ref{prop:onesided_strat}} \label{pf:prop:onesided_strat}

\paragraph{Disclaimer on notations:} For the sake of readability, the proofs in Appendies~\ref{pf:prop:onesided_strat} and \ref{pf:thm:onesided} reuse several notations from the proofs of Section~\ref{section:twosided} (e.g., the continuation value functions $U_n(\cdot)$ and $V_n(\cdot)$, the benefit function $\Delta_n(\cdot)$, the expected estimates $e_i$, and the probabilities $(p_i, q_i)$). Unless otherwise specified, these terms in Appendices~\ref{pf:prop:onesided_strat} and \ref{pf:thm:onesided} are defined locally within the appendices for the one-sided selective publication setting and should not be confused with their counterparts in the proofs for two-sided selective publication. %

We characterize the optimal p-hacking strategy by backward induction. 

\paragraph{Step 1: Optimal slow p-hacking.} 

When the researcher decides to undertake slow p-hacking, he must have stopped drawing new estimates. Suppose he has already drawn $k$ estimates, with the best current estimate being $\breve{X}^*_k:=\max_{l\leq k}X_l$. He slow p-hacks if and only if 
\[V - \kappa_s C(\breve{X}^*_k)  > \frac{V}{\delta}.\]
The LHS is strictly increasing in $\breve{X}^*_k \in [0, t\Sigma]$, strictly smaller than the RHS for $\breve{X}^*_k = 0$ because of Assumption~\ref{assump:largekappas}, and strictly larger than the RHS for $\breve{X}^*_k = t\Sigma$. Thus, there exists $\overline{X}_s \in (0, t\Sigma)$ such that the researcher slow p-hacks if and only if the current best estimate $\breve{X}^*_k \in (\overline{X}_s, t\Sigma)$. Based on his slow p-hacking strategy, we define his continuation value upon stopping drawing new estimates as
\[U_{s}(x) = \frac{V}{\delta}+V\left(1-\frac{1}{\delta}\right)\mathbbm{1}\{x >\overline{X}_s\}- \kappa_s C(x)\mathbbm{1}\{\overline{X}_s< x < t\Sigma\}\]
when the value of his best current estimate is $x$.

\paragraph{Step 2: Optimal fast p-hacking.} 

~\\
\noindent \underline{Step 2.1: Existence of $n^*$.} The same arguments as in the proof of Proposition \ref{prop:combined_strat} yield that the maximal number of studies drawn, $n^*$, is finite. As before, we introduce two functions. Let $U_n(x)$ denote the researcher's continuation value under his optimal strategy when he has drawn $n$ estimates, with the best current estimate being $x$. By the definition of $n^*$, we have $U_{n^*}(x) = U_s(x)$. Let $V_n(x)$ denote the researcher's continuation value if he has drawn $n$ estimates, with the best current estimate being $x$, and decides to draw one more estimate. By definition, we have
\begin{equation*}
V_n(x) = \mathbb{E}_{X_{n+1}\sim\mathcal{N}(0,\Sigma^2)}[U_{n+1}(\max\{x, X_{n+1}\})] - \kappa_f c_{n+1},
\end{equation*}
and also,
\begin{equation*}
U_n(x) = \max\{U_s(x), V_n(x)\}.
\end{equation*}

Moreover, we let $\Delta_n(x) := [V_n(x) + \kappa_f c_{n+1}] - U_s(x)$ denote the benefit of drawing one more estimate (without taking into account the cost). Apparently, $V_n(x) > U_s(x)$ is equivalent to $\Delta_n(x) > \kappa_f c_{n+1}$. From the researcher's cost functions, we can infer that $U_n(x)$, $V_n(x)$, and $\Delta_n(x)$ are all constant for $x \le \overline{X}_s$. It remains to show two statements:

\begin{enumerate}[(a)]
    \item For any $n \le n^* - 1$, $\exists \overline{X}_f^n \in (\overline{X}_s, t\Sigma)$ such that $V_n(x) > U_s(x)$ if and only if $x < \overline{X}_f^n$.
    \item For $n \le n^* - 1$, the threshold $\overline{X}_f^n$ strictly decreases in $n$.
\end{enumerate}

\noindent \underline{Step 2.2: Optimal fast p-hacking for $n = n^* - 1$.} When $n = n^* - 1$, the researcher draws another estimate if and only if $\Delta_{n^*-1}(x) > \kappa_f c_{n^*}$. From its closed form, we can infer that: $\Delta_{n^*-1}(x)$ is constant for $x \le \overline{X}_s$, strictly decreases for $x \in (\overline{X}_s, t\Sigma)$, and becomes constant again at zero for $x \ge t\Sigma$. Since $\Delta_{n^*-1}(t\Sigma) = 0 <\kappa_f  c_{n^*} < \Delta_{n^*-1}(\overline{X}_s)$, whose last inequality holds as the researcher may draw one more estimate with $n = n^* - 1$, we know that there exists $\overline{X}_f^{n^*-1} \in (\overline{X}_s, t\Sigma)$ such that $V_{n^*-1}(x) > U_s(x)$ if and only if $x < \overline{X}_f^{n^*-1}$. Accordingly,
\begin{equation*}
U_{n^*-1}(x) = 
\begin{cases} 
U_s(x) & \text{if } x \ge \overline{X}_f^{n^*-1} \\ 
V_{n^*-1}(x) & \text{if } x < \overline{X}_f^{n^*-1}. 
\end{cases}
\end{equation*}

\noindent \underline{Step 2.3: Optimal fast p-hacking for $n < n^* - 1$.} We now use induction to prove the two statements identified in Step 2.1 for any $n < n^* - 1$.

\begin{lemma}
\label{lemma:delta_onesided}
Given $n < n^* - 1$. Suppose $V_{n+1}(x)$, $\Delta_{n+1}(x)$, and $c_{n+2}$ satisfy the following.
\begin{enumerate}[(i)]
    \item \textit{$V_{n+1}(x)$ is constant for $x \le \overline{X}_s$, strictly increases for $x \in (\overline{X}_s, t\Sigma)$, and becomes constant again for $x \ge t\Sigma$;}
    \item \textit{$\Delta_{n+1}(x)$ is constant for $x \le \overline{X}_s$, strictly decreases for $x \in (\overline{X}_s, t\Sigma)$, and becomes constant again at zero for $x \ge t\Sigma$;}
    \item \textit{$\Delta_{n+1}(\overline{X}_s) > \kappa_f c_{n+2}$.}
\end{enumerate}
Then $V_n(x)$, $\Delta_n(x)$, and $c_{n+1}$ also satisfy these conditions.
\end{lemma}

\begin{proof}
\underline{Part (i).} Since both $V_{n+1}(x)$ and $U_s(x)$ satisfy condition (i), we know that $U_{n+1}(x) = \max\{U_s(x), V_{n+1}(x)\}$ also satisfies condition (i). As $V_n(x) = \mathbb{E}_{X_{n+1}\sim\mathcal{N}(0,\Sigma^2)}[U_{n+1}(\max\{x, X_{n+1}\})] - \kappa_f c_{n+1}$, we can further infer that $V_n(x)$ is also constant for $x \le \overline{X}_s$, strictly increases for $x \in (\overline{X}_s, t\Sigma)$, and becomes constant again for $x \ge t\Sigma$, proving condition (i) for $V_n(x)$.

\noindent \underline{Part (ii).} The above result, together with the fact that $U_s(x)$ satisfies condition (i), immediately indicates that $\Delta_n(x) = V_n(x) - U_s(x) + \kappa_f c_{n+1}$ is constant for $x \le \overline{X}_s$. We also know that $\Delta_n(x) = 0$ for $x \ge t\Sigma$ as the researcher gains nothing from one more draw if his best current estimate is already positively significant. It remains to show that $\Delta_n(x)$ strictly decreases for $x \in (\overline{X}_s, t\Sigma)$.

To show this, let $\overline{X}_s < x' < x'' < t\Sigma$, and it suffices to show $\Delta_n(x') > \Delta_n(x'')$, which is equivalent to
\begin{equation*}
\mathbb{E}_{X_{n+1}\sim\mathcal{N}(0,\Sigma^2)}[U_{n+1}(\max\{x', X_{n+1}\})] - U_s(x') > \mathbb{E}_{X_{n+1}\sim\mathcal{N}(0,\Sigma^2)}[U_{n+1}(\max\{x'', X_{n+1}\})] - U_s(x'')
\end{equation*}
by definition, which, with terms rearranged, is
\begin{equation}
\mathbb{E}_{X_{n+1}\sim\mathcal{N}(0,\Sigma^2)}[U_{n+1}(\max\{x'', X_{n+1}\}) - U_{n+1}(\max\{x', X_{n+1}\})] < U_s(x'') - U_s(x'). \label{eq:a10}
\end{equation}
Notice that $U_{n+1}(\max\{x'', X_{n+1}\}) - U_{n+1}(\max\{x', X_{n+1}\}) \le U_{n+1}(x'') - U_{n+1}(x')$, with the inequality being strict if $X_{n+1} \in (x', x'')$. Hence, the LHS of (\ref{eq:a10}) satisfies
\begin{equation}
\mathbb{E}_{X_{n+1}\sim\mathcal{N}(0,\Sigma^2)}[U_{n+1}(\max\{x'', X_{n+1}\}) - U_{n+1}(\max\{x', X_{n+1}\})] < U_{n+1}(x'') - U_{n+1}(x'). \label{eq:a11}
\end{equation}
Meanwhile, we have
\begin{equation*}
U_{n+1}(x) - U_s(x) = 
\begin{cases} 
0 & \text{if } x \ge \overline{X}_f^{n+1} \\ 
V_{n+1}(x) - U_s(x) = \Delta_{n+1}(x) -\kappa_f  c_{n+2} & \text{if } x < \overline{X}_f^{n+1}. 
\end{cases}
\end{equation*}
Since $\Delta_{n+1}(x)$ is strictly decreasing for $x \in (\overline{X}_s, t\Sigma)$ according to the premise of this lemma, we know that $U_{n+1}(x) - U_s(x)$ is weakly decreasing for $x \in (\overline{X}_s, t\Sigma)$, and therefore,
\begin{equation*}
U_{n+1}(x') - U_s(x') \ge U_{n+1}(x'') - U_s(x'')
\end{equation*}
which is equivalent to
\begin{equation}
U_s(x'') - U_s(x') \ge U_{n+1}(x'') - U_{n+1}(x'). \label{eq:a12}
\end{equation}
Combining (\ref{eq:a11}) and (\ref{eq:a12}), we finally show (\ref{eq:a10}) and thus conclude the proof of (ii) for $\Delta_n(x)$.

\noindent \underline{Part (iii).} Following the exact reasoning of Lemma~\ref{lemma:delta}, we have $\Delta_n(\overline{X}_s) \ge \Delta_{n+1}(\overline{X}_s) > \kappa_f c_{n+2} > \kappa_f c_{n+1}$, which concludes the proof.
\end{proof}

Combining statements (ii) and (iii) of Lemma~\ref{lemma:delta_onesided} and the fact that $\Delta_n(x)$ is constant for $x \le \overline{X}_s$, we conclude that there exists $\overline{X}_f^n \in (\overline{X}_s, t\Sigma)$ such that $V_n(x) > U_s(x)$ if and only if $x < \overline{X}_f^n$. This proves statement (a) specified in Step 2.1 for any $n < n^* - 1$.

Finally, we prove statement (b) in Step 2.1. For any $n < n^* - 1$, we have
\begin{equation}
U_s(\overline{X}_f^n) = V_n(\overline{X}_f^n) > V_{n+1}(\overline{X}_f^n), \label{eq:a13}
\end{equation}
where the equality follows from the definition of $\overline{X}_f^n$ and the inequality holds because $V_n(x) > V_{n+1}(x)$ for any $x \in (\overline{X}_s, t\Sigma)$, as previously proved. Hence, (\ref{eq:a13}) implies that $\overline{X}_f^n > \overline{X}_f^{n+1}$ by the definition of $\overline{X}_f^{n+1}$. Therefore, we conclude that for $n \le n^* - 1$, the threshold $\overline{X}_f^n$ strictly decreases in $n$.

\noindent \underline{Step 2.4: Characterization of $n^*$.} Similarly to the proofs of Proposition~\ref{prop:combined_strat}, the maximum number of draws $n^*$ is uniquely pinned down by evaluating the maximum potential benefit of a draw, which occurs at $x=-\infty$ under one-sided selective publication. Thus, $n^* = \max \{n \mid \kappa_f c_n < \Delta_{n^*-1}(-\infty) \}$.

\subsection{Proof of Theorem \ref{thm:onesided}}
\label{pf:thm:onesided}

Similarly to Appendix~\ref{pf:thm:slow}, we categorize the published results by the value of the original estimate (i.e., the $X_n$ chosen by the researcher). Category (a) consists of published results whose original estimate satisfies $X_n\geq t\Sigma$; category (b) consists of published results whose original estimate satisfies $X_n\in (\overline{X}_s, t\Sigma)$; category (c) consists of published results whose original estimate satisfies $X_n \leq \overline{X}_s$. Based on this categorization, we can express the expected published results of the two scenarios as follows. 
\begin{eqnarray}
\breve{\gamma}(\delta) &=& e_a q_a + e_b q_b + e_c q_c \label{eq:nohack_onesided} \\
\breve{\gamma}_h(\delta) &=& e_a^h q_a^h + e_b^h q_b^h + e_c^h q_c^h. \label{eq:both_onesided}
\end{eqnarray}

\paragraph{\underline{Part 1}: p-hacking mitigates publication bias if $\delta > \overline{\delta}_1$.} 

As $\delta\longrightarrow\infty$, only positively significant studies are published. Therefore, 
\[\lim\limits_{\delta\rightarrow\infty}\breve{\gamma}(\delta) = e_a=\mathbb{E}_{X\sim\mathcal{N}(\Theta,\Sigma^2)}[X \mid X\geq t\Sigma]\]
and 
\[\lim\limits_{\delta\rightarrow\infty}\breve{\gamma}_h(\delta) = e_a^h q_a^h + e_b^h q_b^h.\]
Since the researcher stops p-hacking once obtaining a positively significant estimate, we know that $e_a^h=e_a$. Therefore, it suffices to show that $e_b^h < e_a$. This is straightforward under one-sided selective publication because $e_a > t\Sigma \geq e_b^h$, where the last inequality follows from the same reasoning as Lemma~\ref{lemma:slowhack}. 

To conclude the proof of Part 1, we also need to show that $\lim\limits_{\delta\rightarrow\infty}\breve{\gamma}_h(\delta) > \Theta$. This holds because both the selection of the maximum draw and the slow p-hacking process weakly increase the reported estimate relative to the original normal draw. 

\paragraph{\underline{Part 2}: p-hacking exacerbates publication bias if $\delta < \underline{\delta}_1$.} According to Appendix~\ref{pf:prop:onesided_strat}, when $\delta$ is sufficiently small, fast p-hacking is never undertaken, so the reported result must be based on $X_1$. Hence, we have 
\begin{align}
\breve{\gamma}_h(\delta) &= \frac{e_a p_a + e_b^h p_b + e_c p_c/\delta}{p_a + p_b + p_c/\delta} \label{eq:a21} \\
\breve{\gamma}(\delta) &= \frac{e_a p_a + e_b p_b/\delta + e_c p_c/\delta}{p_a + p_b/\delta + p_c/\delta}, \label{eq:a22}
\end{align}
where $p_a = \mathbbm{Pr}(X_1 \geq t \Sigma)$, $p_b = \mathbbm{Pr}(X_1 \in (\overline{X}_s, t \Sigma))$, $p_c = \mathbbm{Pr}(X_1 < \overline{X}_s)$, $e_a=\mathbb{E}[X_1 \mid X_1\geq t\Sigma]$, $e_b=\mathbb{E}[X_1 \mid X_1\in (\overline{X}_s, t\Sigma)]$, $e_b^h=\mathbb{E}[\xi(X_1) \mid X_1\in (\overline{X}_s, t\Sigma)]$, and $e_c=\mathbb{E}[X_1 \mid X_1 \leq \overline{X}_s]$. Similarly to Appendix~\ref{pf:thm:slow}, we have 
$$\breve{\gamma}_h(\delta)-\breve{\gamma}(\delta) = \frac{p_b}{p_a+p_b+\lambda p_c} \cdot \tilde{\Lambda}(\lambda)$$
where $\lambda = 1/\delta$, $\tilde{\Lambda}(\lambda):= (e_b^h-\lambda e_b) - (1-\lambda)\cdot \frac{e_a p_a + \lambda e_b p_b + \lambda e_c p_c}{p_a + \lambda p_b +\lambda p_c}$. It suffices to show that $\lim\limits_{\lambda\rightarrow 1}\tilde{\Lambda}'(\lambda) < 0$. Following the same reasoning as Appendix~\ref{pf:thm:slow}, we have 
\begin{eqnarray}
\lim\limits_{\lambda\rightarrow 1}\tilde{\Lambda}'(\lambda) &=& \lim\limits_{\lambda\rightarrow 1}\frac{\partial (e_b^h-e_b)}{\partial \lambda} - t\Sigma + \Theta - 0.  \label{eq:a23}
\end{eqnarray}

\begin{lemma}
\label{lemma:onesided}
Under one-sided selective publication, we have $\lim\limits_{\lambda\rightarrow 1}\frac{\partial (e_b^h-e_b)}{\partial \lambda}=-\infty$. 
\end{lemma}
\begin{proof}
Notice that $e_b^h-e_b = \frac{\int_{\overline{X}_s}^{t\Sigma}(\xi(x)-x) \phi(x) dx}{\int_{\overline{X}_s}^{t\Sigma} \phi(x) dx}$. We let $N(\lambda)$ denote its numerator and ${D(\lambda)}$ its denominator, so 
\begin{eqnarray}
\frac{\partial (e_b^h-e_b)}{\partial \lambda} 
&=& \frac{N'(\lambda)}{D(\lambda)} - \frac{N(\lambda)\cdot D'(\lambda)}{[D(\lambda)]^2} \nonumber \\
&=& \frac{\partial \overline{X}_s}{\partial \lambda} \cdot \phi(\overline{X}_s) \cdot \left\{\frac{\overline{X}_s-\xi(\overline{X}_s)  }{\int_{\overline{X}_s}^{t\Sigma} \phi(x) dx}  + \frac{\int_{\overline{X}_s}^{t\Sigma}(\xi(x)-x) \phi(x) dx}{\left[\int_{\overline{X}_s}^{t\Sigma} \phi(x) dx\right]^2}  \right\}. \label{eq:a24}
\end{eqnarray}
Since $\overline{X}_s\rightarrow t\Sigma$ as $\lambda\rightarrow 1$, we apply the L'Hospital rule to the terms in the curly bracket of (\ref{eq:a24}) when $\overline{X}_s\rightarrow t\Sigma$, it ends up being equal to
\begin{eqnarray}
&& \frac{\lim\limits_{x\rightarrow t\Sigma}\xi'(x) - 1}{\phi(t\Sigma)} + \frac{\lim\limits_{x\rightarrow t\Sigma}\xi'(x) - 1}{-2 \phi(t\Sigma)} \nonumber \\
 &=& \lim\limits_{x\rightarrow t\Sigma}(\xi'(x) - 1) \cdot \frac{1}{2 \phi(t\Sigma)}. \nonumber
\end{eqnarray}
Hence, we have 
\begin{equation}
\lim\limits_{\lambda\rightarrow 1} \frac{\partial (e_b^h-e_b)}{\partial \lambda} = \lim\limits_{\lambda\rightarrow 1} \frac{\partial \overline{X}_s}{\partial \lambda} \cdot \left(\lim\limits_{x\rightarrow t\Sigma}\xi'(x) - 1\right) \cdot \frac{1}{2}. \label{eqn13}
\end{equation}
We can directly apply the following results from the proof of Lemma~\ref{lemma:tech}: first, $\lim\limits_{\lambda\rightarrow 1}\frac{\partial \overline{X}_s}{\partial \lambda} = +\infty$; second, $\lim\limits_{x\rightarrow t\Sigma} \left( \xi'(x) - 1 \right) < 0$. These two results allow us to conclude that $\lim\limits_{\lambda\rightarrow 1} \frac{\partial (e_b^h-e_b)}{\partial \lambda} = -\infty$. 

\end{proof}

Applying Lemma~\ref{lemma:onesided} to (\ref{eq:a23}), we have $\lim\limits_{\lambda\rightarrow 1}\tilde{\Lambda}'(\lambda) =-\infty <0$. This completes the proof. 

\newpage

\section{Online Appendix: General Empirical Model}\label{app:twosided_emp}

This appendix studies the empirical model in which $\delta_- \in [1, \delta_0]$ is estimated. It nests both the one- and two-sided selection models. Appendix \ref{app:twosided_emp_model} discusses stochastic slow p-hacking. Appendix \ref{app:twosided_emp_identification} presents our identification result. %
Appendix \ref{app:emp_twosided_likelihood} presents the censored likelihood for the model.

\subsection{Empirical Model with Stochastic Slow p-Hacking}\label{app:twosided_emp_model}

Suppose the researcher picks an estimate $X_n$ for slow p-hacking. Because of Assumption~\ref{assump:largekappas}, his slow p-hacking never changes the sign of the estimate. If $X_n > 0$, slow p-hacking it to positive significance costs $\kappa_s C(X_n)$, and it generates an output $(\hat{X}, \hat{\Sigma})$, drawn from an arbitrary distribution $H^+(\cdot, \cdot \mid X_n, \Sigma)$, that must satisfy $\hat{X}/\hat{\Sigma} \in[t, t_w]$ for some known $t_w \geq t$. Similarly, if $X_n < 0$, slow p-hacking it to negative significance costs $\kappa_s C(X_n)$ and generates an output $(\hat{X}, \hat{\Sigma})$, drawn from an arbitrary distribution $H^-(\cdot, \cdot \mid X_n, \Sigma)$, that must satisfy $\hat{X}/\hat{\Sigma} \in [-t_w, -t]$. Note that if we let $t_w = t$, the setting for slow p-hacking in the empirical model is identical to the theoretical model. As discussed in \cite{simonsohn2020fast}, slow and fast p-hacking differ in how much control researchers have over their p-values. Allowing slow p-hacking to be stochastic captures the fact that this control is imperfect. Researchers nonetheless retain substantially more control under slow than fast p-hacking, as encoded by the support restriction on the former. 

We summarize the p-hacking strategies of researchers in this more general environment:

\begin{proposition}\label{prop:empirical_strat}
The researcher draws at most $n^*$ estimates. His optimal p-hacking strategy can be represented by the slow thresholds $\{\underline{X}_s, \overline{X}_s\}$ and a series of fast thresholds $\{\underline{X}_f^n, \overline{X}_f^n\}_{n=1}^{n^*-1}$ satisfying $-t\Sigma < \underline{X}_s < \overline{X}_s < t\Sigma$ and\footnote{Like Footnote~\ref{footnote1}, as long as no confusion arises, we reuse the notations of $n^*$ and some thresholds in the empirical model, although they do not necessarily have the same values as in the theoretical model.}
\[-\infty \leq \underline{X}_f^{1} \leq \underline{X}_f^{2} \leq ... \leq \underline{X}_f^{n^*-1} < \underline{X}_s < 0 < \overline{X}_s < \overline{X}_f^{n^*-1} < ... < \overline{X}_f^{2} < \overline{X}_f^{1} < t\Sigma. \]
If he has already drawn $k < n^*$ estimates, he draws a new estimate if and only if all his current estimates are contained in the interval $(\underline{X}_f^k, \overline{X}_f^k)$. After he stops drawing new estimates, he picks the estimate $X_n$ that maximizes his continuation value, which is $\left\{V\left(1-\frac{1}{\delta_0}\right)-\kappa_s C(X_n)\right\}$ if $X_n>0$ and $\left\{V\left(\frac{1}{\delta_{-}}-\frac{1}{\delta_0}\right)-\kappa_s C(X_n) \right\}$ if $X_n<0$. He slow p-hacks this estimate if it falls in $(-t\Sigma, \underline{X}_s) \cup (\overline{X}_s, t\Sigma)$. \end{proposition}
\begin{proof}
See Online Appendix~\ref{pf:prop:empirical_strat}.
\end{proof}

\subsection{Identification under Normality}\label{app:twosided_emp_identification}

\begin{theorem}\label{thm:identification_twosided}
   Suppose $t$ and $t_w$ are known and that $\Theta \mid \Sigma = s \sim N(\theta(s), \sigma^2(s))$. Moreover, suppose $\mathbb{P}(\hat{X}_i/\hat{\Sigma}_i > t_w), \mathbb{P}(\hat{X}_i/\hat{\Sigma}_i < -t_w), \mathbb{P}(|\hat{X}_i/\hat{\Sigma}_i| < t) > 0$ and that $\sigma^2(s) > 0$ for all $s \in \text{Supp}(\Sigma)$. Then the joint distribution of $(\Theta, \Sigma)$, as well as the parameters $\delta_0$ and $\delta_-$ are identified. In particular, $\mathbb{E}(\Theta)$, $\mathbb{E}(\gamma)$ and $\mathbb{E}(\gamma_h)$ are identified. %
\end{theorem}
\begin{proof}
See Online Appendix~\ref{pf:thm:identification_twosided}.
\end{proof}

Identification in the the general model is more challenging because it involves more parameters. In particular, we require both $\mathbb{P}(\hat{X}_i/\hat{\Sigma}_i > t_w)$ and $\mathbb{P}(\hat{X}_i/\hat{\Sigma}_i < -t_w)$ to be non-zero. In practice, many literatures appear to be missing one of the two tails. 

For ease of exposition, we have normalized the sign of the most preferentially published results to be positive. The proof of Theorem \ref{thm:identification_twosided} does not assume that this normalization is known. In other words, the main direction of p-hacking is estimated and not imposed.

\subsection{Likelihood Function}\label{app:emp_twosided_likelihood}

As with the one-sided model, we assume that $\Theta$ and $\Sigma$ are independent and let $(\theta, \sigma^2):=(\mathbb{E}(\Theta),\text{Var}(\Theta))$. Suppose also that $\Sigma \sim \Gamma(\kappa, \lambda)$ and that
\[C(X; \Sigma)=\tilde{C}(X/\Sigma).\]
Then, the censored likelihood function is as presented below. It is expressed in terms of the finite dimensional parameter $$\beta = (\theta, \sigma, \kappa, \lambda, \{\underline{X}_f^n, \overline{X}_f^n\}_{n=1}^{n^*-1}, r_{s}, r_{s,1}, \delta_0, \delta_-)~.$$
Observe that relative to the one-sided model, the general model requires the additional parameters $\{\underline{X}_f^n\}_{n=1}^{n^*-1}, r_{s,1}$ and $\delta_-$. This is because the direction of p-hacking is now unknown ex ante. 
The proof of Theorem \ref{thm:identification_twosided} shows that $\beta$ is identified. %

Define:
\begin{align*}
    f_{Q,2}(x, s \mid \beta ) &= \int \mu_\Sigma(s, \kappa, \lambda) p_{f,1}(b,s) \frac{1}{\sigma}\phi\left(\frac{b-\theta}{\sigma}\right)\frac{1}{s}\phi\left(\frac{x-b}{s}\right) \; db  \\
    f_{Q,-2}(x, s \mid \beta ) &= \frac{1}{\delta_-}\int \mu_\Sigma(s, \kappa, \lambda) p_{f,-1}(b,s) \frac{1}{\sigma}\phi\left(\frac{b-\theta}{\sigma}\right)\frac{1}{s}\phi\left(\frac{x-b}{s}\right) \; db 
\end{align*}
\noindent and {\small
\begin{align*}
        \overline{p}_{Q,1} & = \int \mu_\Sigma(s, \kappa, \lambda) p_{f,1}(b,s) \frac{1}{\sigma}\phi\left(\frac{b-\theta}{\sigma}\right)\left(   \Phi\left(\frac{t_w\cdot s - b}{s}\right) - \Phi\left(\frac{t\cdot s - b}{s}\right) \right) \; db \; ds \; + \overline{p}_{s,1}  \\
        \overline{p}_{Q,0} & = \frac{1}{\delta_0}\int \mu_\Sigma(s, \kappa, \lambda) 
        \frac{1}{\sigma}\phi\left(\frac{b-\theta}{\sigma}\right)
        p_{f,0}(b,s)
        \; db \;ds \; - \frac{p_{s,1}}{\delta_0} -  \frac{p_{s,-1}}{\delta_0}\\
        \overline{p}_{Q,-1} & = \frac{1}{\delta_-}\int \mu_\Sigma(s, \kappa, \lambda) p_{f,-1}(b,s) \frac{1}{\sigma}\phi\left(\frac{b-\theta}{\sigma}\right)\left(   \Phi\left(\frac{-t\cdot s - b}{s}\right) - \Phi\left(\frac{-t_w\cdot s - b}{s}\right) \right) \; db \; ds \;+ \frac{\overline{p}_{s,-1}}{\delta_-} ~,
\end{align*}}
where
\begin{align*}
    \overline{p}_{s,1} & = r_s \cdot r_{s,1} \cdot \int \mu_\Sigma(s, \kappa, \lambda) 
    \frac{1}{\sigma}\phi\left(\frac{b-\theta}{\sigma}\right)
    p_{f,0}(b,s) \; db\; ds
    \\
    \overline{p}_{s,-1} & = r_s \cdot (1-r_{s,1}) \cdot \int \mu_\Sigma(s, \kappa, \lambda) 
    \frac{1}{\sigma}\phi\left(\frac{b-\theta}{\sigma}\right)
    p_{f,0}(b,s) \; db \; ds
    ~,
\end{align*}
\noindent and  
\begin{align*}
    p_{f,1}(b,s) & = \sum_{j = 0}^{n^*-1} \mathbbm{1}\{ \overline{X}^{j+1}_f\cdot s < \infty\} \cdot \left( \Phi\left( \frac{\overline{X}_f^{j}\cdot s -b}{s}\right) - \Phi\left( \frac{\underline{X}_f^{j}\cdot s -b}{s}\right) \right)^j~\\
    p_{f,-1}(b,s) & = \sum_{j = 0}^{n^*-1} \mathbbm{1}\{          \underline{X}^{j+1}_f\cdot s > - \infty\} \cdot \left( \Phi\left( \frac{\overline{X}_f^{j}\cdot s -b}{s}\right) - \Phi\left( \frac{\underline{X}_f^{j}\cdot s -b}{s}\right) \right)^j\\
    p_{f,0}(b,s) &= \sum_{j=1}^{n^*-1}  
    \Bigg[ \left( \Phi\left(\frac{\overline{X}_f^{j-1}\cdot s-b}{s}\right) -  \Phi\left(\frac{\underline{X}_f^{j-1}\cdot s-b}{s}\right)\right)^{j-1} \cdot  \left( \Phi\left(\frac{t\cdot s-b}{s}\right) -  \Phi\left(\frac{-t\cdot s-b}{s}\right)\right)  \\
    & \qquad \qquad -\left( \Phi\left(\frac{\overline{X}_f^{j}\cdot s-b}{s}\right) -  \Phi\left(\frac{\underline{X}_f^{j}\cdot s-b}{s}\right)\right)^{j} \Bigg] \\
    & \qquad + \left( \Phi\left(\frac{\overline{X}_f^{n^*-1}\cdot s-b}{s}\right) -  \Phi\left(\frac{\underline{X}_f^{n^*-1}\cdot s-b}{s}\right)\right)^{n^*-1} \left( \Phi\left(\frac{t\cdot s-b}{s}\right) -  \Phi\left(\frac{-t \cdot s -b}{s}\right) \right)~.
\end{align*}
As in Appendix \ref{pf:thm:identification_twosided}, we define $\overline{X}^{n^*}_f : = -\infty$ $\underline{X}^{n^*}_f : = \infty$, $\overline{X}^{0}_f := \infty$, $\underline{X}^{0}_f := -\infty$. In the above expressions, $f_{Q,2}$ and $f_{Q,-2}$ are the densities of the region $Q=2$ and $Q=-2$ respectively. $\overline{p}_{Q,1}, \overline{p}_{Q,0}$ and $\overline{p}_{Q,-1}$ are the probabilities that $Q = 1, 0$ or $-1$ respectively. As such, the total mass from all 5 values of $Q$ is:
\begin{align*}
    \overline{p}_{TOT}  = & \int_{x/s\geq t_w} f_{Q,2}(x, s \mid \beta ) \; dx \; ds + \int_{x/s\leq -t_w} f_{Q,-2}(x, s \mid \beta ) \; dx \; ds \\
    & \qquad + \overline{p}_{Q,1} + \overline{p}_{Q,0} + \overline{p}_{Q,-1}~.
\end{align*}
Finally, define:
\begin{align*}
    f(\hat{X}_i, \hat{\Sigma}_i \mid \beta) = \begin{cases}
        f_{Q,2}(\hat{X}_i, \hat{\Sigma}_i \mid \beta ) & \mbox{ if } \hat{X}_i/\hat{\Sigma}_i \geq t_w \\
        \overline{p}_{Q,1} & \mbox{ if } t \leq \hat{X}_i/\hat{\Sigma}_i < t_w \\
        \overline{p}_{Q,0} & \mbox{ if } -t < \hat{X}_i/\hat{\Sigma}_i < t \\
        \overline{p}_{Q,-1} & \mbox{ if } -t_w < \hat{X}_i/\hat{\Sigma}_i \leq -t \\
        f_{Q,-2}(\hat{X}_i, \hat{\Sigma}_i \mid \beta ) & \mbox{ if } \hat{X}_i/\hat{\Sigma}_i \leq -t_w
    \end{cases}
\end{align*}
Then, the likelihood function is
\begin{equation*}
    L(\beta) = \prod_{i=1}^n \frac{{f}(\hat{X}_i, \hat{\Sigma}_i \mid \beta)}{\overline{p}_{TOT}} ~.
\end{equation*}

When estimating the two-sided model, we also impose the following constraints: 
\begin{itemize}
    \item $-t < \underline{X}_f^1 < ... < \underline{X}_f^{n^*-1} < \overline{X}_f^{n^*-1} < ... < \overline{X}_f^1 < t$
    \item $0 < r_s, r_{s,1} < 1$ 
    \item $\sigma, \delta_0, \delta_-, \kappa, \lambda > 0$
\end{itemize}
The first constraint follows from Proposition \ref{prop:empirical_strat} and our assumption that $-t \leq \underline{X}_f^1$ and $\overline{X}_f^1 \leq t$. In principle, the model also allows either $\overline{X}_f^j = \infty$ or $\underline{X}_f^j = - \infty$. We omit these cases since they are not nested with the main model and would have to be checked case-by-case. The next two constraints follow the definition of the variables. Note also that the first two constraints are technically weak inequality constraints. We implement the strict inequality constraints by transforming the related variables, so to avoid constrained optimization.

\subsection{Proof of Proposition~\ref{prop:empirical_strat}}
\label{pf:prop:empirical_strat}

We characterize the optimal p-hacking strategy by backward induction, determining the slow p-hacking strategy first, followed by the fast p-hacking strategy. 

\paragraph{Step 1: Optimal Slow p-Hacking} Suppose the researcher has stopped drawing new estimates and must decide whether to engage in slow p-hacking. At this point, he evaluates the optimal action for each individual estimate and then selects the one that yields the highest continuation value. 

By Assumption 1, the researcher never changes the sign of an estimate through slow p-hacking. Consider an individual estimate $x$. If it is already significant, the researcher reports it without slow p-hacking. If it is positive and insignificant, he slow p-hacks it if and only if 
\[V - \kappa_s C(x)  > \frac{V}{\delta_0};\]
if it is negative and insignificant, he slow p-hacks it if and only if 
\[\frac{V}{\delta_-} - \kappa_s C(x) > \frac{V}{\delta_0}.\]
The indifference points pin down the slow p-hacking thresholds $\overline{X}_s \in (0, t\Sigma)$ and $\underline{X}_s \in (-t\Sigma, 0)$.\footnote{In the degenerate case where $\delta_{-}=\delta_0$, such a $\underline{X}_s$ does not exist, as slow p-hacking to negative significance is meaningless. This case is analyzed in Section~\ref{section:onesided}.} Based on these two thresholds, the continuation value from using the individual estimate $x$ is: 
\[
U_s(x) = \max \left\{\frac{V}{\delta_0} , \left[V - \kappa_s C(x)\right]\mathbbm{1}\{x \ge 0\}, \left[\frac{V}{\delta_-} - \kappa_s C(x)\right] \mathbbm{1}\{x < 0\} \right\}
\]
Among all estimates that he has, he will choose the estimate $X_n$ that maximizes $U_s(X_n)$. 

\paragraph{Step 2: Optimal Fast p-Hacking} 

~\\
\noindent \underline{Step 2.1: Existence of $n^*$} \ Because the cost of fast p-hacking, $c_n$, strictly increases with $n$ and $\lim_{n \to \infty} c_n = \infty$, there is a finite maximum number of draws, $n^*$, that the researcher is willing to make before the cost strictly exceeds any potential benefit $V$.

Let $U_n(x)$ denote the researcher's continuation value when they have drawn $n$ estimates and the best current estimate (the one maximizing the stopping payoff $U_s(x)$) is $x$. By definition, $U_{n^*}(x) = U_s(x)$. Let $V_n(x)$ denote the expected continuation value if the researcher decides to draw one more estimate:
\[
V_n(x) = \mathbb{E}_{X_{n+1}\sim\mathcal{N}(0,\Sigma^2)} \left[ \max\{U_{n+1}(x), U_{n+1}(X_{n+1})\} \right] - \kappa_f c_{n+1}
\]
The researcher draws another estimate if $V_n(x) > U_s(x)$. Let $\Delta_n(x) = V_n(x) - U_s(x) + \kappa_f c_{n+1}$ represent the gross expected benefit of drawing one more estimate. 

\noindent \underline{Step 2.2: Fast p-hacking for $n = n^* - 1$} \ For the final possible draw, the benefit of drawing again is $\Delta_{n^*-1}(x)$. Since $U_s(x)$ is constant for $x \in [\underline{X}_s, \overline{X}_s]$, $\Delta_{n^*-1}(x)$ is also constant in this interval. 

As $x$ moves into the positive tail, $U_s(x)$ strictly increases in $(\overline{X}_s, t\Sigma)$ and then remains constant in $[t\Sigma, \infty)$, which causes the marginal benefit of a new draw, $\Delta_{n^*-1}(x)$, to strictly decrease in $(\overline{X}_s, t\Sigma)$ and then remain constant in $[t\Sigma, \infty)$. Similarly, as $x$ moves into the negative tail, $\Delta_{n^*-1}(x)$ strictly decreases before $x$ reaches $-t\Sigma$ and remains constant afterwards. 

Since the researcher will draw one more estimate if $\Delta_{n^*-1}(x) > \kappa_f c_{n^*}$, and $\Delta_{n^*-1}(x)$ decreases toward both tails, there exist unique thresholds $\underline{X}_f^{n^*-1}$ and $\overline{X}_f^{n^*-1}$ such that the researcher draws an $(n^*)$-th estimate if and only if $x \in (\underline{X}_f^{n^*-1}, \overline{X}_f^{n^*-1})$.

\noindent \underline{Step 2.3: Optimal fast p-hacking for $n < n^* - 1$} \ We now use induction to prove the properties of the value functions and the existence and nested nature of the fast p-hacking thresholds for any $n < n^* - 1$. 

\begin{lemma}
\label{lemma:empirical}
Given $n < n^* - 1$. Suppose $V_{n+1}(x)$, $\Delta_{n+1}(x)$, and $c_{n+2}$ satisfy the following:
\begin{enumerate}
    \item[(i)] $V_{n+1}(x)$ is constant for $x \in [\underline{X}_s, \overline{X}_s]$, strictly increases for $x \in (\overline{X}_s, t\Sigma)$ and $x \in (-t\Sigma, \underline{X}_s)$ as $x$ approaches the respective significance boundaries, becomes constant in the intervals $[t\Sigma, \infty)$ and $(-\infty, -t\Sigma]$.
    \item[(ii)] $\Delta_{n+1}(x)$ is constant for $x \in [\underline{X}_s, \overline{X}_s]$, strictly decreases for $x \in (\overline{X}_s, t\Sigma)$ and $x \in (-t\Sigma, \underline{X}_s)$ as $x$ approaches the significance boundaries. Furthermore, $\Delta_{n+1}(x)$ is constant at zero for $x \ge t\Sigma$, and constant at some value $\Delta_{n+1}(-\infty) \ge 0$ for $x \le -t\Sigma$.
    \item[(iii)] $\Delta_{n+1}(\overline{X}_s) > \kappa_f c_{n+2}$ and $\Delta_{n+1}(\underline{X}_s) > \kappa_f c_{n+2}$.
\end{enumerate}
Then $V_n(x)$, $\Delta_n(x)$, and $c_{n+1}$ also satisfy these conditions.
\end{lemma}

\begin{proof}
\underline{Part (i).} Since both $V_{n+1}(x)$ and $U_s(x)$ satisfy condition (i), $U_{n+1}(x) = \max\{U_s(x), V_{n+1}(x)\}$ also satisfies condition (i). As $V_n(x) = \mathbb{E}_{X_{n+1}\sim\mathcal{N}(0,\Sigma^2)} \left[ \max\{U_{n+1}(x), U_{n+1}(X_{n+1})\} \right] - \kappa_f c_{n+1}$, it follows that $V_n(x)$ is constant for $x \in [\underline{X}_s, \overline{X}_s]$, strictly increases as $x$ approaches the tails within the insignificant region, and becomes constant for significant results. This proves condition (i) for $V_n(x)$.

\noindent \underline{Part (ii).} The gross expected benefit of drawing one more estimate is defined as $\Delta_n(x) = V_n(x) - U_s(x) + \kappa_f c_{n+1}$. Because both $V_n(x)$ and $U_s(x)$ are constant on $[\underline{X}_s, \overline{X}_s]$, $\Delta_n(x)$ is also constant on this interval.

To evaluate the significant regions, we must observe the asymmetric payoffs. If $x \ge t\Sigma$, the estimate is positively significant and yields the maximum possible payoff $U_s(x) = V$, and therefore, $\Delta_n(x) = 0$ for all $x \ge t\Sigma$. 

Conversely, if $x \le -t\Sigma$, the estimate is negatively significant and yields $U_s(x) = V/\delta_-$. If $\delta_- > 1$, this payoff is strictly less than $V$. A new draw $X_{n+1}$ could potentially be positively significant (or close enough to profitably slow p-hack to positive significance), yielding a payoff greater than $V/\delta_-$. Because the expected marginal gain of obtaining a positively significant result does not depend on the exact magnitude of the already-held negatively significant estimate, $\Delta_n(x)$ is constant for all $x \le -t\Sigma$, but this constant, denoted $\Delta_n(-\infty)$, is strictly positive if $\delta_- > 1$.

It remains to show that $\Delta_n(x)$ strictly decreases as $x$ moves towards $t\Sigma$ in $(\overline{X}_s, t\Sigma)$, and as $x$ moves towards $-t\Sigma$ in $(-t\Sigma, \underline{X}_s)$. The proof needed is identical to its counterpart in Proposition~\ref{prop:combined_strat}. 

\noindent \underline{Part (iii).} We can reuse the proof of Lemma~\ref{lemma:delta} to show $\Delta_n(\overline{X}_s) \ge \Delta_{n+1}(\overline{X}_s) > \kappa_f c_{n+2} > \kappa_f c_{n+1}$ and $\Delta_n(\underline{X}_s) \ge \Delta_{n+1}(\underline{X}_s) > \kappa_f c_{n+2} > \kappa_f c_{n+1}$. 
\end{proof}

Combining statements (ii) and (iii) of Lemma~\ref{lemma:empirical} guarantees the existence of a continuation region for each draw $n$. 

Because $\Delta_n(x)$ is strictly decreasing from $\overline{X}_s$ to $t\Sigma$ and reaches exactly $0$ at $t\Sigma$, there must exist a unique upper threshold $\overline{X}_f^n \in (\overline{X}_s, t\Sigma)$ such that $\Delta_n(x) > \kappa_f c_{n+1}$ (i.e., $V_n(x) > U_s(x)$) if and only if $x < \overline{X}_f^n$. 

For the negative domain, $\Delta_n(x)$ strictly decreases from $\underline{X}_s$ towards $-t\Sigma$ and remains at the constant $\Delta_n(-\infty) \ge 0$. If the cost of the next draw is sufficiently low such that $\Delta_n(-\infty) > \kappa_f c_{n+1}$, the researcher will unconditionally draw another estimate even if they hold a negatively significant result, meaning $\underline{X}_f^n = -\infty$. If $\Delta_n(-\infty) \le \kappa_f c_{n+1}$, the intersection occurs within the finite domain, yielding a unique lower threshold $\underline{X}_f^n \in [-t\Sigma, \underline{X}_s)$. 

Finally, to establish that the continuation region shrinks as $n$ increases, we use the fact that $V_n(x) > V_{n+1}(x)$ for any $x \in (\overline{X}_s, t\Sigma)$. This implies:
\[
U_s(\overline{X}_f^n) = V_n(\overline{X}_f^n) > V_{n+1}(\overline{X}_f^n)
\]
Because $\overline{X}_f^{n+1}$ is defined as the point where $U_s = V_{n+1}$, and $V_{n+1}$ crosses $U_s$ from above, it must be that $\overline{X}_f^n > \overline{X}_f^{n+1}$. Symmetrically for the negative tail, $\underline{X}_f^n \le \underline{X}_f^{n+1}$. 

This establishes the final nested property for all $n \le n^* - 1$:
\[-\infty \leq \underline{X}_f^{1} < \underline{X}_f^{2} < ... < \underline{X}_f^{n^*-1} < \underline{X}_s < 0 < \overline{X}_s < \overline{X}_f^{n^*-1} < ... < \overline{X}_f^{2} < \overline{X}_f^{1} < t\Sigma. \]

\noindent \underline{Step 2.4: Characterization of $n^*$} \ Similarly to the proofs of Proposition~\ref{prop:combined_strat}, the maximum number of draws $n^*$ is uniquely pinned down by evaluating the maximum potential benefit of a draw, which occurs at $x=0$. Thus, $n^* = \max \{n \mid \kappa_f c_n \le \Delta_{n^*-1}(0) \}$.

\subsection{Proof of Theorem~\ref{thm:identification_twosided}}
\label{pf:thm:identification_twosided}

In Section \ref{pf:thm:identification_twosided}, $\Phi(x)$ and $\phi(x)$ will refer to the CDF and PDF of the standard normal distribution respectively. We will make explicit the dependence on $\Theta$ and $\Sigma$ whenever we reference the distribution of a $N(\Theta, \Sigma^2)$ random variable.

Let $\hat{Z}=\hat{X}/\hat{\Sigma}$ be the reported Z-statistic and define
\begin{align*}
    Q =
    \begin{cases}
        2 & \mbox{ if } \hat{Z} > t_w, \\
        1 & \mbox{ if } t \leq \hat{Z} \leq t_w, \\
        0 & \mbox{ if } -t < \hat{Z} < t, \\
        -1 & \mbox{ if } -t_w \leq \hat{Z} \leq -t, \\
        -2 & \mbox{ if } \hat{Z} < -t_w.
    \end{cases}
\end{align*}

Let $[\underline{X}_f^{l},\overline{X}_f^{l}]$ be the continuation region for fast p-hacking after the $l$-th draw. These thresholds do not depend on $\Theta$, because researchers neither know $\Theta$ nor update their beliefs about $\Theta$ during the p-hacking process. Moreover, at least one of $\underline{X}_f^l$ and $\overline{X}_f^l$ is finite.

Conditional on $\Theta=b$ and $\Sigma=s$, a report with $Q=2$ is generated when the researcher reaches some draw $j+1$, all previous draws lie in the relevant continuation region, and the $(j+1)$th draw is positive and exceeds $t_w s$. Hence
\begin{equation*}
    \mathbb{P}(Q=2\mid \Theta=b,\Sigma=s)
    =
    \left(1-\Phi\left(\frac{t_w s-b}{s}\right)\right)p_{f,1}(b,s),
\end{equation*}
where
\begin{align*}
    p_{f,1}(b,s)
    &=
    \sum_{j=0}^{n^*-1}
    \mathbbm{1}\{\overline{X}_f^{j+1}<\infty\}
    \left(
        \Phi\left(\frac{\overline{X}_f^j-b}{s}\right)
        -
        \Phi\left(\frac{\underline{X}_f^j-b}{s}\right)
    \right)^j.
\end{align*}
Similarly,
\begin{equation*}
    \mathbb{P}(Q=-2\mid \Theta=b,\Sigma=s)
    =
    \Phi\left(\frac{-t_w s-b}{s}\right)p_{f,-1}(b,s),
\end{equation*}
where
\begin{align*}
    p_{f,-1}(b,s)
    &=
    \sum_{j=0}^{n^*-1}
    \mathbbm{1}\{\underline{X}_f^{j+1}>-\infty\}
    \left(
        \Phi\left(\frac{\overline{X}_f^j-b}{s}\right)
        -
        \Phi\left(\frac{\underline{X}_f^j-b}{s}\right)
    \right)^j.
\end{align*}
We define $\overline{X}_f^{n^*}:=-\infty$, $\underline{X}_f^{n^*}:=\infty$, $\overline{X}_f^0:=\infty$, and $\underline{X}_f^0:=-\infty$.

Observe that the two functions satisfy
\begin{equation}\label{equation--identification_lim_pf}
    \lim_{b\to\infty}p_{f,1}(b,s)
    =
    \lim_{b\to-\infty}p_{f,-1}(b,s)
    = 1~.
\end{equation}
Moreover,
\begin{multline*}
    \frac{\partial}{\partial b}p_{f,1}(b,s)
    =
    -\sum_{j=1}^{n^*-1}
    \frac{j}{s}
    \mathbbm{1}\{\overline{X}_f^{j+1}<\infty\}
    \left(
        \Phi\left(\frac{\overline{X}_f^j-b}{s}\right)
        -
        \Phi\left(\frac{\underline{X}_f^j-b}{s}\right)
    \right)^{j-1}
    \\
    \cdot
    \left(
        \phi\left(\frac{\overline{X}_f^j-b}{s}\right)
        -
        \phi\left(\frac{\underline{X}_f^j-b}{s}\right)
    \right),
\end{multline*}
and analogously for $p_{f,-1}$. Therefore,
\begin{equation}\label{equation--identification_derivative_pf}
    \lim_{|b|\to\infty}\frac{\partial}{\partial b}p_{f,1}(b,s)
    =
    \lim_{|b|\to\infty}\frac{\partial}{\partial b}p_{f,-1}(b,s)
    =
    0.
\end{equation}

We now identify the distribution of latent effects. Let $f_{\hat{X}}$ denote the density of published estimates. Conditional on $Q=2$, $\Sigma=s$, and $\Theta=b$,
\begin{equation*}
    f_{\hat{X}\mid\Theta,\Sigma,Q}(x\mid b,s,2)
    =
    \frac{
        \frac{1}{s}\phi\left(\frac{x-b}{s}\right)
    }{
        1-\Phi\left(\frac{t_w s-b}{s}\right)
    }
    \mathbbm{1}\{x>t_w s\}.
\end{equation*}
Meanwhile,
\begin{equation*}
    \mu_{\Theta\mid\Sigma,Q}(b\mid s,2)
    =
    \frac{1}{p_{Q,2}(s)}
    \left(1-\Phi\left(\frac{t_w s-b}{s}\right)\right)
    p_{f,1}(b,s)
    \mu_{\Theta\mid\Sigma}(b\mid s),
\end{equation*}
where
\begin{equation*}
    p_{Q,2}(s)
    =
    \int
    \left(1-\Phi\left(\frac{t_w s-b}{s}\right)\right)
    p_{f,1}(b,s)
    \mu_{\Theta\mid\Sigma}(b\mid s)
    \; db > 0~.
\end{equation*}
Thus, for $x>t_w s$,
\begin{equation*}
    f_{\hat{X}\mid\Sigma,Q}(x\mid s,2)
    =
    \frac{1}{p_{Q,2}(s)}
    \int
    \frac{1}{s}\phi\left(\frac{x-b}{s}\right)
    p_{f,1}(b,s)
    \mu_{\Theta\mid\Sigma}(b\mid s)
    \; db.
\end{equation*}
The above Gaussian convolution is real analytic in $x$, so its values on the open set $[t_ws, \infty)$ uniquely determine it on $\mathbb{R}$. Gaussian convolution is injective because the Fourier transform of the Gaussian kernel is everywhere non-zero. Hence, partial deconvolution identifies
\begin{equation*}
    g_1(b,s)
    :=
    \frac{p_{f,1}(b,s)}{p_{Q,2}(s)}
    \mu_{\Theta\mid\Sigma}(b\mid s)
    =
    \frac{p_{f,1}(b,s)}{p_{Q,2}(s)}
    \frac{1}{\sigma(s)}
    \phi\left(\frac{b-\theta(s)}{\sigma(s)}\right).
\end{equation*}
Differentiating,
\begin{equation*}
    \frac{\partial}{\partial b}g_1(b,s)
    =
    \frac{1}{p_{Q,2}(s)}
    \left[
        -p_{f,1}(b,s)\frac{b-\theta(s)}{\sigma^2(s)}
        +
        \frac{\partial}{\partial b}p_{f,1}(b,s)
    \right]
    \frac{1}{\sigma(s)}
    \phi\left(\frac{b-\theta(s)}{\sigma(s)}\right).
\end{equation*}
Define $g_{-1}(b,s)$ analogously using $p_{f,-1}(b,s)$. Next, form
\begin{equation*}
    h(b,s)
    =
    \frac{
        \partial_b g_1(b,s)+\partial_b g_{-1}(b,s)
    }{
        g_1(b,s)+g_{-1}(b,s)
    }
    =
    -\frac{b-\theta(s)}{\sigma^2(s)}
    +
    \frac{
        \partial_b p_{f,1}(b,s) \cdot p_{Q,-2}(s)+\partial_b p_{f,-1}(b,s) \cdot p_{Q,2}(s)
    }{
        p_{f,1}(b,s) \cdot p_{Q,-2}(s)+p_{f,-1}(b,s) \cdot p_{Q,2}(s)
    }.
\end{equation*}
By \eqref{equation--identification_lim_pf}, the denominator in the
second term is bounded away from zero in the tails, while
\eqref{equation--identification_derivative_pf} implies that its numerator
converges to zero. Hence
\begin{equation*}
    -\frac{1}{\sigma^2(s)}
    =
    \lim_{b\to\infty}\left[h(b+1,s)-h(b,s)\right],
\end{equation*}
and
\begin{equation*}
    \frac{2\theta(s)}{\sigma^2(s)}
    =
    \lim_{b\to\infty}\left[h(b,s)+h(-b,s)\right].
\end{equation*}
Therefore $\theta(s)$ and $\sigma^2(s)$ are identified, and so is
$\mu_{\Theta\mid\Sigma}(\cdot\mid s) = N\left(\theta(s),\sigma^2(s)\right)$. In turn, the following are identified:
\begin{align*}
	 \frac{g_1(b,s)}{\mu_{\Theta\mid\Sigma}(b\mid s)} = \frac{p_{f,1}(b,s)}{p_{Q,2}(s)} \quad \mbox{ and } \quad \frac{g_{-1}(b,s)}{\mu_{\Theta\mid\Sigma}(b\mid s)} = \frac{p_{f,-1}(b,s)}{p_{Q,-2}(s)}~.
\end{align*}
Then, \eqref{equation--identification_lim_pf} allows us to obtain
\begin{align*}
	p_{Q,2}(s)
	=
	\left(
	\lim_{b\to\infty}
	\frac{g_1(b,s)}
	{\mu_{\Theta\mid\Sigma}(b\mid s)}
	\right)^{-1} \quad \mbox{and} \quad 
	p_{Q,-2}(s)
	=
	\left(
	\lim_{b\to-\infty}
	\frac{g_{-1}(b,s)}
	{\mu_{\Theta\mid\Sigma}(b\mid s)}
	\right)^{-1}~.
\end{align*}
 Substituting them back into the previous equation then identifies $p_{f,1}(b,s)$ and $p_{f,-1}(b,s)$. Observe that
\begin{align*}
	\overline{X}_f^1 &=\infty \quad \iff \quad	\lim_{b\to-\infty}p_{f,1}(b,s)=0 \quad \mbox{, and} \\
	\underline{X}_f^1 &=-\infty \quad \iff \quad	\lim_{b\to\infty}p_{f,-1}(b,s)=0 ~.
\end{align*}
so that we also know whether $\overline{X}_f^1$ and $\underline{X}_f^1$ are finite.

Suppose WLOG that $\overline{X}^1_f < \infty$. Otherwise, $\underline{X}^1_f > -\infty$ and the following argument goes through using a symmetric argument involving $p_{f,-1}(b,s)$, $g_{-1}(b,s)$. 
The density of $\Sigma$ conditional on $Q=2$ satisfies
\begin{equation}\label{equation--sigma_cond_significance}
    \mu_{\Sigma\mid Q}(s\mid 2)
    =
    \frac{1}{\overline p_{Q,2}}
    \mu_\Sigma(s)p_{Q,2}(s),
\end{equation}
where
\begin{equation*}
    \overline p_{Q,2}
    =
    \int p_{Q,2}(s)\mu_\Sigma(s)\;ds.
\end{equation*}
Normalization gives
\begin{equation*}
    \overline p_{Q,2}
    \int
    \frac{\mu_{\Sigma\mid Q}(s\mid 2)}{p_{Q,2}(s)}
    \;ds
    =
    1,
\end{equation*}
so $\overline p_{Q,2}$ is identified. \eqref{equation--sigma_cond_significance} then identifies $\mu_\Sigma(s)$ and therefore the joint distribution $\mu$. 

We next identify the fast p-hacking thresholds. Fix $s$ and define the probability of the $j$-th continuation interval by
\begin{equation*}
    W_j(b,s)
    :=
    \Phi\left(\frac{\overline X_f^j-b}{s}\right)
    -
    \Phi\left(\frac{\underline X_f^j-b}{s}\right).
\end{equation*}
Suppose that two sets of fast-p-hacking parameters,
\begin{equation*}
    \left(
        n,\left\{\underline X_f^j,\overline X_f^j\right\}_{j=1}^{n-1}
    \right)
    \quad\text{and}\quad
    \left(
        n^\dagger,
        \left\{\underline X_f^{\dagger,j},
        \overline X_f^{\dagger,j}\right\}_{j=1}^{n^\dagger-1}
    \right),
\end{equation*}
generate the same identified functions $p_{f,1}(b,s)$ and
$p_{f,-1}(b,s)$ for every $b\in\mathbb R$.

Since $\overline X_f^1$ and $\overline X_f^{\dagger,1}$ are both finite, nesting of the continuation intervals implies that all subsequent upper
thresholds are then finite as well. Therefore
\begin{equation}
    p_{f,1}(b,s)
    =
    1+\sum_{j=1}^{n-1}W_j(b,s)^j.
    \label{eq:pf1_finite_upper}
\end{equation}
Define $W_j^\dagger(b,s)$ analogously for the second set of parameters.

We argue that the two continuation intervals coincide. Suppose that
the intervals have already been shown to agree for every $j<k$. Subtract their
common contributions and define
\begin{align*}
    R_k(b,s)
    &:=
    p_{f,1}(b,s)
    -
    1
    -
    \sum_{j=1}^{k-1}W_j(b,s)^j,\\
    R_k^\dagger(b,s)
    &:=
    p_{f,1}^\dagger(b,s)
    -
    1
    -
    \sum_{j=1}^{k-1}W_j^\dagger(b,s)^j.
\end{align*}
Observational equivalence and the induction hypothesis imply that
\begin{equation}
    R_k(b,s)=R_k^\dagger(b,s)
    \qquad\text{for all }b \in \mathbb{R }~.
    \label{eq:residual_equality}
\end{equation}

By finiteness of $\overline X_f^1$ and $\overline X_f^{\dagger,1}$, as $b\to\infty$, the interval probability satisfies
\begin{equation*}
    W_j(b,s)
    =
    \Phi\left(\frac{\overline X_f^j-b}{s}\right) \cdot (1+o(1))
    = 
    \frac{s}{b-\overline X_f^j}
    \phi\left(\frac{b-\overline X_f^j}{s}\right) \cdot
    (1+o(1))~,
\end{equation*}
where the second equality above follows from the usual asymptotic approximation of the Mill's ratio. Taking logs, we have that
\begin{equation}\label{equation--logW}
	\log W_j(b,s)
	=
	-\frac{(b-\overline X_f^j)^2}{2s^2}
	-\log\left(\frac{b-\overline X_f^j}{s}\right)
	-\frac{1}{2}\log(2\pi)
	+o(1)~.
\end{equation}
Expanding the quadratic term, we have that for every $j>k$, 
\begin{equation*}
    \log
    \left(
        \frac{W_j(b,s)^j}{W_k(b,s)^k}
    \right)
    =
    -\frac{(j-k)b^2}{2s^2}
    +O(b)~,
\end{equation*}
which converges to $-\infty$. As such, the $k$-th term dominates:
\begin{align*}
    R_k(b,s)
    &=
    W_k(b,s)^k(1+o(1))~,
    \\
    R_k^\dagger(b,s)
    &=
    W_k^\dagger(b,s)^k(1+o(1))~.
\end{align*}
In other words, \eqref{eq:residual_equality} implies that 
\begin{equation}
    \frac{W_k(b,s)^k}
         {W_k^\dagger(b,s)^k}
    \to 1
    \quad\text{as} \quad b\to\infty~.
    \label{eq:positive_tail_ratio}
\end{equation}

Suppose, toward a contradiction, that
$\overline X_f^k\neq\overline X_f^{\dagger,k}$. Then \eqref{equation--logW} gives us that
\begin{equation*}
    \log
    \left(
        \frac{W_k(b,s)}
             {W_k^\dagger(b,s)}
    \right)
    =
    \frac{
        \overline X_f^k-\overline X_f^{\dagger,k}
    }{s^2}b
    +O(1).
\end{equation*}
When the two thresholds are not equal, the right-hand side converges to either $+\infty$ or $-\infty$, contradicting
\eqref{eq:positive_tail_ratio}. As such, we must have that $\overline X_f^k = \overline X_f^{\dagger,k}$.

It remains to identify the lower threshold. First suppose that
$\underline X_f^k$ and $\underline X_f^{\dagger,k}$ are both finite. A similar argument as above gives us that
as $b\to-\infty$, 
\begin{equation*}
    \frac{W_k(b,s)^k}
         {W_k^\dagger(b,s)^k}
    \to 1
    \quad \mbox{as} \quad b\to-\infty~,
\end{equation*}
so that
\begin{equation*}
    \log
    \left(
        \frac{W_k(b,s)}
             {W_k^\dagger(b,s)}
    \right)
    =
    \frac{
        \underline X_f^k-\underline X_f^{\dagger,k}
    }{s^2}b
    +O(1)~.
\end{equation*}
For \eqref{eq:residual_equality} to hold, we must therefore have that $ \underline X_f^k = \underline X_f^{\dagger,k}$ whenever they are both finite.

We are done when $\underline X_f^k = \underline X_f^{\dagger,k} = -\infty$. It remains to rule out the case when $\underline X_f^k=-\infty$ but $\underline X_f^{\dagger,k}>-\infty$. In this case,
\begin{equation*}
    \lim_{b\to-\infty}R^\dagger_k(b,s) = 
    \lim_{b\to-\infty} \left(W_k^\dagger(b,s)\right)^k=0~,
\end{equation*}
but
\begin{equation*}
    \lim_{b\to-\infty}R_k(b,s) \geq  
    \lim_{b\to-\infty} \left(W_k(b,s)\right)^k=1~,
\end{equation*}
contradicting \eqref{eq:residual_equality}. As such,   $ \overline X_f^{\dagger,k} =  \overline X_f^{k}$ and  $ \underline X_f^{\dagger,k} =  \underline X_f^{k}$. Consequently, the common term $W_k(b,s)^k$ can be subtracted exactly from both
representations. Repeating this argument for all $k$ identifies the full set of thresholds.

Finally, suppose without loss of generality that $n<n^\dagger$. After the first
$n-1$ common terms have been subtracted, the LHS of \eqref{eq:residual_equality} is zero, whereas the RHS is
$
    \sum_{j=n}^{n^\dagger-1}
    W_j^\dagger(b,s)^j
$.
Every term in this sum is strictly positive for finite $b$ because the continuation intervals are nondegenerate. The residual therefore cannot vanish identically, contradicting \eqref{eq:residual_equality}. Hence $n=n^\dagger$. Therefore, we have that that $n^*$ and $\{\overline{X}^{j}_f, \underline{X}^{j}_f\}_{j=1}^{n^*-1}$ are identified as functions of $s$. 

We can now identify the selective-publication parameters. Let $\overline{p}_{Q,1} = \mathbb{P}(Q = 1)$. Now, 
\begin{equation*}
    \mathbb{P}\left(
        t\leq\hat Z\leq t_w
        \mid
        \hat Z\geq t,D=1
    \right)
    =
    \frac{
        \overline p_{Q,1}
    }{
        \overline p_{Q,1}+\overline p_{Q,2}
    }.
\end{equation*}
Hence $\overline p_{Q,1}$ is identified whenever $\overline p_{Q,2}>0$. The symmetric argument identifies $\overline p_{Q,-1}$ whenever $\overline p_{Q,-2}>0$. This also yields $$\overline{p}_{Q,0} = 1 - \overline p_{Q,2}-\overline p_{Q,1} - \overline p_{Q,-1}-\overline p_{Q,-2}~.$$

Recall that negative significant results are published with relative probability
$\lambda_-=1/\delta_-$. Since
\begin{equation*}
    \mathbb{P}\left(
        \hat Z\geq t
        \mid
        |\hat Z|\geq t,D=1
    \right)
    =
    \frac{
        \overline p_{Q,1}+\overline p_{Q,2}
    }{
        \overline p_{Q,1}+\overline p_{Q,2}
        +
        \lambda_-
        (\overline p_{Q,-1}+\overline p_{Q,-2})
    },
\end{equation*}
$\lambda_-$ is identified whenever $\overline p_{Q,-2}>0$.

To identify $\lambda_0=1/\delta_0$, use
\begin{align*}
    \mathbb{P}\left(|\hat Z| < t\mid D=1\right)
    =
    \frac{
        \lambda_0 \overline{p}_{Q,0}
    }{
          \overline p_{Q,1}+\overline p_{Q,2}
                +
                \lambda_-
                (\overline p_{Q,-1}+\overline p_{Q,-2}) + 
        \lambda_0 p_{Q,0}
    },
\end{align*}
As such, $\lambda_0$ is
identified whenever $\overline{p}_{Q,0} > 0$, a sufficient condition for which is $\mathbb{P}(|\hat Z|<t)>0$.

Having identified $\delta_0$ and $\delta_-$, we can identify $\mathbb{E}[\gamma(\delta_0,\delta_-)]$, the counterfactual mean under
selective publication only.

Finally, we can also identify the mass that has been p-hacked from the insignificant region into the positive significant region. Specifically, 
\begin{equation*}
	\overline p_{s,1}
	=
	\overline p_{Q,1}
	-
	\int
	\left[
	\Phi\!\left(\frac{t_ws-b}{s}\right)
	-
	\Phi\!\left(\frac{ts-b}{s}\right)
	\right]
	p_{f,1}(b,s)\,
	d\mu_{\Theta,\Sigma}(b,s)
\end{equation*}
where the second term above is the probability that $Q = 1$ as a result of fast p-hacking only. Any difference between $\overline p_{Q,1}$ and this term must therefore be due to slow p-hacking. Analogously for $\overline p_{s,-1}$.

\end{document}